\documentclass[11pt]{article}
\usepackage{cite}
\usepackage{amssymb}
\usepackage{amsfonts}
\usepackage{mathrsfs}
\usepackage{graphicx}
\usepackage{amsmath}
\usepackage{color}
\usepackage[dvipsnames]{xcolor}
\usepackage{amsthm}
\usepackage{epstopdf}
\usepackage{booktabs}
\usepackage{pifont,bm}
\usepackage{tikz}
\numberwithin{equation}{section}
\usepackage{syntonly}

\usetikzlibrary{patterns}  
\usepackage{enumerate}
\usepackage{threeparttable}
\usepackage{diagbox}
\usepackage{ulem}

\usepackage[square,sort,comma,numbers]{natbib}

\usepackage{graphicx}
\usepackage{pifont,latexsym,ifthen,amsthm,rotating,calc,textcase,booktabs}
\usepackage{amsfonts,amssymb,amsbsy,amsmath}
\usepackage{extpfeil}
\newtheorem{theorem}{Theorem}[section]
\newtheorem{lemma}[theorem]{Lemma}
\newtheorem{corollary}[theorem]{Corollary}
\newtheorem{remark}[theorem]{Remark}
\newtheorem{proposition}[theorem]{Proposition}

\usepackage{natbib}
\newtheorem{prob}{RH problem}[section]
\newtheorem{prob2}{$\bar{\partial}$-RH problem}[section]
\newtheorem{prob3}{$\bar{\partial}$-problem}[section]
\numberwithin{equation}{section}
\usepackage{verbatim}
\usepackage{todonotes}
\usepackage{float}
\usepackage{subfigure}
\usepackage{hyperref}
\usepackage{appendix}
\usepackage[percent]{overpic}
\hypersetup{
    colorlinks=true, 
    linktoc=page,     
    linkcolor=blue,  
    citecolor=blue
}

\DeclareMathOperator*{\im}{Im}
\DeclareMathOperator*{\re}{Re}

\begin{document}



\begin{titlepage}
\title{\bf{The Painlev\'{e}-type  asymptotics of the complex mKdV equation with   finite density initial data
\footnote{Corresponding author, E-mail addresses: llwen@sspu.edu.cn (Lili Wen).\protect\\
\hspace*{3ex}$^\dag$E-mail addresses: faneg@fudan.edu.cn (Engui Fan).}
}}

\author{Lili Wen$^{a,*}$, Engui Fan$^{b,\dag}$\\
\small \textit{$^{a}$Department of Mathematics, Shanghai Polytechnic University, Shanghai, 201209, P.R. China.}\\
\small \textit{$^{b}$School of Mathematical Sciences and Key Laboratory of Mathematics for Nonlinear Science,}\\
\small \textit{Fudan University, Shanghai 200433, P.R. China.}
\\
\date{}}
\thispagestyle{empty}

\end{titlepage}
\maketitle

\vspace{0.cm}
\begin{center}

\parbox{15cm}{\small
{\bf Abstract:}
 	 We  consider the Cauchy problem for the defocusing
   complex mKdV equation with finite density initial data
   \begin{align*}
&q_t+\frac{1}{2}q_{xxx}-3|q|^2q_{x}=0,\\
&q(x,0)=q_{0}(x) \sim \pm 1, \ x\to \pm\infty,
\end{align*}
 which  can be formulated into a Riemann-Hilbert (RH) problem.   With
  $\bar\partial$-generation   of the nonlinear steepest descent approach and a double scaling limit technique,
  in  the transition region $$\mathcal{D}:=\left\{(x,t)\in\mathbb{R}\times\mathbb{R}^+\big|-C< \left(x/(2t)+3/2\right) t^{2/3}<0, C\in\mathbb{R}^+\right\},$$
     we find that the long-time asymptotics   of the solution  $q(x,t)$   to the Cauchy problem
        is associated  with  the Painlev\'{e}-II transcendents.
 }

\parbox{15cm}{\small{

\vspace{0.3cm} {$\mathbf{Key words:}$}
  defocusing complex mKdV equation;   Painlev\'{e}-II transcendents; long-time asymptotics; transition region

\vspace{0.3cm} {$\mathbf{Mathematics~ Subject~ Classification:}$}
  35Q53; 35P25; 35Q15; 35C20; 35G25; 33E17
} }
 \end{center}

\tableofcontents

    \section{Introduction}
    \hspace*{\parindent}
    \begin{align}
  u(x, 0) =
  \begin{cases}
  u_+, & x \to +\infty, \\
  u_-, & x \to -\infty,
  \end{cases}
\end{align}
In this work, we focus on the Painlev\'{e} asymptotics of the solution  to the  Cauchy problem for the defocusing complex mKdV equation \cite{kcff19781,kcff1978}
\begin{align}
&q_t+\frac{1}{2}q_{xxx}-3|q|^2q_{x}=0,\ \ (x,t)\in\mathbb{R}\times\mathbb{R}^+,\label{cmkdv}\\
&q(x,0)=q_{0}(x) \sim \pm 1, \ x\to \pm\infty,\label{init1}
\end{align}
where $q$ is a complex value function  depending on $x$ and $t$. It is well-known that   the complex mKdV equation (\ref{cmkdv})
is more general  complex  representation than
the real-valued mKdV equation
 \begin{equation}
q_t+q_{xxx}-6q^2q_x=0.\label{cmkdv3}
\end{equation}

 The complex mKdV equation (\ref{cmkdv})
is   the model of lower hybrid waves in plasma \cite{kcff1978,kcff78}, transverse waves in a molecular chain \cite{gob1983}, nonlinear transverse waves in a generalized elastic solid \cite{es1989,eha1998}
and circularly polarized few-optical-cycle pulses in Kerr media \cite{lh2011,lh2012}.
The dynamics of the smooth positions and the traveling wave solutions of the   complex  mKdV equation (\ref{cmkdv})
were given in  \cite{lw2017,wam2004}.
For a Schwartz initial value $q_0 \in\mathcal{S}(\mathbb{R})$    without  presence of solitons,
 the long-time asymptotics  of the complex mKdV equation (\ref{cmkdv})  in the physical interested region $|x/(2t)|<M$   was   analyzed   by  Deift-Zhou steepest descent method in \cite{zhy2022}.

In the study on the Painlev\'e-type equation,  Ablowitz and Segur characterized a non-elementary one-parameter family of solutions of the Painlev\'{e}-II equation
via the Gel'fand-Levitan-Marchenko integral equation of the Fredholm type \cite{mja1981}. Flaschka and
Newell solved the initial value problem of Painlev\'{e}-II equation by solving an inverse problem of the corresponding ordinary differential equation \cite{hf1980}. Fokas and Ablowitz found that the Painlev\'{e}-II equation can be solved by an RH problem \cite{af1983}. Furthermore, the asymptotics of the Painlev\'{e}-II equation was studied
in a series of literatures \cite{sph1980,ari1987,hs1981,bis1987,pad1995,asf2006}. The  long-time  asymptotics for the KdV equation   in a transition  region  was first  shown in terms of Painlev\'{e} transcendents by
Segur and Ablowitz  \cite{hs1981}.
Deift and Zhou found  the connection between   the mKdV equation (\ref{cmkdv3})  and Painlev\'e equation \cite{pd1993}.
Boutet de Monvel found the Painlev\'{e}-type asymptotics of the Camassa-Holm
equation via the nonlinear steepest descent approach \cite{ab2010}.  Charlier and Lenells investigated the Airy and the Painlev\'{e} asymptotics for the mKdV equation under the zero boundary conditions \cite{cc2020}.
 Huang and Zhang further obtained the Painlev\'{e} asymptotics for the mKdV hierarchy \cite{lh2022}.
Recently, Wang and Fan found  the Painlev\'{e}-type asymptotics for the defocusing NLS equation with non-zero boundary conditions in two transition regions \cite{wzy2023}.

In our work,  we are interested in whether there is a certain connection between the complex mKdV equation (\ref{cmkdv})
and the Painlev\'e equation.  For this purpose, we consider
 finite density initial data (\ref{init1}) in weighted  Sobolev space   $q_0(x)-\tanh x\in H^{4,4}(\mathbb{R})$.
 According to the number of stationary points
appearing on the jump contour $\mathbb{R}$, we divide the $(x,t)\in\mathbb{R}\times\mathbb{R}^+$ half-plane into four categories of  space-time regions  (Figure \ref{spacetime}),
\begin{align*}
&\mathcal{D}_1:=\{(x,t)|    -\infty <\xi< -3/2\},\ \
\mathcal{D}_2:=\{(x,t)|    -3/2 <\xi\leq -1/2\} ,\\
&\mathcal{D}_3:=\{(x,t)|    -1/2 <\xi< +\infty\} , \ \  \mathcal{D}_4:=\{(x,t)|     \xi\approx -3/2\},
\end{align*}
where $\xi=x/(2t)$.
 The soliton resolution and  long-time asymptotic behavior in regions   $\mathcal{D}_j$, $j=1,2,3$
 can be obtained via the
 $\bar\partial$-steepest descent method and  a parabolic cylinder    model \cite{wf}.
The remaining problem is how to depict the asymptotics in the   region  $\mathcal{D}_4$,
in which  either   the critical lines $\re [{2i\theta(z)}]=0$ shrink to a single point
or  two   second-order stationary points   emerge   (Figure \ref{proptheta}).
This characteristic of the critical lines leads to  a new phenomenon that will arise since    $\left(1-|r(z)|^2\right)^{-1}$  blow up as $z=\pm1$.
To show the asymptotics  in the region  $\mathcal{D}_4$,  we consider a
   transition region
    \begin{align}
\mathcal{D}:=\left\{(x,t)\in\mathbb{R}\times\mathbb{R}^+\big|-C< \left(\xi+3/2\right) t^{2/3}<0, C\in\mathbb{R}^+\right\}.\label{D}
\end{align}
We noticed that the local RH problems  in  $\mathcal{D}$ can be
 matched with a solvable RH model associated with the Painlev\'{e}-II function.
The key to dealing with singularity caused by $|r(\pm1)|=1$ is inspired by the method used by Cuccagna and Jenkins \cite{cs2016}.

Before starting our main work,  some notations need to be briefly introduced as follows:
\begin{itemize}
\item[$\diamond$] Three kinds of the  Pauli matrixes
\begin{equation}
\sigma_{1}=\left(
\begin{array}{cc}
0&1\\
1&0
\end{array}
\right),\ \
 \sigma_{2}=\left(
\begin{array}{cc}
0&-i\\
i&0
\end{array}
\right),\ \ \sigma_{3}=\left(
\begin{array}{cc}
1&0\\
0&-1
\end{array}
\right).\nonumber
\end{equation}
\item[$\diamond$] $a\lesssim b$ denotes $a\leq cb$ for a constant $c>0$.
\item[$\diamond$] A weighted space $L^{p,s}(\mathbb{R})$ defined by
\begin{equation}
L^{p,s}(\mathbb{R})=\{q\in L^{p}(\mathbb{R})\mid\langle \cdot \rangle^{s}q \in L^{p}(\mathbb{R})\},\nonumber
\end{equation}
with
the norm  $\|q\|_{L^{p,s}(\mathbb{R})}=\|\langle \cdot \rangle^{s}q \|_{L^{p}(\mathbb{R})}$ where  $\langle x\rangle =\left(1+x^{2}\right)^{-1/2}$.
\item[$\diamond$] $H^{m}(\mathbb{R})$ defined by
\begin{align*}
\|q\|_{H^{m}(\mathbb{R})}=\|\langle \cdot \rangle^{m}\widehat{q}\|_{L^{2}(\mathbb{R})},
\end{align*}
where $\widehat{q}$ denotes the Fourier transform for $q$.
\item[$\diamond$] A Sobolev space $W^{m,p}(\mathbb{R})$ defined by
\begin{equation}
W^{m,p}(\mathbb{R})=\{q\in L^{p}(\mathbb{R})\mid\partial^{j}q \in L^{p}(\mathbb{R})\},\ \ j=0,1,2,\cdots,m,\nonumber
\end{equation}
where the norm defined by $\|q\|_{W^{m,p}(\mathbb{R})}=\sum\limits_{j=0}^{m}\|\partial^{j}q\|_{L^{p}(\mathbb{R})}$.
\item[$\diamond$] A weighted Sobolev space defined by
\begin{equation}
H^{m,s}(\mathbb{R})=L^{2,s}(\mathbb{R})\cap H^{m}(\mathbb{R}).\nonumber
\end{equation}
\end{itemize}

Our  main result is now  stated   as follows.
\begin{theorem} \label{th}
For the initial value  $q_0-\tanh x \in H^{4,4}(\mathbb{R})$,  the associated
	reflection coefficient  and  the discrete spectrums  are    $\{r(z), \nu_n\}_{n=1}^{N} $.
	The long-time asymptotics    of the   solution  $q(x,t)$  to the Cauchy problem (\ref{cmkdv})-(\ref{init1}) for the
	defocusing  complex mKdV  equation in
	a transition region   (\ref{D}) is recovered by
\begin{align}
q(x,t)
 =e^{\alpha(\infty)}\left[-1-i\tau^{-1/3}\left(\beta_1+\beta_2\right)\right]+ \mathcal{O}(t^{-1/3-\varsigma}),\label{q1}
\end{align}
     where
\begin{align*}
&\beta_1=\frac{1}{2}\left[\hat{u}(s)e^{2i\Theta(\arg\tilde{r}(1))}-\int_{s}^{\infty}\hat{u}^2(\zeta)\mathrm{d}\zeta\right],\ \
\beta_2=\frac{1}{2}\left[\check{u}(s)e^{2i\Theta(\arg\tilde{r}(-1))}+\int_{s}^{\infty}\check{u}^2(\zeta)\mathrm{d}\zeta\right],\\
	&\alpha(\infty) = 2 \sum_{n=1}^{N} \log\bar{ \nu}_n +\frac{1}{2i\pi} \int_\Gamma\frac{v(\zeta)}{\zeta} \, \mathrm{d}\zeta,\ \  s  = 2 \left(9/4\right)^{ \frac{1}{3}}  \left(\xi +3/2\right) t^{\frac{2}{3}},
\end{align*}
$v(\zeta)$ and $\Gamma$  defined by (\ref{nu}) and (\ref{gamma}), respectively.
In addition, $\Theta(\cdot)=\frac{i(\cdot+\pi/2)}{2}$
and $\tau$ defined by (\ref{hatk}).
The real function $u(s ) $ is a solution of the Painlev\'{e}-II equation (\ref{p2})
with the asymptotics
\begin{align*}
	&\hat{u}(s) =-|\tilde{r}(1)| \mathrm{Ai}(s) +\mathcal{O}\left(e^{-(4/3)s^{3/2}}s^{-1/4}\right),\quad s\to +\infty,\\
&\check{u}(s) =-|\tilde{r}(-1)| \mathrm{Ai}(s) +\mathcal{O}\left(e^{-(4/3)s^{3/2}}s^{-1/4}\right),\quad s\to +\infty,
\end{align*}
where $\mathrm{Ai}(s)$ denotes the Airy function.
\end{theorem}



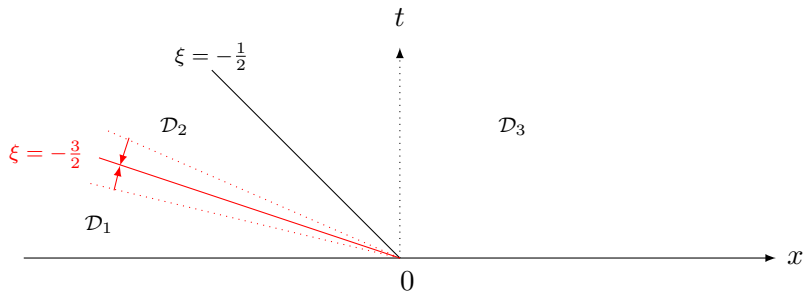
\begin{figure}[H]
	\begin{center}
		\begin{tikzpicture}

		\draw [-latex ](-5,0)--(5,0);
		\draw [dotted,-latex ](0,0)--(0,2.8);
		\node    at (0.1,-0.3)  {$0$};
		\node    at (5.26,0)  { $x$};
		\node    at (0,3.2)  { $t$};
		 \node  [below]  at (-2.5,3) {\scriptsize $\xi=-\frac{1}{2}$};
		 \node  [red,below]  at (-4.7,1.7) {\scriptsize $\xi=-\frac{3}{2}$};
\node  [below]  at (1.5,2) {\scriptsize $\mathcal{D}_3$};
\node  [below]  at (-3,2) {\scriptsize $\mathcal{D}_2$};
\node  [below]  at (-4,0.7) {\scriptsize $\mathcal{D}_1$};

  		\draw [red](0,0)--(-4.0,4/3);
		\draw [](0,0)--(-2.5,2.5);
\draw [dotted,red](0,0)--(-4.15,1.0);
\draw [dotted,red](0,0)--(-3.9,1.7);

\draw[-latex] [red] (-3.6,1.6)--(-3.72,1.23);
\draw[-latex] [red] (-3.8,0.9)--(-3.72,1.23);
	  		\end{tikzpicture}
	\end{center}
	\caption{\footnotesize The space-time regions of $x$ and $t$, the red solid line is $\xi=-3/2$. The two red dashed lines which belong to the regions $\mathcal{D}_1$ and $\mathcal{D}_2$, respectively, approach to the red solid line, these can be expressed as  $\xi\approx-3/2$. The region passed through by the dashed lines  is the transition region.  }
	\label{spacetime}
\end{figure}

The crucial technology to prove  Theorem \ref{th}  is the  application of the $\bar{\partial}$-steepest descent method to  the RH formulation
of the Cauchy problem (\ref{cmkdv})-(\ref{init1}).  This  method was  introduced by McLaughlin-Miller \cite{mkt2006,mkt2008} and Dieng-McLaughlin \cite{dm08}, has been extensively   applied to analyze long-time asymptotics and soliton resolution for  integrable systems \cite{bm2018,jr2018,ljq2018,yyl2022,yyl2023}.
The significant advantage of this method is that we only need  the lower  regularity initial data
  and   simplify  tedious  estimates of the standard steepest descent method.
   We outline  the steps  to  prove   Theorem \ref{th} as follows.



In  Section  \ref{sec1},  we focus on the direct  scattering transform of the Cauchy problem  (\ref{cmkdv})-(\ref{init1}).
The properties of Jost solutions and scattering data  are analyzed.

In Section  \ref{insrh},    we     construct an    RH problem \ref{RHP0}    associated with    the Cauchy problem  (\ref{cmkdv})-(\ref{init1})
via  the   inverse scattering.

  In Section \ref{sec3},  we solve the  RH problem \ref{RHP0} and show   Painlev\'e-type asymptotics in the transition region $\mathcal{D}$.
   With  a series of deformations and $\bar{\partial}$-steepest analysis,  the RH problem \ref{RHP0}
 is changed into a  hybrid  $\bar{\partial}$-RH problem \ref{ms3},
  which  can be  solved by a pure RH problem \ref{mrhp}  and a pure $\bar{\partial}$-problem \ref{trad}.
Moreover, we  show that  the pure RH problem  \ref{mrhp}   matched by a solvable Painlev\'{e}-II model
 near the critical points $z=\pm1$ in the transition region $\mathcal{D}$.
  The residual  error  comes from   the pure $\bar{\partial}$-problem \ref{trad}.
Finally, we complete the proof of  Theorem \ref{th}  by inverting  the previous transformations.

\section{Direct scattering transform} \label{sec1}
\hspace*{\parindent}
In this section, we state the direct  scattering transform of the defocusing complex mKdV equation (\ref{cmkdv}) with the initial data (\ref{init1}).
\subsection{Jost functions }
\hspace*{\parindent}
The defocusing complex mKdV equation (\ref{cmkdv}) admits the Lax pair representation \cite{ly1999}
\begin{equation}
  	  	\phi_x(k,x,t)=L \phi(k,x,t), \quad\quad\quad \phi_t(k,x,t)=N \phi(k,x,t), \label{lax}
  \end{equation}
where the $2\times2$ matrix functions $L:=L(k,x,t)$ and $N:=N(k,x,t)$ are given by
    \begin{align}
  L=-ik\sigma_{3}+Q,\ \ \ \ N=-2ik^3\sigma_{3}+2k^2Q-ikQ^2\sigma_{3}+ik\sigma_{3}Q_{x}+N_0, \label{laxt}
  \end{align}
with
$$
N_0=\frac{1}{2}[Q,Q_{x}]-\frac{1}{2}Q_{xx}+Q^{3},
$$
the bracket denotes a commutation relation
and $k\in\mathbb{C}$ is the spectrum parameter.  The $2\times2$ matrix $Q:=Q(x,t)$ is given by
	\begin{align*}
		Q=\left(\begin{array}{cc}
			0 & q(x,t)\\
			\overline{q}(x,t) & 0
		\end{array}\right),
	\end{align*}
the overbar denotes the Schwartz conjugate ($\overline{q}(x, t) =\overline{ q(\overline{x}, t)}$).
Under the initial data (\ref{init1}), matrix functions $L$ and $T$ admit the limits
\begin{equation}
L_{\pm}:=\lim_{x\to\pm\infty}L,\ \ \ \ N_{\pm}:=\lim_{x\to\pm\infty}N=(2k^2+1)L_\pm,\label{lt}
\end{equation}
where the matrix $Q$ is replaced by
$$Q_\pm:=\lim\limits_{x\to\pm\infty}Q=\left(\begin{array}{cc}
0&\pm1\\
\pm1&0
\end{array}
\right).$$
$L_{\pm}$ has the eigenvalues $\pm i\lambda$, where $\lambda=\sqrt{k^{2}-1}$. The second equation in (\ref{lt}) implies that $N_{\pm}$ has eigenvalues $\pm (2k^2+1)i\lambda$.
To eliminate the multi-valued effect of $\lambda$, we introduce an affine parameter $z:=k+\lambda$ and  derive two single-valued functions
\begin{equation}
k:=k(z)=\frac{1}{2}\left(z+\frac{1}{z}\right),\ \ \ \  \lambda:=\lambda(z)=\frac{1}{2}\left(z-\frac{1}{z}\right).\label{klamz}
\end{equation}
The commutation relation $[L_\pm,N_\pm]=0$ indicates that $L_\pm$ and $N_\pm$ enjoy a common eigenvector
\begin{align}
    &Y_\pm(z)=
    I\pm\sigma_2z^{-1}.\label{ypm}
     \end{align}
where $I$ is the $2\times2$ identity matrix. Moreover, we notice that  $\text{det}Y_\pm(z)=1-z^{-2}$. Specifically, $\det{Y_{\pm}(z)}\mid_{z=\pm1}=0$ which imply that $Y_{\pm}(z)$ are non-invertible matrix as $z=\pm1$.

The  ordinary differential equations (\ref{lax}) admits Jost solutions $\phi^\pm(z):=\phi^\pm(z;x,t)$ which  have the  asymptotic behavior
	\begin{equation*}
		\phi^\pm(z) \sim Y_\pm(z) e^{-it\theta(z)\sigma_3},  \quad x \to \pm \infty,
	\end{equation*}
	where $\theta(z)=\lambda \big[x/t+(2k^2+1)\big].$ Making  a transformation
	\begin{equation*}
		\mu^\pm(z )=\phi^\pm(z )e^{it\theta(z)\sigma_3},
	\end{equation*}
and $\mu^{\pm}(z):=\mu^{\pm}(z,x,t)$ satisfy the Cauchy system
\begin{align*}
&\left(Y_{\pm}^{-1}\mu^{\pm}\right)_{x}=-i\lambda[\sigma_{3},Y_{\pm}^{-1}\mu]+Y_{\pm}^{-1}\Delta Q_{\pm}\mu^{\pm},\\
&\left(Y_{\pm}^{-1}\mu^{\pm}\right)_{t}=2ik\lambda[\sigma_{3},Y_{\pm}^{-1}\mu]+Y_{\pm}^{-1}\Delta \tilde{Q}_{\pm}\mu^{\pm},
\end{align*}
where $\Delta Q_\pm:=L-L_{\pm}$ and $\Delta \tilde{Q}_{\pm}:=N-N_{\pm}$.
The functions  $\mu^\pm(z ) :=\left(\mu_{1}^{\pm}(z), \mu_{2}^{\pm}(z)\right)$ satisfy  the Volterra integral equations
\begin{align}
\mu^\pm(z)=Y_\pm(z) + \begin{cases}
\int_{\pm \infty}^{x} \left( Y_\pm e^{-i \lambda(z) (x-y)\widehat\sigma_3}  Y^{-1}_\pm    \right) \left( \Delta Q_\pm(z,y) \mu^\pm(z;y)      \right)  \mathrm{d}y, \ z\not= \pm 1,\\
\int_{\pm \infty}^{x} \left( I + (x-y)L_\pm      \right) \Delta Q_\pm (z;y) \mu^\pm(z;y)  \mathrm{d}y, \ z = \pm 1,
\end{cases}\label{mupm}
\end{align}
where $\widehat{\sigma}_{3}$ map a matrix $A$ as $e^{\widehat{\sigma}_{3}}=e^{\sigma_{3}}Ae^{-\sigma_{3}}$.
The existence, analyticity  and differentiation  of $\mu^\pm(z )$ can be proven directly,     here we just sketch their  properties.

   \begin{lemma}\label{lemm21}
Let $q_0(x)-\tanh x \in L^{1,2}(\mathbb{R})$ and $q'(x)\in W^{1,1}(\mathbb{R})$,
    	$\mu^{\pm}(z)$ admit the following properties:

\begin{itemize}

 \item    $\mu_1^+(z )$ and $\mu_2^-(z )$  can be analytically extended to $z \in \mathbb{C}^-$ and continuously extended to $\bar{\mathbb{C}}^-$.  $\mu_1^-(z )$ and $\mu_2^+(z )$  can be analytically extended to $z \in \mathbb{C}^+$ and continuously extended to $\bar{\mathbb{C}}^+$, where $\bar{\mathbb{C}}^\pm:=\{\mathbb{C}^\pm\cup\mathbb{R}\}\backslash\{0,\pm1\}$ (see Figure \ref{cpm}).

    \item  $\mu^\pm(z )$ admit symmetries
  \begin{equation}
    \mu^{\pm}(z;x,t)=
    \sigma_{1}\overline{\mu^{\pm}(\overline{z};x,t)}\sigma_{1}, \quad
    \mu^{\pm}(z,x,t)=\pm z^{-1}\mu^{\pm}(z^{-1};x,t)\sigma_{2}.\label{musym}
    \end{equation}



  \item     $\mu^\pm_1(z )$ and $\mu^\pm_2(z )$ have  asymptotic properties as $z\to\infty$ and $z\to0$
    	\begin{align*}
    	&\mu^\pm_1(z ) = e_1 + \frac{ \mu_1^{\pm}[1]}{z}+\mathcal{O}(z^{-2}), \ \  \mu^\pm_2(z ) = e_2 +   \frac{\mu_2^{\pm}[1]}{z}+\mathcal{O} (z^{-2}), \ \ z \to \infty,\nonumber\\
    	&\mu^\pm_1(z ) = \mp \frac{i}{z}e_2 +\mathcal{O}(1), \ \
    \mu^\pm_2(z )  = \pm \frac{i}{z}e_1 +\mathcal{O}(1),  \  \ z\to 0,
    	\end{align*}
     where
     \begin{align*}
     \mu_1^{\pm}[1]:=\left(
     -i\int_{\pm\infty}^{x}(|q|^2+1)\mathrm{d}x,\ -i\overline{q}
         \right)^T,\ \  \mu_2^{\pm}[1]:=\left(
            iq,\ i\int_{\pm\infty}^{x}(|q|^2+1)\mathrm{d}x
     \right)^T,
     \end{align*}
     with $e_{1}:=(1,0)^{T}$ and $e_{2}:=(0,1)^{T}$. The superscript ``T" denotes the transpose of a matrix.
     \end{itemize}
    \end{lemma}

\begin{lemma}
Let $q_0(x)-\tanh x\in L^{1,2}(\mathbb{R})$ and $q'(x)\in W^{1,1}(\mathbb{R})$,
 we have
the map $q(x,t)\mapsto\frac{\partial^n}{\partial z^n}\mu_{\pm,j}(z) (j=1,2,~ and ~n>0)$ are Lipschitz continuous, particularly for $\forall x_0\in\mathbb{R}$, $\mu_{\pm,j}(z)$ are continuous differentiable mappings:
\begin{align*}
&\partial_z^n\mu_{+,1}:\bar{\mathbb{C}}^-\mapsto L_{loc}^{\infty}\{\bar{\mathbb{C}}^-,C^1([x_0,\infty),\mathbb{C}^2)\cap W^{1,\infty}([x_0,\infty),\mathbb{C}^2)\},\\
&\partial_z^n\mu_{-,2}:\bar{\mathbb{C}}^-\mapsto L_{loc}^{\infty}\{\bar{\mathbb{C}}^-,C^1((-\infty,x_0],\mathbb{C}^2)\cap W^{1,\infty}((-\infty,x_0],\mathbb{C}^2)\},\\
&\partial_z^n\mu_{+,2}:\bar{\mathbb{C}}^+\mapsto L_{loc}^{\infty}\{\bar{\mathbb{C}}^+,C^1([x_0,\infty),\mathbb{C}^2)\cap W^{1,\infty}([x_0,\infty),\mathbb{C}^2)\},\\
&\partial_z^n\mu_{-,1}:\bar{\mathbb{C}}^+\mapsto L_{loc}^{\infty}\{\bar{\mathbb{C}}^+,C^1((-\infty,x_0],\mathbb{C}^2)\cap W^{1,\infty}((-\infty,x_0],\mathbb{C}^2)\},
\end{align*}
 where $\bar{\mathbb{C}}^\pm:=\{\mathbb{C}^\pm\cup\mathbb{R}\}\backslash\{0,\pm1\}$.
\end{lemma}

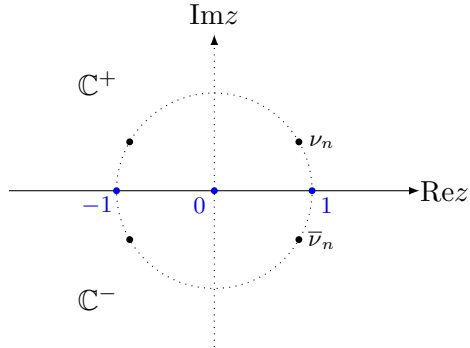
\begin{figure}[H]
\begin{center}
\begin{tikzpicture}[scale=0.65]
\path[ ] (-4,0)rectangle (4,3.0);
\draw[-latex] (-4.2,0)--(4.2,0);
\draw[ ](4,0)node[right] {$\mathrm{Re}z$};
\draw[dotted,-latex](0,-3.2)--(0,3.2);
\draw[ ](0,3.2)node[above]{$\mathrm{Im}z$};
\draw [dotted] (0,0) circle [radius=2.0];
\draw [fill] (-3,2.2) circle [radius=0] node[right] {$\mathbb{C}^+$};
\draw [fill] (-3,-2.2) circle [radius=0] node[right] {$\mathbb{C}^-$};
       \coordinate (A) at (1.73,  1);
		\coordinate (B) at (1.73,  -1);
		\coordinate (C) at (1,  1.73);
		\coordinate (D) at (1,  -1.73);
		\coordinate (E) at (-1.73,  1);
		\coordinate (F) at (-1.73,  -1);
			\coordinate (I) at (0.5,  1.93);
		\coordinate (J) at (0.5,  -1.93);
		\coordinate (K) at (0,  0);
		\coordinate (L) at (2,  0);
		\coordinate (M) at (-2,  0);
\fill[] (A) circle (2pt);
\fill[] (B) circle (2pt);
\fill[] (E) circle (2pt);
\fill[] (F) circle (2pt);
\fill[blue] (K)  circle (2pt);
\fill[blue] (L) circle (2pt);
\fill[blue] (M) circle (2pt);
\node    at (2.2,1)  {\footnotesize $\nu_{n}$};
\node    at (2.2, -1)  {\footnotesize $\overline{\nu}_{n}$};
\node [blue]   at (-0.3,-0.3)  {\footnotesize $0$};
\node  [blue]  at (2.3, -0.3)  {\footnotesize $1$};
\node  [blue]  at (-2.4, -0.3)  {\footnotesize $-1$};
\end{tikzpicture}
\end{center}
\caption{  The  discrete spectrum is restricted to the unit circle $\mathcal{G}: \{\nu_n\in\mathbb{C}^+, \overline{\nu}_n\in\mathbb{C}^-\mid|\nu_n|=|\overline{\nu}_n|=1\}$. The real axis $\mathbb{R}$ is the jump contour of the RH problem \ref{RHP0}.}\label{cpm}
\end{figure}

\subsection{Scattering data  }
\hspace*{\parindent}
The matrix functions $\phi^\pm(z)$ admit the linear relation
 \begin{equation}
    \phi^-(z )=\phi^+(z )S(z),\quad z\in\mathbb{R}\backslash\{\pm1,0\},\label{phis}
    \end{equation}
where $S(z):=(s_{ij}(z))_{2\times2}$ is defined as  the   scattering  matrix.
Scattering data $s_{ij}(z)$ can be described by the Jost functions
       	    \begin{align}
       	    	&s_{11}(z)=\frac{\det  \left(\phi_1^-(z ), \phi_2^+(z )\right) }{1-z^{-2}},\quad    	s_{12}(z)=\frac{\det  \left(\phi_2^-(z ), \phi_2^+(z )\right)  }{1-z^{-2}},\label{s12}\\
       &s_{21}(z)=\frac{\det  \left(\phi_1^+(z ), \phi_1^-(z )\right)  }{1-z^{-2}},\quad    	s_{22}(z)=\frac{\det  \left(\phi_1^+(z ), \phi_2^-(z )\right)  }{1-z^{-2}}.\label{s34}
       	    \end{align}
       We define a reflection coefficient as
       \begin{equation}
       r(z) := \frac{s_{21}(z)}{s_{11}(z)}.\label{2rz}
       \end{equation}
       The properties of the scattering data $s_{ij}(z)$ and $r(z)$ given by the following Lemma \ref{sprop}.

 \begin{lemma}\label{sprop}
Let  $q_0(x)-\tanh x\in L^{1,2} (\mathbb{R}), \ \ q'\in W^{1,1}(\mathbb{R})$, we have some properties of the scattering data shown as follows.
\begin{itemize}
 \item The scattering matrix  $S(z)$ satisfies the symmetries $ {S( {z})}  = \sigma_1\overline{S (\bar z) }\sigma_1 =  -\sigma_{2}{S ( z^{-1} )}\sigma_2,$ and we have
    \begin{align*}
    s_{11}(z)=\overline{s_{22}(\overline{z})}=-s_{22}(z^{-1}),\ \ s_{12}(z)=\overline{s_{21}(\overline{z})}=s_{21}(z^{-1}).
    \end{align*}

\item     $s_{11}(z)$ can be analytically extended to $z \in \mathbb{C}^+$ and  has no singularity on the contour  $\mathbb{R}\backslash\{\pm1,0\}$.
 Zeros of $s_{11}(z)$ in  $\mathbb{C}^+ $ are simple, finite and distribute on the unit circle $|z|=1$. $s_{22}(z)$ has similar properties applying symmetries.

 \item The scattering data $s_{11}(z)$ with the asymptotics
 \begin{align}
 \lim_{z\rightarrow\infty}(s_{11}-1)z=i\int_\mathbb{R}(|q|^2+1)\mathrm{d}x,\ \ \lim_{z\rightarrow0}(s_{11}-1)z^{-1}=i\int_\mathbb{R}(|q|^2+1)\mathrm{d}x,\label{s11as}
 \end{align}
 and $s_{21}(z)$ admits
\begin{align}
|s_{21}(z)|=\mathcal{O}(|z|^{-2}),\ \ |z|\rightarrow \infty,\ \ |s_{21}(z)|=\mathcal{O}(|z|^2),\ \ |z|\rightarrow0,\label{s21as}
\end{align}
so we have
\begin{align}
r(z)\sim z^{-2},\ \ z\rightarrow \infty, \ \  r(z)\sim 0,\ \ z\rightarrow 0.\label{rasy}
\end{align}

\item  For $z\in\mathbb{R}\backslash\{\pm1,0\}$, from (\ref{phis}), we have
\begin{equation}
\det{S(z)}=|s_{11}(z)|^2-|s_{21}(z)|^2=1\Longrightarrow |r(z)|<1.\label{r<1}
\end{equation}

 \item {Generic:}   Scattering datas $s_{11}(z)$ and $s_{21}(z)$ with the same simple pole points $z=\pm1$,
 \begin{align*}
s_{11}(\pm1)=\frac{d_{\pm}}{z\mp1}+\mathcal{O}(1),\ \ s_{21}(\pm1)=-i\frac{d_{\pm}}{z\mp1}+\mathcal{O}(1),
 \end{align*}
 where
 \begin{align}
 d_{\pm}:=\pm\frac{1}{2}\det\left(\phi_1^-(\pm1 ), \phi_2^+(\pm1 )\right).\label{dpm}
 \end{align}
Furthermore, we have
 \begin{equation}
r(\pm1)=\mp i\Longrightarrow|r(\pm1)|=1.\label{rpm1}
 \end{equation}

\item {Non-Generic:} Scattering datas $s_{11}(z)$ and $s_{21}(z)$ are continuous as $z=\pm1$ and $|r(\pm1)|<1$.

\end{itemize}

\end{lemma}

The initial value $q_{0}(x)$ (\ref{init1}) with sufficient smoothness and decay
properties, the reflection coefficient $r(z)$ will also be smooth and decaying, which will be shown in the  lemma \ref{r11}.
 \begin{lemma}\label{r11}
Let $q_0(x)-\tanh x\in L^{1,2}(\mathbb{R})$, $q'\in W^{1,1}(\mathbb{R})$, we have
$r(z)\in H^{1 }(\mathbb{R})$. Moreover, $r(z)\in H^{1,1}(\mathbb{R})$.
\end{lemma}
\begin{proof}
$\uuline{Step~1}:$
Lemma \ref{lemm21} and (\ref{s12})-(\ref{s34}) imply that $s_{11}(z)$ and $s_{21}$ are continuous as $z\in\mathbb{R}\backslash\{\pm1,0\}$, $r(z)$ is also. (\ref{rasy}) and (\ref{rpm1}) indicate that $r(z)$ is bounded in the small neighborhoods of $\{\pm1,0\}$ and $r(z)\in L^1(\mathbb{R})\cap L^2(\mathbb{R})$. So we just need to prove that $r'(z)\in L^2(\mathbb{R})$. For $\delta_0>0 $ is sufficiently small, from the maps
\begin{align*}
q\mapsto\det  \left(\phi_1^-(z ), \phi_2^+(z )\right),\ \ q\mapsto\det  \left(\phi_1^+(z ), \phi_1^-(z )\right)
\end{align*}
are locally Lipschitz maps
\begin{align}
\{q:q'\in W^{1,1}(\mathbb{R}), q\in L^{1,n+1}(\mathbb{R})\}\mapsto W^{n,\infty}(\mathbb{R}\backslash(-\delta_0,\delta_0)), \ \ n\geq0.\label{Lips}
\end{align}
In fact, $q\mapsto\phi^-_1(z,0)$ is a locally Lipschitz map with values in $W^{n,\infty}(\bar{\mathbb{C}}^+\backslash D(0,\delta_0),\mathbb{C}^2)$. For $q\mapsto\phi^+_1(z,0)$ and $q\mapsto\phi^+_2(z,0)$ are also. Collecting (\ref{s11as}) and (\ref{s21as}), we derive $q\mapsto r(z)$ is locally Lipschitz map from the domain in (\ref{Lips}) into
\begin{align*}
W^{n,\infty}(I_{\delta_0})\cap H^{n}(I_{\delta_0}),
\end{align*}
where $I_{\delta_0}:=\mathbb{R}\backslash\{(-\delta_0,\delta_0)\cup(1-\delta_0,1+\delta_0)\cup(-1-\delta_0,-1+\delta_0)\}$. We fix  $\delta_0>0$ is sufficiently small so that the three disks $|z\mp1|\leq\delta_0$ and $|z|\leq\delta_0$ have no intersection. In the
complement of their union
\begin{align*}
|\partial_z^jr(z)|\leq C_{\delta_0}\langle z\rangle^{-1}, \ \ j=0,1.
\end{align*}
For the discussion above, the other Jost functions with the analogous results. Next, we will prove the boundedness of $r(z)$ in a small neighborhood  of $z=1$.

Let $|z-1|<\delta_0$, we recollect $d_\pm$
\begin{align*}
r(z)=\frac{s_{21}(z)}{s_{11}(z)}=\frac{\det  \left(\phi_1^+(z ), \phi_1^-(z )\right)}{\det  \left(\phi_1^-(z ), \phi_2^+(z )\right)}=\frac{\int_{1}^z\lhd(s)\mathrm{d}s-id_+}{\int_{1}^z\rhd(s)\mathrm{d}s+id_+},
\end{align*}
where
\begin{align*}
\lhd(s):=\partial_s\det  \left(\phi_1^+(s ), \phi_1^-(s )\right), \ \ \partial_s\det  \left(\phi_1^-(s ), \phi_2^+(s )\right):=\rhd(s).
\end{align*}
If $d_+\neq0$, $r'(z)$ exist and is bounded near $z=1$.
If $d_+=0$, $z=1$ is not the pole of $s_{11}(z)$ and $s_{21}(z)$, they are continuous as $z=1$, then
\begin{align*}
r(z)=\frac{\int_{1}^z\lhd(s)\mathrm{d}s}{\int_{1}^z\rhd(s)\mathrm{d}s}.
\end{align*}
From (\ref{s12}) we have
\begin{align}
&(z^2-1)s_{11}(z)=z^2\det  \left(\phi_1^-(z ), \phi_2^+(z )\right)\label{s11z}\\
\xRightarrow{d_+=0}&\det  \left(\phi_1^-(1 ), \phi_2^+(1 )\right)=0.
\end{align}
Differentiating (\ref{s11z}) at $z=1$, we have
\begin{align*}
2s_{11}(1)=\partial_z\det  \left(\phi_1^-(1 ), \phi_2^+(1 )\right)\mid_{z=1}:=\rhd(1).
\end{align*}
For the reason that $|s_{11}(1)|^2=1+|s_{21}(1)|^2>1$, we have $\rhd(1)\neq0$. It follows that $r'(z)$ is bounded near $z=1$.

For $z=-1$, we have the similar proof process. For $z=0$, we use the symmetry $r(z)=-\overline{r(z^{-1})}$ to infer that $r(z)$ vanished. It  follows that $r'(z)\in L^2(\mathbb{R})$.

$\uuline{Step~2}:$ For the reason that $H^{1,1}(\mathbb{R})=L^{2,1}(\mathbb{R})\cap H^1(\mathbb{R})$,  we just need to prove $r(z)\in L^{2,1}(\mathbb{R})$. From (\ref{rasy}), we have
\begin{align*}
|z|^2r^2(z)\sim|z|^{-2},\ \ |z|\to\infty,
\end{align*}
which lead to
\begin{align*}
\int_{\mathbb{R}}\left|\langle z\rangle r(z)\right|^2<\infty.
\end{align*}
This prove the result.
\end{proof}

\subsection{Discrete spectrum}\label{disspe}
\hspace*{\parindent}
From the symmetries of $s_{ij}(z)$, one can derive that $s_{11}(z)$ with the zero points $\{\nu_{n}\}_{n=1}^{N}$ and  $s_{22}(z)$ with the zero points $\{\overline{\nu}_{n}\}_{n=1}^{N}$.
We define a notation $\mathcal{N}:=\{1,\cdots,N\}$. Each zero $\nu_n$ of $s_{11}(z)$ with the properties which shown as the Lemma \ref{zeros}.
\begin{lemma}\label{zeros}
Let $q_0(x)-\tanh x\in L^{1,2}(\mathbb{R})$, the discrete spectrum $\forall\nu_n\in\mathcal{Z}^+$ has the following properties
\begin{itemize}
\item   $\nu_n$ lies on the unit circle and $\{\nu_n\in\mathbb{C}^+|\nu_n=e^{i\iota},\iota\in(0,\pi)\}$,
\item $\nu_n$ is a simple zero point of $s_{11}(z)$, i.e., $s'_{11}(z)=0$, $s''_{11}(z)\neq0$,
\item $\nu_n$ cannot appear on the real axis $\mathbb{R}$, i.e., $\nu_n\neq\pm1$,
\end{itemize}
where $n\in\mathcal{N}$.  The discrete spectrum  $\forall\overline{\nu}_n\in\mathcal{Z}^-$ is the same true.
\end{lemma}
\begin{proof}
$\uuline{Step~1}:$
Rewrite the $x$-part in (\ref{lax}) as an operator $\mathcal{A}:L^2\rightarrow L^2$, i.e.,
\begin{equation*}
\mathcal{A}:=\mathcal{A}(\phi)=\left(i\sigma_{3}\partial_{x}+Q\right)\phi=k\phi.
\end{equation*}
The operator $\mathcal{A}$ is a self-adjoint operator, which implies the discrete spectrum $k\in\mathbb{R}$. In addition, the imaginary part of $k$ is expressed as
\begin{equation}
\im{k}=\im{\frac{1}{2}\frac{(|z|^2-1)z+(z+\overline{z})}{|z|^2}}=\frac{1}{2}\left(1-\frac{1}{|z|^2}\right)\im z=0.\label{imzk}
\end{equation}
This equation (\ref{imzk}) implies that the discrete spectrums in the $z$-plane only locate on a circle $|z|=1$, i.e., $|\nu_n|=1$.

$\uuline{Step~2}:$ From (\ref{s12}), we infer that there exist $\kappa_n$ satisfy
\begin{align}
\phi^-_1(\nu_n,x)=\kappa_n\phi^+_2(\nu_n,x)&\xLeftrightarrow{\nu_n\mapsto\bar{\nu}_n}\phi^-_1(\bar{\nu}_n,x)=\bar{\kappa}_n\phi^+_2(\bar{\nu}_n,x)\label{kappa}\\
&\xLeftrightarrow{\nu_n\mapsto\bar{z}^{-1}_n}\overline{\phi^-_1(\bar{\nu}^{-1}_n,x)}=\bar{\kappa}_n\overline{\phi^+_2(\bar{\nu}^{-1}_n,x)}\\
&\xLeftrightarrow{(\ref{musym})}-i\nu_n\sigma_1\phi^-_1(\nu_n,x)\sigma_3=\bar{\kappa}_ni\nu_n\sigma_1\phi^+_2(\nu_n,x).
\end{align}
We derive
\begin{align*}
\bar{\kappa}_n=-\kappa_n\Longrightarrow\kappa_n\in i\mathbb{R}.
\end{align*}
We notice that
\begin{align}
L^\dag\sigma=-\sigma L^\dag, \ \ \sigma=\left(\begin{array}{cc}
0&1\\
-1&0
\end{array}\right),\label{lsigm}
\end{align}
where $L$ is defined in (\ref{laxt}) and $``\dag"$ denotes conjugate transpose.
\begin{align*}
\frac{\partial s_{11}(z)}{\partial k}\Big|_{z=\nu_n}=\frac{\det\left(\partial_k\phi^-_1,\phi^+_2\right)+\det\left(\phi^-_1,\partial_k\phi^+_2\right)}{1-z^{-2}}\Big|_{z=\nu_n}.
\end{align*}
Using (\ref{lsigm}), we have
\begin{align*}
\det\left(L\partial_k\phi^-_1,\phi^+_2\right)&=\left(L\partial_k\phi^-_1\right)^\dag\sigma\phi^+_2=\left(\partial_k\phi^-_1\right)^\dag L^{\dag}\sigma\phi^+_2\\
&=-\left(\partial_k\phi^-_1\right)^\dag \sigma L\phi^+_2=-\det\left(\partial_k\phi^-_1,L\phi^+_2\right).
\end{align*}
From the Lax pair, we derive
\begin{align*}
&\partial_x\partial_k\phi^-_1=\partial_k\partial_x\phi^-_1=\partial_k(L\phi^-_1)=L_k\phi^-_1+L\partial_k\phi^-_1,\\
\Longrightarrow &L_k=-i\sigma_3,~and,~ \partial_x\phi^+_2=L\phi^+_2.
\end{align*}
Furthermore,
\begin{align*}
\frac{\partial}{\partial x}\det\left(\partial_k\phi^-_1,\phi^+_2\right)&=\det\left(L_k\phi^-_1,\phi^+_2\right)+\det\left(L\partial_k\phi^-_1,\phi^+_2\right)+\det\left(\partial_k\phi^-_1,L\phi^+_2\right)\\
&=-i\det\left(\sigma_3\phi^-_1,\phi^+_2\right),\\
\frac{\partial}{\partial x}\det\left(\phi^-_1,\partial_k\phi^+_2\right)&=\det\left(\phi^-_1,L_k\phi^+_2\right)+\det\left(\phi^-_1,L\partial_k\phi^+_2\right)+\det\left(L\phi^-_1,\partial_k\phi^+_2\right)\\
&=-i\det\left(\phi^-_1,\sigma_3\phi^+_2\right).
\end{align*}
Recalling (\ref{kappa}) and the symmetries, we have
\begin{align*}
\det\left(\partial_k\phi^-_1,\phi^+_2\right)=-i\kappa_n\int_{-\infty}^{x}\det\left(\sigma_3\phi^+_2(\nu_n,s),\phi^+_2(\nu_n,s)\right)\mathrm{d}s,\\
\det\left(\phi^-_1,\partial_k\phi^+_2\right)=-i\kappa_n\int_x^{\infty}\det\left(\sigma_3\phi^+_2(\nu_n,s),\phi^+_2(\nu_n,s)\right)\mathrm{d}s,
\end{align*}
Sum the above equations, we have
\begin{align*}
\frac{\partial s_{11}(z)}{\partial k}\Big|_{z=\nu_n}\neq0.
\end{align*}
So $\nu_n$ is a simple zero.

$\uuline{Step~3}:$ Define a function
\begin{align*}
\mathbf{f}(z):=\det  \left(\phi_1^-(e^{i\eta},x ), \phi_2^+(e^{i\eta},x )\right):=u(\eta)+iv(\eta),
\end{align*}
where $z=e^{i\eta}, \eta\in[0,\pi]$ and $u(\eta), v(\eta)\in\mathbb{R}$. We notice that $\mathbf{f}(z)$ has infinitely zeros as $\eta\in[0,\pi]$.  According to the Bolzano-Weierstrass theorem, at least one accumulation point exists. However, since a non-zero analytic function can only have isolated singularities, we derive the accumulation point just locate on $z=\pm1$. We let the accumulation point is $z=1$,  there exists a convergence sequence $\{z_j=e^{i\eta_j}\}_{j=1}^{\infty}$
\begin{align}
z_j=e^{i\eta_j}\to1,\ \ j\to\infty,\ \ \mathbf{f}(z_j)=0.\label{squ}
\end{align}
In fact, $\lim\limits_{j\to\infty}\mathbf{f}(z_j)=\mathbf{f}(1)=d_+$. From (\ref{r<1}), we have
\begin{align}
|s_{11}(1)|^2=1+|s_{21}(1)|^2\geq1, \ \ z\to1.\label{s11>1}
\end{align}
We notice that the convergence sequence (\ref{squ}) is equivalent to
\begin{align*}
\eta_j\to0,\ \ j\to\infty,\ \ u(\eta_j)=v(\eta_j)=0.
\end{align*}
For each interval $[\eta_j,\eta_{j+1}]$, $u(\eta_j)$ and $v(\eta_j)$ perform the Rolle's theorem. There are two points $\eta_{j_1}$ and $\eta_{j_2}$ which satisfy
\begin{align*}
u'(\eta_{j_1})=v'(\eta_{j_1})=0.
\end{align*}
Moreover, we have
\begin{align*}
\eta_{j_1}, \eta_{j_2}\to0\Longleftrightarrow  e^{i\eta_{j_1}},  e^{i\eta_{j_2}}\to1, \ \ j\to\infty.
\end{align*}
From (\ref{s12}), we have
\begin{align*}
(z^2-1)s_{11}(z)=z^2\mathbf{f}(z),
\end{align*}
which differentiate for $z=e^{i\eta_{j_1}}$ and derive
\begin{align*}
2s_{11}(1)=\mathbf{f}'(1)=0,
\end{align*}
which is contradicts with (\ref{s11>1}). We prove that $z_n\neq\pm1$.
\end{proof}
From the Lemma \ref{zeros}, we derive the trace formulas of $S(z)$.
\begin{lemma}\label{trace}
The trace formulas are given by
\begin{align*}
&s_{11}(z)=\prod_{n=1}^{N}\frac{z-\nu_{n}}{z-\overline{\nu}_{n}}\exp\left(-\frac{1}{2i\pi}\int_{\mathbb{R}}\frac{v(\zeta)}{\zeta-z}\mathrm{d}{\zeta}\right),\\
&s_{22}(z)=\prod_{n=1}^{N}\frac{z-\overline{\nu}_{n}}{z-\nu_{n}}\exp\left(\frac{1}{2i\pi}\int_{\mathbb{R}}\frac{v(\zeta)}{\zeta-z}\mathrm{d}{\zeta}\right),
\end{align*}
where $\nu_n\in\mathcal{Z}^+$ and $\overline{\nu}_{n}\in\mathcal{Z}^-$ are the zero sets of $s_{11}(z)$ and $s_{22}(z)$, respectively. In addition, $v(\zeta)$ defined by
\begin{align}
v(\zeta) =  \log (1-|r(\zeta)|^2).\label{nu}
\end{align}
$s_{11}(z)$ and $s_{22}(z)$ are analytic  and ,has no singularity in $\mathbb{C}^+$ and  $\mathbb{C}^-$, respectively.
\end{lemma}

\section{Inverse scattering transform} \label{insrh}

\subsection{A basic RH problem}
\hspace*{\parindent}
According to the relation between $\phi^\pm$ and $S(z)$ in (\ref{phis}), together with the properties of Jost solutions and scattering data, on the upper/lower half plane, we define a sectionally meromorphic function $M(z):=M(z,x, t)$
\begin{align}
&M(z)=\left\{\begin{matrix} \left(\frac{\mu_1^{-}(z )}{{s_{11}(z) }},  \mu_2^{+}(z )\right), \  \ &z\in \mathbb{C}^+,\vspace{2mm} \cr
\left(\mu_1^{+}(z ),  \frac{\mu_2^{-}(z )}{{\overline{s_{11}(\bar z)}}} \right), \  \ &z\in \mathbb{C}^-. \end{matrix}\right. \nonumber
 \end{align}
This sectionally meromorphic function  $M(z)$ has the simple poles $\nu_n$ and $\overline{\nu}_{n}, (n=1,2,\cdots,N)$, with the following residue conditions
\begin{align}
& \mathop{\mathrm{Res}}\limits_{z=\nu_n} M (z) =\lim_{z\rightarrow \nu_n}M(z)\left(\begin{array}{cc} 0&0\\ c_ne^{2i t\theta(\nu_n)} &0\end{array}  \right), \label{fresm1}\\
&\mathop{\mathrm{Res}}\limits_{z=\bar \nu_n}M(z) =\lim_{z\rightarrow \bar \nu_n}M(z)\left(\begin{array}{cc} 0& -\bar c_ne^{-2i t\theta(\bar \nu_n)} \\0&0\end{array}\right),\label{fresm2}
\end{align}
where $c_n=\frac{s_{21}(\nu_n)}{s_{11}'(\nu_n)}. $

The piecewise meromorphic function $M(z)$ admits the  RH problem \ref{RHP0}.

\begin{prob} \label{RHP0}
 Fin  $M(z)$ with  properties:
\begin{itemize}
\item   $M(z)$ is meromorphic in $\mathbb{C}\backslash (\mathbb{R}\cup\mathcal{Z})$.

\item $M(z)$ satisfies the symmetries $M(z)=\sigma_1 \overline{M(\bar{z})} \sigma_1=z^{-1} M(z^{-1}) \sigma_2$.

\item  $M(z)$ admits  the asymptotic behavior
\begin{align}
&M(z)=I+\mathcal{O}(z^{-1}), \quad  z\rightarrow \infty,     \qquad\qquad \label{fRHP3}\\
&M(z)= \frac{\sigma_2}{z}+\mathcal{O}(1), \quad  z\rightarrow 0.
\end{align}

\item   $M(z)$ satisfies the jump condition
	$$M_+(z)=M_-(z)V(z), \; z \in \mathbb{R},$$
      	where
      	\begin{equation}\label{V0}
      		V(z)=\left(\begin{array}{cc}
      			1-|r(z)|^2 & -\overline{r(\overline{z})} e^{-2it \theta(z)}\\
      			r(z)e^{2it \theta(z)} & 1
      		\end{array}\right),
      	\end{equation}
      with $\theta(z)= \lambda\left(x/t+2k^2+1\right)$.
\item residue conditions:
\begin{align}
& \mathop{\mathrm{Res}}\limits_{z=\nu_n} M (z) =\lim_{z\rightarrow \nu_n}M(z)\left(\begin{array}{cc} 0&0\\ c_ne^{2i t\theta(\nu_n)} &0\end{array}  \right), \label{fresm1}\\
&\mathop{\mathrm{Res}}\limits_{z=\bar \nu_n}M(z) =\lim_{z\rightarrow \bar \nu_n}M(z)\left(\begin{array}{cc} 0& -\bar c_ne^{-2i t\theta(\bar \nu_n)} \\0&0\end{array}\right),\label{fresm2}
\end{align}
where $c_n=\frac{s_{21}(\nu_n)}{s_{11}'(\nu_n)}. $
\end{itemize}
\end{prob}
The real axis $\mathbb{R}$ is the jump contour of RH problem \ref{RHP0}. The poles lie on the unit circle $\mathcal{G}$, see Figure \ref{cpm}. The solution of the complex mKdV equation (\ref{cmkdv}) can be given by
\begin{align}
 q(x, t)=-i\lim_{z\rightarrow \infty} (zM (z ))_{12}, \label{sol}
 \end{align}
the subscript denotes the element located on the first row and second column of matrix function $M(z)$.

\subsection{Signature table and saddle points}
\hspace*{\parindent}
As $t\rightarrow\infty$,  the exponential oscillatory term $e^{\pm2it\theta(z) }$ in jump matrix $ V(z)$ plays a crucial role in the asymptotic analysis by the RH method. We rewrite $\theta(z)$ as
\begin{equation*}
\theta(z)=\left(\xi+1\right)\left(z-\frac{1}{z}\right)+\frac{1}{4}\left[z^3-\frac{1}{z^3}-\left(z-\frac{1}{z}\right)\right],\ \ \xi=\frac{x}{2t}.
\end{equation*}
The  decay of $e^{\pm2it\theta(z) }$ is determined by $\re[2i\theta(z)]$. After a series of calculations, $\re[2i\theta(z)]$ shaped as
\begin{equation*}
\re[2i\theta(z)]=-\im{z}\left(1+\frac{1}{|z|^2}\right)\left[2\xi+2+\frac{1}{2}\left(1+\frac{1}{|z|^4}\right)\left(3\mathrm{Re} ^2 z-\mathrm{Im} ^2 z\right)-\frac{2\mathrm{Re} ^2 z}{|z|^2}\right].
\end{equation*}
Let the critical lines as $\mathcal{Y}:\re [2i\theta(z)]=0$ and the unit circle $\mathcal{G}:\{z\in\mathbb{C}\mid|z|=1\}$.
The signature table of $\re[2i\theta(z)]$ shown in Figure \ref{proptheta}.

          \begin{figure}
    	\centering
    	    	\subfigure[$ \xi<-1.5$]{\label{figurea}
    		\begin{minipage}[t]{0.3\linewidth}
    			\centering
    			\includegraphics[width=1.5in]{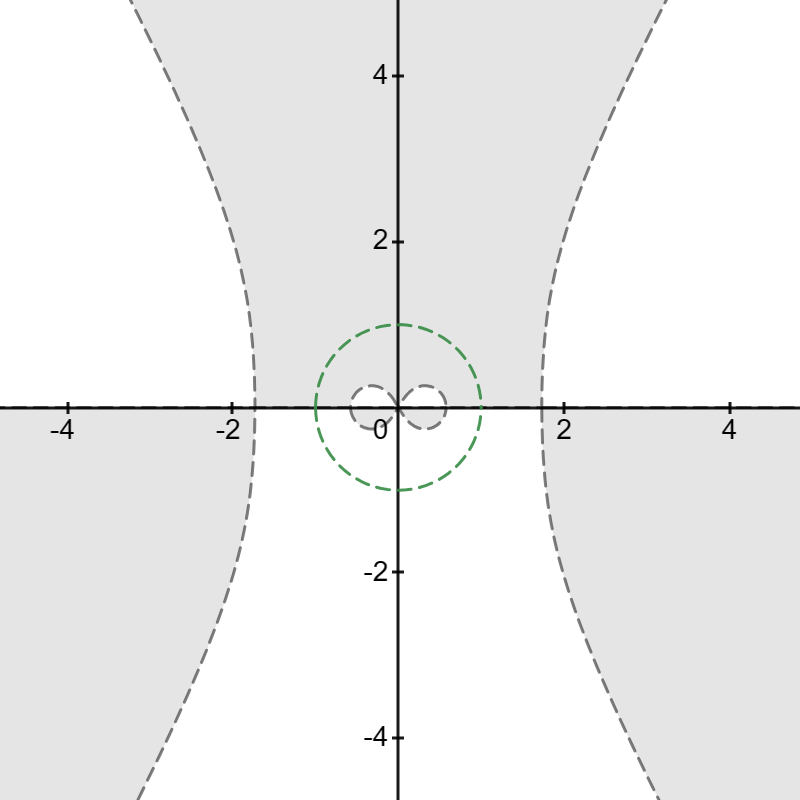}
    		\end{minipage}
    	}
    	\subfigure[$\xi=-1.5$]{\label{figureb}
    		\begin{minipage}[t]{0.3\linewidth}
    			\centering
    			\includegraphics[width=1.5in]{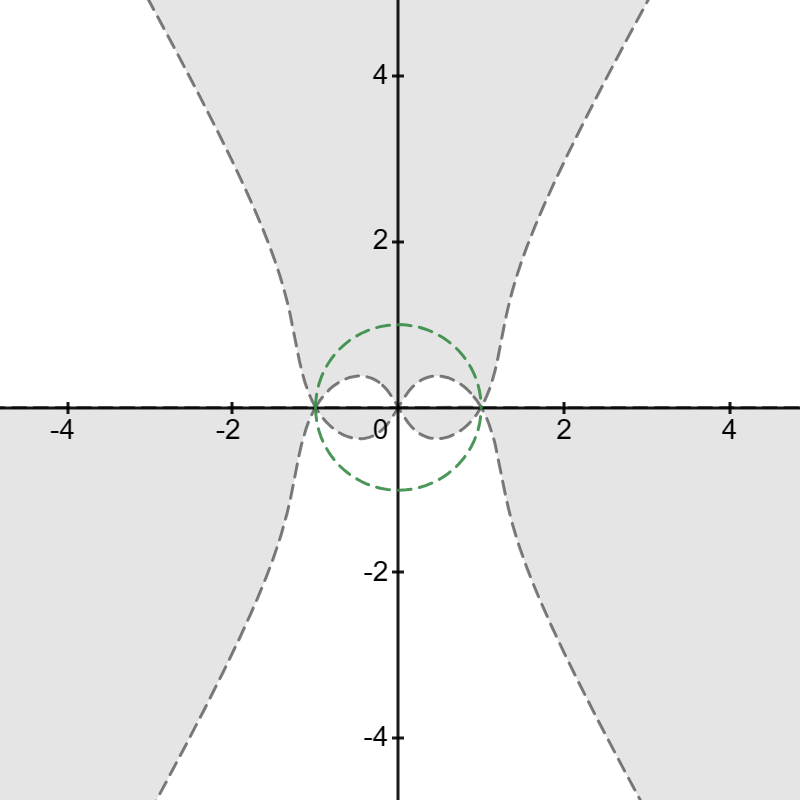}
    		\end{minipage}
    	}
    	\subfigure[$-1.5<\xi<-0.75$]{\label{figurec}
	\begin{minipage}[t]{0.3\linewidth}
		\centering
		\includegraphics[width=1.5in]{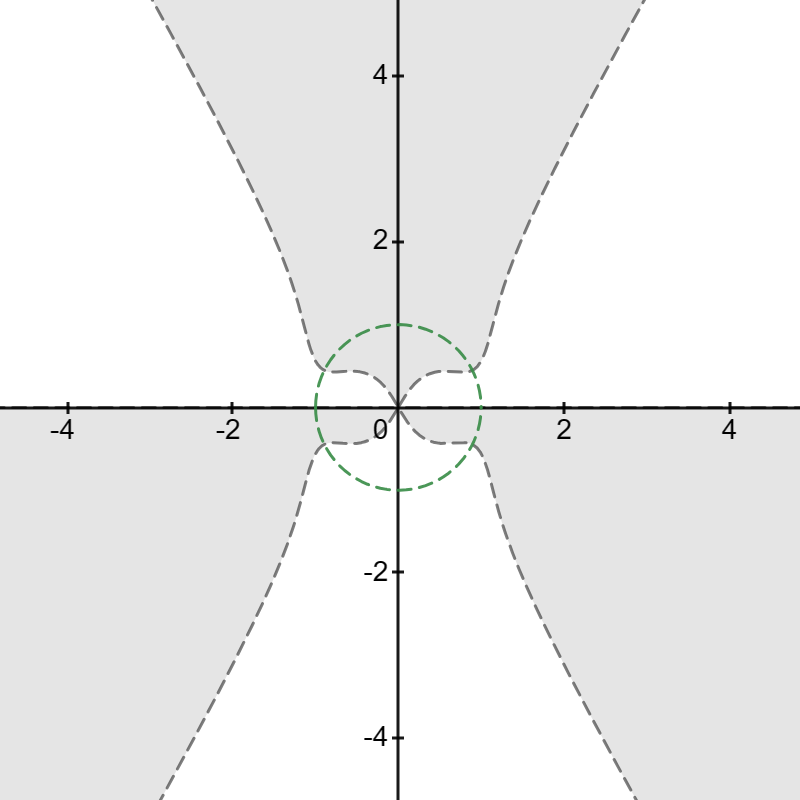}
	\end{minipage}
}

    \subfigure[$-0.75<\xi<-0.5$]{\label{figured}
   	\begin{minipage}[t]{0.3\linewidth}
    		\centering
    		\includegraphics[width=1.5in]{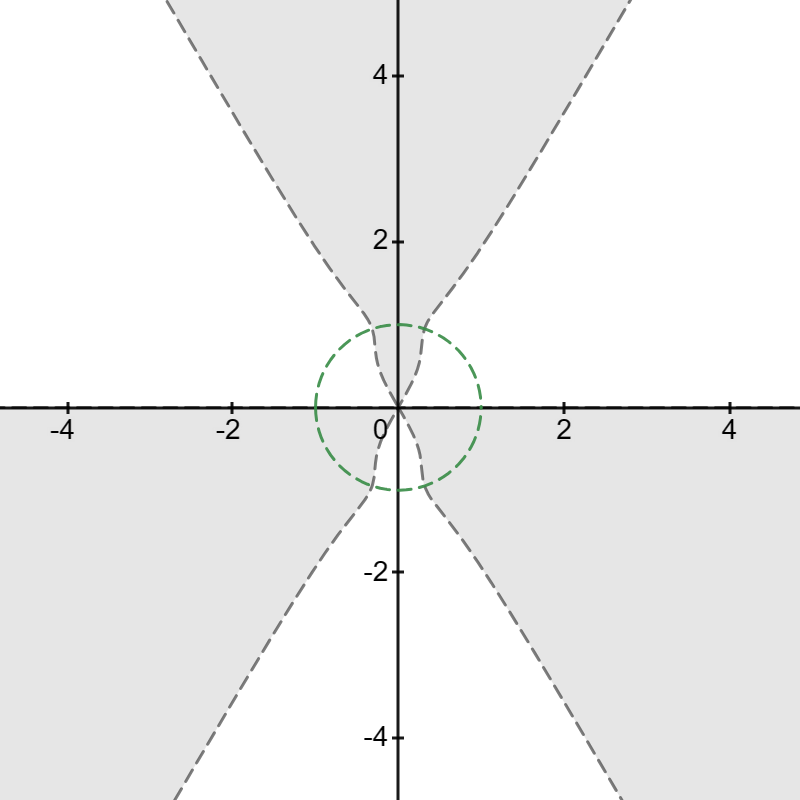}
    	\end{minipage}
    }
    \subfigure[$\xi=-0.5$]{\label{figuree}
    	\begin{minipage}[t]{0.3\linewidth}
    		\centering
    		\includegraphics[width=1.5in]{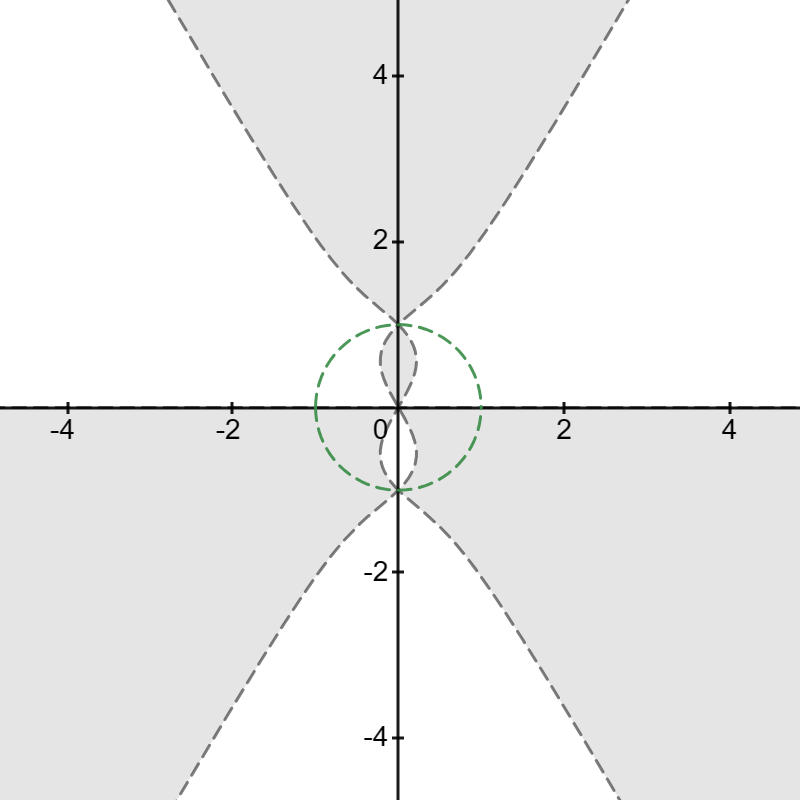}
    	\end{minipage}
    }
    \subfigure[$\xi>-0.5$]{\label{figuref}
	\begin{minipage}[t]{0.3\linewidth}
		\centering
		\includegraphics[width=1.5in]{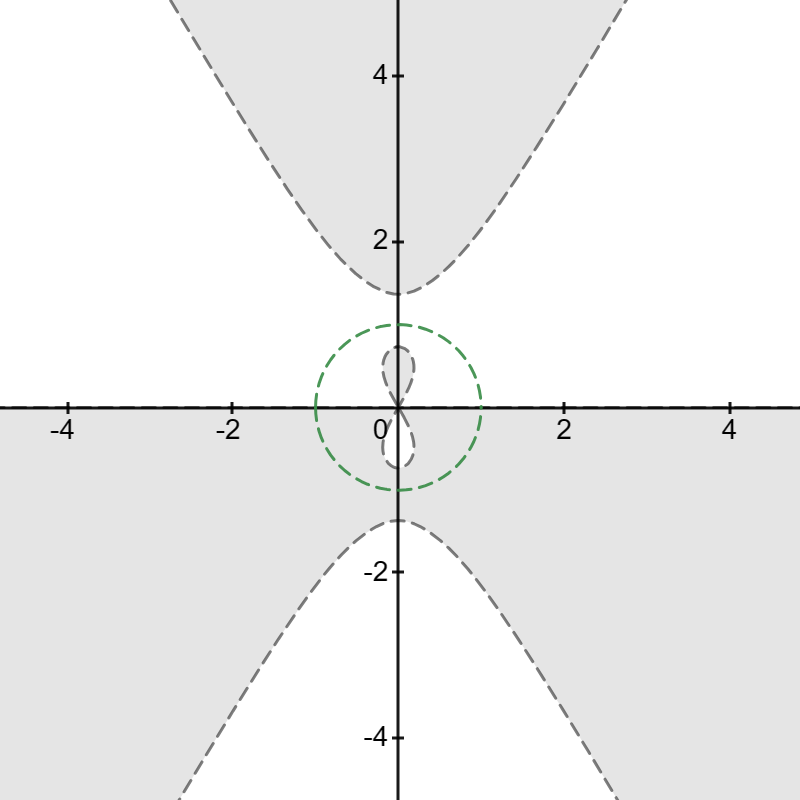}
	\end{minipage}
}
   	\centering
    	\caption{\footnotesize The signature table   of $\re [2i\theta(z)]$. The gray dashed lines denote the  critical line  $\mathcal{Y}:\re (2i\theta(z))=0$.
    In the gray regions, we have $\re [2i\theta(z)]>0$, which implies that $\left|e^{-2it\theta(z)}\right| \to 0$ as $t\to \infty$. In the white regions,  $\re [2i\theta(z)]<0$,
    which implies that $\left|e^{2it\theta(z)}\right| \to 0$ as $t\to \infty$. The green dashed circle is the  unit circle  $\mathcal{G}:$ $\{z\in\mathbb{C}\mid|z|=1\}$. }
    	\label{proptheta}
    \end{figure}

The stationary phase points are determind by  $\theta'(z)$ which can be simplified as
        \begin{align}
            &0=\theta'(z) =\frac{1}{4z^{4}}(1+z^2)(3+3z^4+4z^2\xi).\label{2theta}
        \end{align}
 The differential (\ref{2theta}) implies that there exist two fixed pure imaginary saddle points $z=\pm i$ and  the  distribution of the other stationary phase points is based on a discriminant $$\Bbbk:=4\xi^2-9.$$
\begin{itemize}
\item
For $\Bbbk>0\Longrightarrow|\xi|>\frac{3}{2}$,

$\uuline{\xi<-\frac{3}{2}}$, ~~~$\theta'(z)=0$  with  four different zeros on the real axis $z_{j}\in\mathbb{R}, (j=1,2,3,4)$, see Figure \ref{figurea},
   \begin{align}
z_{1}=-z_4=\frac{\sqrt{-2\xi+\sqrt{\Bbbk}}}{\sqrt{3}},\ \ \ \ \ z_{2}=-z_3=\frac{\sqrt{-2\xi-\sqrt{\Bbbk}}}{\sqrt{3}}.\label{z1234}
\end{align}
In fact, $z_j$ meeting a sequence   $z_{4}<-1<z_{3}<0<z_{2}<1<z_{1}$ and $z_{1}z_{2}=z_{3}z_{4}=1$.

$\uuline{\xi>-\frac{3}{2}}$, ~~~$\theta'(z)=0$  still with  four different zeros, which are restricted on the imaginary axis $z_{j}\in i\mathbb{R}, (j=1,2,3,4)$, see Figure \ref{figuref},
\begin{align}
z_{1}=-z_4=\frac{i\sqrt{2\xi+\sqrt{\Bbbk}}}{\sqrt{3}},\ \ \ \ \ z_{2}=-z_3=\frac{i\sqrt{2\xi-\sqrt{\Bbbk}}}{\sqrt{3}}.\label{iz1234}
\end{align}
Similarly, $\im{z}_{4}<-1<\im{z}_{3}<0<\im{z}_{2}<1<\im{z}_{1}$ and $z_{1}z_{2}=z_{3}z_{4}=-1$.

\item For $\Bbbk<0\Longrightarrow|\xi|<\frac{3}{2}$, we derive four stationary phase points on the complex plane $z_{j}\in\mathbb{C}\backslash\{\mathbb{R}\cup i\mathbb{R}\}, (j=1,2,3,4)$ and $|z_j|=1$ which implies that the zero points are controlled on the unit circle;   see Figure \ref{figurec}-\ref{figured}.

\end{itemize}
We classify the $x$-$t$ half-plane as three asymptotic regions based upon the interaction between the critical line $\mathcal{Y}$ and the unit circle $\mathcal{G}$; see Figure \ref{spacetime}.
\begin{itemize}
\item solitonless~region~$\mathcal{D}_1$ and $\mathcal{D}_3$:  There is no  interaction between the critical line $\mathcal{Y}$ and the unit circle $\mathcal{G}$ as $\xi<-3/2$ and $\xi>-1/2$. Moreover, $\mathcal{Y}$ is always far away from $\mathcal{G}$, corresponding with Figure \ref{figurea} and \ref{figuref}, respectively.
\item soliton~region~$\mathcal{D}_2$: For $-3/2<\xi\leq-1/2$, there is always interaction between $\mathcal{Y}$ and $\mathcal{G}$, the interaction points are held on the unit circle $\mathcal{G}$,  see Figures \ref{figurec} and \ref{figured}. Especially, for $\xi=-1/2$, the interaction  points  $z=\pm i$ locate on the imaginary axis $i\mathbb{R}$.
\item $\mathbf{transition~region~\mathcal{D}_4}$: For $\xi=-3/2$, the interaction points between $\mathcal{Y}$ and $\mathcal{G}$ are $z=1$ and $z=-1$ which are merge points of  saddle points $z_j, j=1,2$ and $z_j, j=3,4$, respectively, see Figure \ref{figureb}. Moreover, the norm of $(1-|r(\pm1)|^2)^{-1}$ blows up,  which implies the emergence of a new phenomenon   as $\xi\approx-3/2$. The original method to deal with the cases $\xi<-3/2$ and $-3/2<\xi<3/2$ is broken. This case is a transition region, which is the main study region in this paper.
\end{itemize}

\section{Painlev\'e-type asymptotics} \label{sec3}
\hspace*{\parindent}
In this section, we study  the Painlev\'e asymptotics    in the transition region $\mathcal{D}$ (\ref{D}).
In this case,
the four stationary points $z_j, j=1,2,3,4$ defined by (\ref{z1234}) are real and close to $z=\pm1$ at least the  velocity  of $t^{-1/3}$ as $t\to \infty$.

\subsection{Modifications to the basic RH problem}\label{modi1}
\hspace*{\parindent}
To obtain a standard RH problem without poles and  singularities, we perform some modifications to the basic RH problem \ref{RHP0} by removing the poles $\nu_n, \overline{\nu}_n$ and the singularities at $z=0$.
Since the  discrete spectrum numbers are finite and distributed on the unit circle $\mathcal{G}$. In particular, the  discrete spectrum is far away from the jump contour $\mathbb{R}$ and the critical line $\mathcal{Y}$.
The influence of the discrete spectrum decays  exponentially to the identity matrix $I$
when we convert their residues  into jumps on some small closed circles.  This shows that a modification of the basic RH problem \ref{RHP0} is possible. First, we remove the  influence of the discrete spectrums.

We split the set $\mathcal{N}$ into two subsets
\begin{align*}
\mathcal{N}_1:=\{n\in\mathcal{N}\mid\re [2i\theta(\nu_n)]>0\}, \ \ \mathcal{N}_2:=\{n\in\mathcal{N}\mid\re [2i\theta(\nu_n)]<0\}.
\end{align*}
 The zero $\nu_n$ leads to $\left|e^{2it\theta(\nu_n)}\right| \nrightarrow 0$ as  $n\in\mathcal{N}_1$. For $n\in\mathcal{N}_2$, the zero $\nu_n$ leads to $\left|e^{2it\theta(\nu_n)}\right| \rightarrow 0$.
In particular, in the region $\mathcal{D}$, all $\{n\in\mathcal{N}\mid n=1,\cdots,N\}\in\mathcal{N}_1$, i.e. $\mathcal{N}_1=\mathcal{N}$ which implies that $\mathcal{N}_2=\varnothing$.

    To remove poles $\nu_j$ and $\bar \nu_j$ and  open   the contour $\Gamma$ (which is defined in (\ref{gamma})) by  a  matrix   decomposition,
    we  introduce the following function
      \begin{equation}
        T(z)=
        \prod_{n=1}^{N} \Bigg(\frac{z-  \nu_n}{z\nu_n-1}\Bigg) \exp \left[  -\frac{1}{2i\pi} \int_\Gamma   v(\zeta) \left(\frac{1}{ \zeta-z}- \frac{1}{2\zeta} \right) \, \mathrm{d}\zeta \right],\label{tfcs}
        \end{equation}
where  $v(\zeta)$ is defined by (\ref{nu}) and
\begin{equation}
\Gamma=(-\infty,z_{4})\cup(z_{3},0)\cup(0,z_{2})\cup(z_{1},+\infty).\label{gamma}
\end{equation}
The properties of $T(z)$ were shown in Proposition \ref{prop1}.
\begin{proposition}\label{prop1}
            The function
           $T(z)$ is  meromorphic  in $\mathbb{C} \backslash   \Gamma $ with  simple zeros at the points $\nu_n$ and simple poles at the points $\bar{\nu}_n$, and satisfies the jump condition:
                    \begin{equation*}
                    T_+(z)=T_-(z)(1-|r(z)|^2), \quad z\in \Gamma.
                \end{equation*}
            Moreover, $T(z)$ admits the following properties:
             \begin{itemize}
                \item $T(z)$ satisfies the symmetries $\overline{T(\bar{z})}=T^{-1}(z)=T(z^{-1})$.
                                \item $T(z)$ admits the asymptotic behavior
                \begin{equation}
                    T(\infty) := \lim_{z \to \infty} T(z)= \Bigg(\prod_{n=1}^{N}\bar \nu_n\Bigg) \exp \left(\frac{1}{4i\pi} \int_\Gamma \frac{v(\zeta)}{\zeta} \, \mathrm{d}\zeta \right),\ \ \ \ z\to\infty.\label{Texpan1}
                \end{equation}
                Moreover, $| T(\infty)|^2=1$.
              And $T(z)$ hold a   asymptotic expansion
                \begin{equation}
                   T(z)=T(\infty)\Big[ 1-iz^{-1}T_1+\mathcal{O}(z^{-2})\Big],\label{Texpan}
                \end{equation}
                where $$T_1:=2\operatorname*{\sum}\limits_{n=1 }^{N} {\rm Im} \nu_{n}+\frac{1}{2\pi}\int_\Gamma v(\zeta)d\zeta.$$
            \end{itemize}
        \end{proposition}
Our main goal is to construct some small regions that can remove the influence coming from the discrete spectrums. To this end,
we define $\rho$, which is small enough
\begin{equation}\label{definerho}
 0<\rho <  \frac{1}{2} {\rm min}  \left\{   \operatorname*{min}\limits_{ \nu_n,  \nu_l\in  \mathcal{Z}^+ }  |\nu_n-\nu_l|,    \operatorname*{min}\limits_{\nu_n\in  \mathcal{Z}^+} |{\rm Im} \nu_n|,  \operatorname*{min}\limits_{\nu_n\in \mathcal{Z}^+,  { \rm Im} \theta(z) =0 }| \nu_n-z|  \right\}.
 \end{equation}
We construct $N$ small open disks $\odot_{n}:=\{z\in\mathbb{C}^+\mid|z-\nu_{n}|<\rho\}$ with counter-clockwise direction on the upper half-plane.
The construction way of $\odot_{n}$ ensure circles are pairwise disjoint. In the same way, $N$ open disks $\overline{\odot}_{n}:=\{z\in\mathbb{C}^-\mid|z-\overline{\nu}_{n}|<\rho\}$ with clockwise direction on the lower half-plane can also be constructed. We define a region $\mathcal{K}$ by
$$\mathcal{K}:=\bigcup\limits_{n=1}^{N}(\odot_{n}\cup\overline{\odot}_{n}),\ \ \partial\mathcal{K}:=\bigcup\limits_{n=1}^{N}(\partial\odot_{n}\cup\partial\overline{\odot}_{n}),$$
where $\partial\odot_{n}$ and $\partial\overline{\odot}_{n}$ denote the boundary of the disks $\odot_{n}$ and $\overline{\odot}_{n}$, respectively.
The direction of the real axis $\mathbb{R}$ goes from left to right, the direction of $\partial\odot_n$ is  counterclockwise, and  the direction of $\partial\overline{\odot}_n$  is clockwise; see Figure \ref{Djump65}.
To change  the residue of the pole $\nu_n$ into a jump on a circle $\odot_{n}$ , we define an interpolation function
            \begin{equation*}
           G(z) = \begin{cases}
             \left(\begin{array}{cc} 1& - \displaystyle {\frac{z-\nu_n}{c_n e^{2it\theta(\nu_n)}} }\\ 0&1\end{array}  \right), \;   z\in\odot_{n},   \\[4pt]
           \left(\begin{array}{cc} 1&0 \\ -\displaystyle {\frac{z-\bar{\nu}_n}{\bar{c}_n e^{ -2it\theta(\bar{\nu}_n) }}}&1\end{array}  \right), \;   z\in\overline{\odot}_n,\\
            I, \;  \;  \; else.
           \end{cases}
       \end{equation*}
and a directed contour
\begin{equation*}
	\Sigma^{(1)}=\mathbb{R} \cup\partial\mathcal{K}.
\end{equation*}

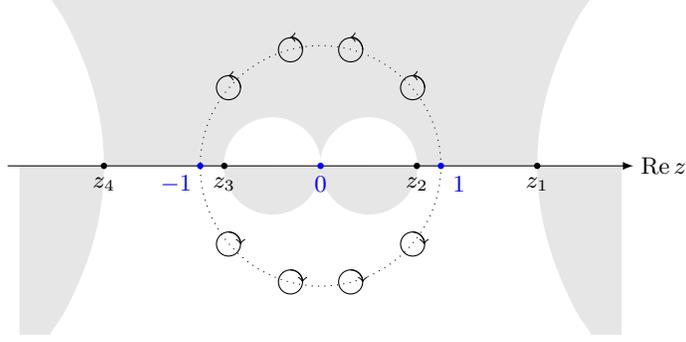
\begin{figure}[H]
\begin{center}
\begin{tikzpicture}[scale=0.8]
\path[fill=gray!20] (-4.5,0)rectangle (4.5,2.8);
\fill [white]  plot [domain=2.5:0,smooth] (3.6+\x*\x/9,\x) -- (4.6,0) -- (4.5,2.8);
\fill [gray!20]  plot [domain=-2.5:0,smooth] (3.6+\x*\x/9,\x) -- (4.6,0) -- (4.5,-2.8);
\fill [white]  plot [domain=2.5:0,smooth] (-3.6-\x*\x/9,\x) -- (-4.6,0) -- (-4.5,2.8);
\fill [gray!20]  plot [domain=-2.5:0,smooth] (-3.6-\x*\x/9,\x) -- (-4.6,0) -- (-4.5,-2.8);
\path[fill=gray!20] (-5,0)rectangle (-4.5,-2.8);
\path[fill=gray!20] (5,0)rectangle (4.5,-2.8);
  \filldraw[gray!20](-0.8,0)--(0,0) arc (0:-180:0.8);
    \filldraw[gray!20](0.8,0)--(1.6,0) arc (0:-180:0.8);
    \filldraw[white](-0.8,0)--(0,0) arc (0:180:0.8);
    \filldraw[white](0.8,0)--(1.6,0) arc (0:180:0.8);

\node    at (-3.6, -0.3)  {\footnotesize $z_4$};
\node    at (-1.6, -0.3)  {\footnotesize $z_3$};

\node    at (3.6, -0.3)  {\footnotesize $z_1$};
\node    at (1.6, -0.3)  {\footnotesize $z_2$};

 \draw [dotted ] (0, 0) circle [radius=2];
\draw[   -latex ](-5.2, 0)--(5.2, 0);
\node    at (5.7, 0)  {\footnotesize $\re z$};
\node  [blue]  at (0, -0.3)  {\footnotesize $0$};
        \coordinate (A) at (1.53,  1.3);
		\coordinate (B) at (1.53,  -1.3);
	
		\coordinate (E) at (-1.53,  1.3);
		\coordinate (F) at (-1.53,  -1.3);
		\coordinate (G) at (-1,  1.73);
		\coordinate (H) at (-1,  -1.73);
		\coordinate (I) at (0.5,  1.93);
		\coordinate (J) at (0.5,  -1.93);
		\coordinate (K) at (0,  0);
		\coordinate (L) at (2,  0);
		\coordinate (M) at (-2,  0);
	\coordinate (N) at (-0.5,  1.93);
		\coordinate (O) at (-0.5,  -1.93);

	\coordinate (P) at (-3.6,  0);
		\coordinate (Q) at (-1.6,  0);
	\coordinate (R) at (3.6, 0);
		\coordinate (S) at (1.6,  0);

\fill[] (P) circle (1.5pt);
\fill[] (Q) circle (1.5pt);

\fill[] (R) circle (1.5pt);
\fill[] (S) circle (1.5pt);
\fill[blue] (K)  circle (1.5pt);
\fill[blue] (L) circle (1.5pt);
\fill[blue] (M) circle (1.5pt);
 \draw [] (A) circle [radius=0.2];
  \draw [  -> ]  (1.73,  1.3) to  [out=90,  in=0] (1.53,  1.5);	
  \draw [] (B) circle [radius=0.2];
  \draw [ -> ]  (1.53,  -1.1) to  [out=0,  in=90] (1.72,  -1.3);

   \draw [] (E) circle [radius=0.2];
   \draw [  -> ]  (-1.33,  1.3) to  [out=90,  in=0]  (-1.53,  1.5);	
    \draw [] (F) circle [radius=0.2];
     \draw [  -> ]  (-1.53, -1.1) to  [out=0,  in=90]  (-1.33,  -1.3);	

       \draw [] (I) circle [radius=0.2];
          \draw [ -> ]  (0.7,  1.93) to  [out=90,  in=0]   (0.5,  2.14);
        \draw [] (J) circle [radius=0.2];
 \draw [  -> ]  (0.5,  -1.73) to  [out=0,  in=90]   (0.7,  -1.93);

    \draw [] (N) circle [radius=0.2];
          \draw [ -> ]  (-0.3,  1.93) to  [out=90,  in=0]   (-0.5,  2.14);
        \draw [] (O) circle [radius=0.2];
 \draw [  -> ]  (-0.5,  -1.73) to  [out=0,  in=90]   (-0.3,  -1.93);

  \node [blue]   at (2.3, -0.3)  {\footnotesize $1$};
    \node [blue]   at (-2.4, -0.3)  {\footnotesize $-1$};

\end{tikzpicture}
\end{center}
\caption{\footnotesize The jump contour  $\Sigma^{(1)}$  of  $M^{(1)}(z)$.}
\label{Djump65}
\end{figure}
A new function given by the following transformation
       \begin{equation}
        M^{(1)}(z)=T(\infty)^{-\sigma_3} M(z) G(z)T(z)^{\sigma_3}, \label{trans1}
       \end{equation}
and $M^{(1)}(z)$ satisfies the RH problem \ref{m1}.

 \begin{prob} \label{m1}
 Find  $M^{(1)}(z)=M^{(1)}(z;x,t)$ with properties
       \begin{itemize}
        \item   $M^{(1)}(z)$ is  analytical  in $ \mathbb{C} \setminus \Sigma^{(1)}$.
        \item   $M^{(1)}(z)=\sigma_1 \overline{M^{(1)}(\bar{z})}\sigma_1 =z^{-1}M^{(1)}(z^{-1})\sigma_2$.
        \item  $M^{(1)}(z)$ satisfies the jump condition
        \begin{equation*}
            M^{(1)}_+(z)=M^{(1)}_-(z)V^{(1)}(z),
        \end{equation*}
        where
\begin{equation}
	V^{(1)}(z)=\begin{cases}\begin{array}{ll}
b_-^{-1} b_+,  \  \ \  \;  z\in \Gamma,\\[8pt]
T(z)^{-\sigma_3} V(z)T(z)^{\sigma_3}, \  \ \ \;  z\in \mathbb{R}\backslash\Gamma,\\[6pt]
		\left(\begin{array}{cc}
			1 & -\frac {  z-\nu_n} {c_n} T^{-2}(z) e^{ -2it\theta( \nu_n)} \\
			0 & 1
		\end{array}\right),  \ \  \ \ 	 z\in\partial\odot_{n}, \\[12pt]
		\left(\begin{array}{cc}
			1 & 0	\\
			 \frac {  z-\bar \nu_n} {\bar c_n} T^{ 2}(z) e^{  2it\theta(\bar \nu_n)} & 1
		\end{array}\right),   \ \  \ \	z\in\partial\overline{\odot}_{n},
	\end{array}
\end{cases}
\label{jumpv1}
\end{equation}
with
\begin{equation}
b_-^{-1} b_+:=\left(\begin{array}{cc}
		1 &  -\overline{r(z)} T^{-2}(z)e^{-2it\theta(z)} \\
		0 & 1
	\end{array}\right)
		\left(\begin{array}{cc}
			1 & 0\\
r(z)T^2(z)e^{ 2it\theta(z)} & 1
		\end{array}\right).  \label{opep1}\\
\end{equation}
        \item $ M^{(1)}(z)$ admits the  asymptotic behaviors
        \begin{align*}
                &M^{(1)}(z)=I+\mathcal{O}(z^{-1}),	\quad  z \to  \infty,\\
                &M^{(1)}(z)=\frac{\sigma_2}{z}+\mathcal{O}(1), \quad z \to 0.
        \end{align*}

    \end{itemize}
\end{prob}

\begin{remark}
The jump on $\mathbb{R}\backslash\Gamma$ still is the jump matrix $V(z)$. The aim is to match the RH problem with the Painlev\'{e}-II model in Appendix \ref{appx}.
\end{remark}
  Since   the jump matrixes on the contour $\partial\mathcal{K}$ exponentially decay to the
  identity matrix at $t \to \infty$,  it can be shown that   the RH problem \ref{m1} is  asymptotically equivalent to the following  RH problem \ref{m2}.

 \begin{prob} \label{m2}
 Find   $M^{(2)}(z)=M^{(2)}(z;x,t)$ with properties
       \begin{itemize}
        \item   $M^{(2)}(z)$ is  analytical  in $\mathbb{C}\backslash \mathbb{R}$.
        \item  $M^{(2)}(z)=\sigma_1 \overline{M^{(2)}(\bar{z})}\sigma_1 =z^{-1}M^{(2)}(z^{-1})\sigma_2$.
        \item $M^{(2)}(z)$ satisfies the jump condition
        \begin{equation}
            M^{(2)}_+(z)=M^{(2)}_-(z)V^{(2)}(z),
        \end{equation}
        where
\begin{equation}
	V^{(2)}(z)=\begin{cases}
b_-^{-1}b_+,    \;  z\in \Gamma,\\[6pt]
T(z)^{-\sigma_3} V(z)T(z)^{\sigma_3},  \;  z\in \mathbb{R}\setminus\Gamma.
	\end{cases} \label{jumpv2}
\end{equation}

        \item $M^{(2)}(z)$ admits the asymptotic behaviors
        \begin{align}
                &M^{(2)}(z)=I+\mathcal{O}(z^{-1}),	\quad  z \to  \infty,\\
                &M^{(2)}(z)=\frac{\sigma_2}{z}+\mathcal{O}(1), \quad z \to 0.\label{m20}
        \end{align}

    \end{itemize}
\end{prob}



For removing the singularity as $z=0$, we introduce a transformation for $M^{(2)}(z)$
\begin{align}
M^{(2)}(z)=J(z) M^{(3)}(z),\label{trans3}
\end{align}
where $J(z)$ is given by
\begin{align}
J(z):= I+ \frac{1}{z} \sigma_2 M^{(3)}(0)^{-1},\label{J}
\end{align}
and $M^{(3)}(z)$ satisfies the RH problem \ref{ms3} without spectrum singularity.

 \begin{prob} \label{ms3}
 Find  $M^{(3)}(z)=M^{(3)}(z;x,t)$ with properties
       \begin{itemize}
        \item  $M^{(3)}(z)$ is  analytical  in $\mathbb{C}\backslash \mathbb{R}$.
        \item  $M^{(3)}(z)=\sigma_1 \overline{M^{(3)}(\bar{z})}\sigma_1 =\sigma_2 M^{(3)}(0)^{-1}M^{(3)}(z^{-1})\sigma_2$.
        \item $M^{(3)}(z) $ satisfies the jump condition:
        \begin{equation*}
            M^{(3)}_+(z)=M^{(3)}_-(z)V^{(2)}(z),
        \end{equation*}
        where
        $V^{(2)}(z)$ is given by   (\ref{jumpv2}).
        \item $M^{(3)}(z)$ admits the  asymptotics    $  M^{(3)}(z)=I+\mathcal{O}(z^{-1}),	\quad  z \to  \infty.$

    \end{itemize}
\end{prob}

\begin{proof}

We confirm that $M^{(3)}(z)$ satisfies the jump condition
\begin{align*}
M^{(3)}_+(z)&=J^{-1}(z)  M^{(2)}_+(z)=J^{-1}(z)  M^{(2)}_-(z) V^{(2)}(z)=J^{-1}(z) J(z)M^{(3)}_-(z)  V^{(2)}(z)=M^{(3)}_-(z)  V^{(2)}(z).
\end{align*}
using the information that $M^{(2)}(z)$ satisfies the RH problem \ref{m2}.
To show $M^{(3)}(z)$ has no singularity as $z=0$, we verify that  $M^{(3)}(z)$ with
the expansion
\begin{align*}
M^{(3)}(z) = M^{(3)}(0)+ z \widetilde{M}^{(3)}(z),
\end{align*}
which tests and verifies $M^{(2)}(z)$ satisfies (\ref{m20}).
Substituting the expansion of $M^{(3)}(z)$ into (\ref{trans3}) yields
\begin{align*}
M^{(2)}(z) =  \frac{1}{z} \sigma_2 + M^{(3)}(0) + z \widetilde{M}^{(3)}(z)+\sigma_2 M^{(3)}(0)^{-1} \widetilde{M}^{(3)}(z)= \frac{1}{z} \sigma_2 +\mathcal{O}(1),
\end{align*}
which implies that $M^{(3)}(z)$ has no singularity as $z=0$.
\end{proof}

\subsection{A hybrid $\bar{\partial}$-RH problem} \label{modi2}
\hspace*{\parindent}
In this subsection, we open the $\bar{\partial}$-lenses to deform  the contour $\Gamma$ as shown in Figure \ref{signdbar}. We are able to do the $\bar{\partial}$ continuous extensions due to  this contour deformation. With  a suitably small angle $0<\varphi <\pi/4$,
the directed rays $\Sigma_j\in\mathbb{C}^+, (j=0\pm,1,2,3,4)$ and $L_{0\pm}\in\mathbb{C}^+$ in Figure \ref{signdbar} defined by
 \begin{align*}
     & \Sigma_{0+}:=h e^{i\varphi},\ \ L_{0+}:=\hat{h}e^{i\pi/2}+z_{2}/2,\ \  \Sigma_{2}:=h e^{i(\pi-\varphi)}+z_{2}, \ \ \Sigma_{1}:=\mathbb{R}^+ e^{i\varphi}+z_{1},\\
     & \Sigma_{0-}:=h e^{i(\pi-\varphi)},\ \ L_{0-}:=\hat{h}e^{i\pi/2}+z_{3}/2,\ \ \Sigma_{3}:=h e^{i\varphi}+z_{3}, \ \ \Sigma_{4}:=\mathbb{R}^+ e^{i(\pi-\varphi)}+z_{4},
     \end{align*}
   where
    $h\in \left(0,  \frac{1}{2}z_{2}\sec \varphi\right)$ and $\hat{h}\in \left(0,  \frac{1}{2}z_{2}\tan \varphi\right)$
    .
Conjugate rays of $\Sigma_j\in\mathbb{C}^+$ and $L_{0\pm}\in\mathbb{C}^+$ are the directed rays $\overline{\Sigma}_j\in\mathbb{C}^-$ and $\overline{L}_{0\pm}\in\mathbb{C}^-$, respectively. A directed ray $\Sigma$ is defined as
$$\Sigma:=\bigcup\limits_{j=0\pm,1,2,3,4}\left(\Sigma_j\cup\overline{\Sigma}_j\right)\bigcup\limits_{j=0\pm}\left(L_j\cup\overline{L}_j\right).$$
The regions $\Omega_{j}\subset\mathbb{C}^+ (j=0\pm,1,2,3,4)$  in Figure \ref{signdbar} are defined by
\begin{align*}
&\Omega_{0+}:=\left\{z\in\mathbb{C}^{+}\mid\arg{z}\in(0,\varphi),\re z\in{(0,z_{2}/2)}\right\},\ \ \Omega_{0-}:=\left\{z\in\mathbb{C}^{+}\mid-\overline{z}\in\Omega_{0+}\right\},\\
&\Omega_{2}:=\left\{z\in\mathbb{C}^{+}\mid\arg{z}\in(\pi-\varphi,\pi),\re z\in{(z_{2}/2,z_{2})}\right\},\ \ \Omega_{3}:=\left\{z\in\mathbb{C}^{+}\mid-\overline{z}\in\Omega_{2}\right\},\\
&\Omega_{1}:=\left\{z\in\mathbb{C}^{+}\mid\arg{z}\in(0,\varphi),\re z\in{(z_{1},+\infty)}\right\},\ \
\Omega_{4}:=\left\{z\in\mathbb{C}^{+}\mid-\overline{z}\in\Omega_{1}\right\},
\end{align*}
and $\overline{\Omega}_j\subset\mathbb{C}^-$ are the  conjugate regions of $\Omega_j$. We define a region $\Omega$ by
$$\Omega:=\bigcup\limits_{j=0\pm,1,2,3,4}\left(\Omega_j\cup\overline{\Omega}_j\right).$$
In fact, $\Sigma$ is the border of $\Omega$.
\begin{remark}

All sectors opened by this angle must fall within their respective decaying areas, hence the angle $\varphi$ must small enough. Furthermore,
$\Sigma$ does not cross with  $\partial\mathcal{K}$. Figure \ref{signdbar} illustrates this. The same is true between $\Sigma$  and crucial lines $\mathcal{Y}$, .
\end{remark}

    \begin{figure}[H]
        \begin{center}
  \begin{tikzpicture}[scale=0.8]
  \path[fill=gray!20] (-6,0)rectangle (6,2.5);
\fill [white]  plot [domain=2.5:0,smooth] (4.6+\x*\x/9,\x) -- (6,0) -- (6.1,2.5);
\fill [gray!20]  plot [domain=-2.5:0,smooth] (4.6+\x*\x/9,\x) -- (6,0) -- (6.1,-2.5);
\fill [white]  plot [domain=2.5:0,smooth] (-4.6-\x*\x/9,\x) -- (-6,0) -- (-6.1,2.5);
\fill [gray!20]  plot [domain=-2.5:0,smooth] (-4.6-\x*\x/9,\x) -- (-6,0) -- (-6.1,-2.5);
 \filldraw[gray!20](-1.6,0)--(-2.6,0) arc (180:360:1.3);
   \filldraw[gray!20](-1.0,0)--(-0.0,0) arc (180:360:1.3);
    \filldraw[white](-1.6,0)--(0,0) arc (0:180:1.3);
     \filldraw[white](1.3,0)--(2.6,0) arc (0:180:1.3);

                \draw[dotted,black!40,-latex](-6,0)--(6.2,0)node[black,right]{$\re z$};
                 \draw [ ] (-4.6, 0)--(-2.6,0);
                                    \draw [ ] (4.6, 0)--(2.6,0);
                                 \draw [-latex ] (-4.6, 0)--(-3.5,0);
                                    \draw [-latex ] (2.6, 0)--(4,0);

                \node[shape=circle,fill=blue, scale=0.15]  at (-3.2,0){0} ;
                      \node[shape=circle,fill=blue, scale=0.15]  at (3.2,0){0} ;
            \node[shape=circle,fill=black,scale=0.15] at (-2.6,0) {0};
              \node[shape=circle,fill=black,scale=0.15] at (2.6,0) {0};
             \node[shape=circle,fill=black,scale=0.15] at (-4.6,0) {0};
             \node[shape=circle,fill=black,scale=0.15] at (4.6,0) {0};
                \node[shape=circle,fill=blue,scale=0.15] at (0.0,0) {0};
                		\coordinate (A) at (-1.3,  0);
                \coordinate (B) at (1.3,  0);

                \node[below,blue] at (0.0,0) {\footnotesize $0$};
                 \node[below,blue] at (-3.2,0) {\footnotesize $-1$};
                 \node[below,blue] at (3.2,0) {\footnotesize $1$};

                  \node[below] at (-4.3,0) {\footnotesize $z_4$};
                 \node[below] at (4.4,0) {\footnotesize $z_1$};
                   \node[below] at (-2.6,0) {\footnotesize $z_3$};
                 \node[below] at (2.6,0) {\footnotesize $z_2$};

                  \node[below] at (-0.75,0.6) {\footnotesize $\Omega_{0-}$};
                 \node[below] at (0.85,0.55) {\footnotesize $\Omega_{0+}$};
                 \node[below] at (-0.85,0.1) {\footnotesize $\overline{\Omega}_{0-}$};
                 \node[below] at (0.85,0.1) {\footnotesize $\overline{\Omega}_{0+}$};

                     \node[below] at (-0.75,1.2) {\footnotesize $\Sigma_{0-}$};
                 \node[below] at (0.85,1.2) {\footnotesize $\Sigma_{0+}$};
                 \node[below] at (-0.4,-0.4) {\footnotesize $\overline{\Sigma}_{0-}$};
                 \node[below] at (0.62,-0.4) {\footnotesize $\overline{\Sigma}_{0+}$};

                  \node[below] at (-1.8,0.6) {\footnotesize $\Omega_{3}$};
                 \node[below] at (1.8,0.55) {\footnotesize $\Omega_{2}$};
                 \node[below] at (-1.8,0.1) {\footnotesize $\overline{\Omega}_{3}$};
                 \node[below] at (1.8,0.1) {\footnotesize $\overline{\Omega}_{2}$};

                   \node[below] at (-1.8,1.2) {\footnotesize $\Sigma_{3}$};
                 \node[below] at (1.8,1.2) {\footnotesize $\Sigma_{2}$};
                 \node[below] at (-1.8,-0.4) {\footnotesize $\overline{\Sigma}_{3}$};
                 \node[below] at (2,-0.4) {\footnotesize $\overline{\Sigma}_{2}$};

                     \node[below] at (-5.4,0.6) {\footnotesize $\Omega_{4}$};
                 \node[below] at (5.4,0.55) {\footnotesize $\Omega_{1}$};
                 \node[below] at (-5.4,0.1) {\footnotesize $\overline{\Omega}_{4}$};
                 \node[below] at (5.4,0.1) {\footnotesize $\overline{\Omega}_{1}$};

                    \node[below] at (-5.4,1.6) {\footnotesize $\Sigma_{4}$};
                 \node[below] at (5.4,1.6) {\footnotesize $\Sigma_{1}$};
                 \node[below] at (-5.4,-1.0) {\footnotesize $\overline{\Sigma}_{4}$};
                 \node[below] at (5.4,-1.0) {\footnotesize $\overline{\Sigma}_{1}$};

                                 \draw [ ] (0.0, 0)--(-1.3,0.9);
                                    \draw [ ] (0.0, 0)--(1.3,0.9);
                                 \draw [ ] (0.0, 0)--(-1.3,-0.9);
                                  \draw [ ] (0.0, 0)--(1.3,-0.9);
                                 \draw[](-2.6,0)--(-1.3,0.9);
                                    \draw[](2.6,0)--(1.3,0.9);
                                  \draw[](-2.6,0)--(-1.3,-0.9);
                                  \draw[](2.6,0)--(1.3,-0.9);
                               \draw [ ] (4.6,0 )--(6,-1.5);
                                 \draw [] (4.6,0 )--(6,1.5);
                                 \draw[-latex](4.6,0 )--(5.25,-0.7);
                                       \draw[-latex](4.6,0 )--(5.25,0.7);
                                 \draw [ ] (-4.6,0)--(-6,-1.5);
                                 \draw [ ] (-4.6,0 )--(-6,1.5);
                                 \draw[-latex,](-6.0,-1.5 )--(-5.25,-0.7);
                                       \draw[-latex](-6.0,1.5 )--(-5.25,0.7);

                                       \draw[](-1.3,0)--(-1.3,0.9);
                                             \draw[](1.3,0)--(1.3,0.9);
                                              \draw[-latex](-1.3,0)--(-1.3,0.45);
                                             \draw[-latex](1.3,0)--(1.3,0.45);

                                    \draw[](-1.3,0)--(-1.3,-0.9);
                                             \draw[](1.3,0)--(1.3,-0.9);
                                              \draw[-latex](-1.3,0)--(-1.3,-0.45);
                                             \draw[-latex](1.3,0)--(1.3,-0.45);

                            \draw [-latex] (-2.6,0)--(-1.95,0.45);
                              \draw [-latex] (-1.3,0.9)--(-0.65,0.45);
                                  \draw [-latex] (0,0)--(0.65,0.45);
                              \draw [-latex] (1.3,0.9)--(1.95,0.45);
                             \draw [-latex] (-2.6,0)--(-1.95,-0.45);
                             \draw [-latex] (-1.3,-0.9)--(-0.65,-0.45);
                                \draw [-latex] (0,0)--(0.65,-0.45);
                              \draw [-latex] (1.3,-0.9)--(1.95,-0.45);

  \end{tikzpicture}
            \caption{ \footnotesize{The jump contour $\Gamma$ be opened as $\Sigma$ . }}
      \label{signdbar}
        \end{center}
    \end{figure}
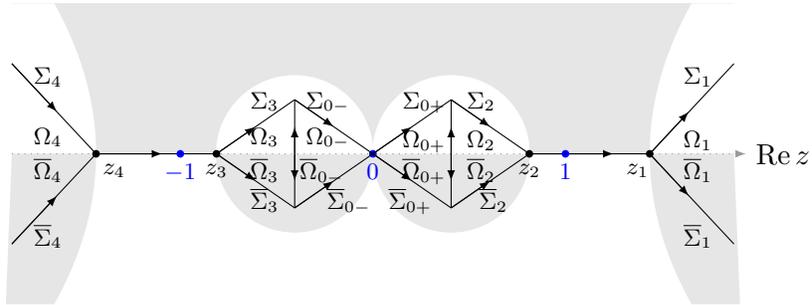

The estimates of $\re[2i\theta(z)]$, which are displayed in the following Proposition \ref{rezpm}, demonstrate the decaying charateristics of the oscillating factors $e^{\pm2it\theta(z)}$ in $\Omega$.
\begin{proposition}\label{rezpm}In the transition region  $ \mathcal{D}$,
$\re[2i\theta(z)]$ holds the following estimates
 \begin{itemize}
        \item For $z\in\{\Omega_{0\pm}\cup\overline{\Omega}_{0\pm}\},$
\begin{subequations}
\begin{align}
&\re[2i\theta(z)]\leq -c_{0}(\tilde{\varphi})|\im z||\sin{\tilde{\varphi}}|,\ \ z\in\Omega_{0\pm},\\
&\re[2i\theta(z)]\geq c_{0}(\tilde{\varphi})|\im z||\sin{\tilde{\varphi}}|,\ \ z\in\overline{\Omega}_{0\pm},
\end{align}\label{0pm}
\end{subequations}
where  $\tilde{\varphi}=\arg{z}$ and $c_{0}(\tilde{\varphi})$ depend on $\tilde{\varphi}$.

\item For $z\in\{\Omega_{j}\cup\overline{\Omega}_{j}\},~j=1,4$,
\begin{align}
&\re[2i\theta(z)]\leq\begin{cases}
-2|\re z-z_j|^2|\im z|,\ \ z\in\Omega_{j}\cap\left\{|z|\leq\sqrt2\right\},\\
-2(\sqrt{2}-1)|\im{z}|,\ \ z\in\Omega_j\cap\left\{|z|>\sqrt{2}\right\},
\end{cases}\\
&\re[2i\theta(z)]\geq\begin{cases}
2|\re z-z_j|^2|\im z|,\ \ z\in\overline{\Omega}_{j}\cap\left\{|z|\leq\sqrt2\right\},\\
2(\sqrt{2}-1)|\im{z}|,\ \ z\in\overline{\Omega}_j\cap\left\{|z|>\sqrt{2}\right\}.
\end{cases}
\end{align}
\item For $z\in\{\Omega_{j}\cup\overline{\Omega}_{j}\},~j=2,3$,
\begin{align}
&\re[2i\theta(z)]\leq-c(z_j)|\re z-z_j|^2|\im z|,\ \ z\in\Omega_{j},\\
&\re[2i\theta(z)]\geq c(c_j)|\re z-z_j|^2|\im z|,\ \ z\in\overline{\Omega}_{j},
\end{align}
where  $c(z_j)$ depends on $z_j$.
\end{itemize}
\end{proposition}
\begin{proof}
For $z\in\{\Omega_{0\pm}\cup\overline{\Omega}_{0\pm}\}$, we take $z\in\Omega_{0+}$ as an example and the remaining regions can be derived using the same method.  Let $z=|z| e^{i\tilde{\varphi}}=u+iv$ where $u>v>0$ and $0<\tilde{\varphi}<\varphi$. Define a function $F(l)=l+l^{-1}$ with $l>0$. Rewrite $\re{(2i\theta)}$ as
\begin{equation*}
\re{(2i\theta)}=-F(|z|)g(z)\sin{\tilde{\varphi}},
\end{equation*}
where $$g(z)=2\xi-3\cos{2\tilde{\varphi}}+\frac{1}{2}F^{2}(|z|)(1+2\cos{2\tilde{\varphi}}).$$
Since $\re{(2i\theta)}=0$, we have
\begin{equation}
g(z)=2\xi-3\cos{2\tilde{\varphi}}+\frac{1}{2}F^{2}(|z|)(1+2\cos{2\tilde{\varphi}})=0.\label{gz=0}
\end{equation}
This equation (\ref{gz=0}) implies that
\begin{equation*}
4<\alpha:=F^{2}(|z|)=3-\frac{4\xi+3}{1+2\cos{2\tilde{\varphi}}}.
\end{equation*}
Moreover, we have $F(|z|)=|z|+|z|^{-1}=\sqrt{\alpha}$, which admits two solutions
\begin{equation*}
|z|_{1}=\frac{\sqrt{\alpha}-\sqrt{\alpha-4}}{2},\ \ |z|_{2}=\frac{\sqrt{\alpha}+\sqrt{\alpha-4}}{2},
\end{equation*}
and $|z|_{1}<|z|_{2}$. It is easy to check that $h(l):=l^2-\sqrt{\alpha}l+1$ is a monotonically increasing function as $l\in\left(\frac{\sqrt{\alpha}}{2},+\infty\right)$. And it is a monotonically decreasing function as $l\in\left(-\infty,\frac{\sqrt{\alpha}}{2}\right)$. Since $|z|_{1}<\frac{\sqrt{\alpha}}{2}$, $h(l)$ is monotonically decreasing on $(0,|z|_{1})$. We have $h(l)>h(|z|_{1})=0$ and $F(l)>\sqrt{\alpha}$, which implies that $$g(z)>0.$$
There is a constant $c_{0}=c_{0}(\tilde{\varphi})$ such that
\begin{equation*}
\re{(2i\theta)}\leq -c_{0}F(|z|)\sin{\tilde{\varphi}}=-c_{0}\sqrt{\alpha}\sin{\tilde{\varphi}}.
\end{equation*}
We derive the equations (\ref{0pm}).

For $z\in\{\Omega_{j}\cup\overline{\Omega}_j\}_{j=1}^4$, we prove $z\in\Omega_1$ and the remian domains can be proved using the same method. Let $z=z_{1}+u+iv=z_1+|z-z_1|e^{i\tilde{\varphi}}\in\Omega_{1}$, where $u>0$ and
$v=u\tan{\tilde{\varphi}}>0, |z|^2=(u+z_{1})^2+u^2\tan^2\tilde{\varphi}.$
By simplification,  $\re(2i\theta)$ can be rewritten as
\begin{align*}
\re(2i\theta)
&=-vg(u,v),
\end{align*}
where
\begin{align*}
g(u,v)&=\left(1+\frac{1}{|z|^6}\right)\left[3(z_{1}+u)^2-v^2+\frac{|z|^6+|z|^4}{|z|^6+1}\left(2\xi+\frac{3}{2}\right)\right]\\
&>3(z_{1}+u)^2-v^2+\frac{|z|^6+|z|^4}{|z|^6+1}\left(2\xi+\frac{3}{2}\right)\\
&:=f(u,v).
\end{align*}
Let $h(|z|):=\frac{|z|^6+|z|^4}{|z|^6+1}$, and $h(|z|)$ with   max/min value
\begin{equation}
\begin{cases}
0\leq h(|z|)\leq\frac{4}{3},\ \ |z|\leq\sqrt{2},\\
1< h(|z|)\leq\frac{4}{3},\ \ |z|>\sqrt{2}.
\end{cases}\label{zsqrt2}
\end{equation}
For $|z|\leq\sqrt{2}$, from  (\ref{z1234}), we have
\begin{equation}
\xi=-\frac{3z_{1}^4+3}{4z_{1}^2}\Longrightarrow2\xi+\frac{3}{2}=-\frac{3}{2}(z_{1}^2+z_{1}^{-2}-1).\label{ab}
\end{equation}
Substituting (\ref{ab}) into $f(u,v)$, we have
\begin{equation*}
f(u,v)=3(z_{1}+u)^2-v^2-\frac{3}{2}h(|z|)(z_{1}^2+z_{1}^{-2}+1).
\end{equation*}
By making some estimates for $f(u,v)$ under the condition (\ref{zsqrt2}), we have
\begin{equation*}
f(u,v)\geq 3(z_{1}+u)^2-v^2-2(z_{1}^2+z_{1}^{-2}+1).
\end{equation*}
For the reason that $v=(z_1+u)\tan{\tilde{\varphi}}$ and $0<\tilde{\varphi}<\varphi<\frac{\pi}{4}$,
\begin{align*}
f(u,v)&\geq 3(z_{1}+u)^2-(z_1+u)^2\tan^2{\tilde{\varphi}}-2(z_{1}^2+z_{1}^{-2}+1)\\
&>2u^2+4z_1u-2(z_{1}^{-2}+1)\\
&\geq2u^2.
\end{align*}
For $|z|>\sqrt{2}$, moreover, $\xi\rightarrow\left(-\frac{3}{2}\right)^-$,
\begin{equation*}
f(u,v)=3(z_{1}+u)^2-v^2-\frac{3}{2}h(|z|).
\end{equation*}
Considering  condition (\ref{zsqrt2}) as $|z|>\sqrt{2}$, and  $v=(z_1+u)\tan{\tilde{\varphi}}$ with $0<\tilde{\varphi}<\varphi<\frac{\pi}{4}$,
\begin{align*}
f(u,v)&\geq3(z_{1}+u)^2-v^2-2
\geq2(z_{1}+u)^2-2
\geq4\cos{\tilde{\varphi}}-2
\geq2(\sqrt{2}-1).
\end{align*}
The third case can be derived in the same way.
\end{proof}

     Next we perform continuous  expansions for the jump matrix $V^{(2)}$ to remove the jump from $\Gamma$.
\begin{lemma} \label{prop3}
   Let $q_0(x)-\tanh x\in L^{1,2}(\mathbb{R})$.    There exist the boundary value functions $R_{j}(z):\overline{\Omega}_j\mapsto \mathbb{C}, j=0\pm,1,2,3,4$ continuous on $\overline{\Omega}_j$ and with continuous first partial derivative on $\Omega_j$, which defined by
       \begin{align*}
       R_{j}(z)= \begin{cases}
            r(z)T^2(z), \quad z \in \Gamma,\\
            r(z_j)T^2(z_j),\ \ z\in \Sigma_{j},\ j=1,2,3,4,\\
            0,\quad z \in \Sigma_{0\pm},\ j=0\pm,
          \end{cases}
  \end{align*}
  the conjugate function $\overline{R}_{j}(z)$ is given by
  \begin{align*}
        \overline{R}_{j}(z)= \begin{cases}
                \overline{r}(z)T^{-2}(z), \quad  z \in \Gamma,\\
                \overline{r}(z_j)T^{-2}(z_j),\ \ z\in\overline{\Sigma}_j,\ j=1,2,3,4,\\
                 0,\quad z \in \overline{\Sigma}_{0\pm},\ j=0\pm.
               \end{cases}
       \end{align*}
       There exist a  constant $c:=c(q_{0})$ and a cutoff function $\omega\in C_0^{\infty}(\mathbb{R},[0,1])$ with small support near 1. We have
      \begin{equation*}
      |\overline{\partial}R_j|, |\overline{\partial}\overline{R}_j|\leq\begin{cases}
      c(|r'(|z|)|+|z-z_j|^{-1/2}+\varphi(|z|)),\ \ z\in\Omega_j\cup\overline{\Omega}_j,\\
      c|z-1|,\ \ z\in\Omega_j\cup\overline{\Omega}_j ~and~|z-1| ~small ~enough,\\
      c(|r'(\pm|z|)|+|z|^{-1/2}),\ \ z\in\Omega_{0\pm}\cup\overline{\Omega}_{0\pm}.
      \end{cases}
      \end{equation*}
Moreover, setting $R:\Omega\mapsto\mathbb{C}$ by $R(z)|_{z\in\Omega_j}=R_j(z)$ and $R(z)|_{z\in\overline{\Omega}_j}=\overline{R}_j(z)$, the extension can be made such that $R(z)=\overline{R}(\overline{z}^{-1})$.
        \end{lemma}

Define a matrix function as
    \begin{align}\label{R3}
    	R^{(3)}(z)=\begin{cases}
      		\left(\begin{array}{cc}
    	1&  0\\
    -R_{j}(z)e^{2it\theta(z)}&1
    \end{array} \right), \quad z \in \Omega_{j},\ j=0\pm,1,2,3,4,\\
          		\left(\begin{array}{cc}
    	1&  -\overline{R}_{j}(z)e^{-2it\theta(z)}\\0&1
    \end{array} \right), \quad z \in \overline{\Omega}_{j},\ j=0\pm,1,2,3,4,\\
    I, \quad { else}.
    	\end{cases}
    \end{align}
and  a new jump contour given by
$$\Sigma^{(4)}=\bigcup\limits_{j=0\pm}(L_j\cup \overline{L}_j)\cup[z_4,z_3]\cup [z_2, z_1].$$
The following  transformation shows a new function
    \begin{equation}\label{trans4}
        M^{(4)}(z)=M^{(3)}(z)R^{(3)}(z),
    \end{equation}
where  $M^{(4)}(z)$ satisfies a hybrid $\bar{\partial}$-RH problem \ref{m4}.

\begin{prob2} \label{m4}
    Find    $M^{(4)}(z)=M^{(4)}(z;x,t)$  with properties
    \begin{itemize}
        \item  $M^{(4)}(z)$ is continuous in $\mathbb{C}\setminus  \Sigma^{(4)} $. See Figure \ref{jumpm4}.
        \item $M^{(4)}(z)$ satisfies the  jump condition
        \begin{equation}
            M^{(4)}_+(z)=M^{(4)}_-(z)V^{(4)}(z),
        \end{equation}
     \end{itemize}
        where
       \begin{align}\label{V3}
       	V^{(4)}(z)=\begin{cases}
           \left( \begin{array}{cc}
    	1& 0\\ r\left( z_j\right)T^2(z_j) e^{2it\theta(z)}&1
    \end{array} \right), \  z \in \Sigma_{j}, \ j=1,2,3,4;\vspace{2mm}\\
    \left(\begin{array}{cc}
    	1&  -\overline{r}\left( z_j\right)T^{-2}(z_j) e^{-2it\theta(z)}\\0&1
    \end{array} \right), \  z \in \overline{\Sigma}_{j},\ j=1,2,3,4;\vspace{2mm}\\
     	\left(	\begin{array}{cc}
     	1& 0\\
     	r\left( z_2\right)T^2(z_2) e^{2it\theta(z) } & 1
     \end{array}\right),\ z\in L_{0+};\vspace{2mm}\\
     \left(	\begin{array}{cc}
     	1& -\overline{r}\left( z_2\right)T^{-2}(z_2) e^{-2it\theta(z) }\\
     	0 & 1
     \end{array}\right),\ z\in \overline{L}_{0+};\vspace{2mm}\\
     	\left(	\begin{array}{cc}
     	1& 0\\
     	-r\left( z_3\right)T^2(z_3) e^{2it\theta(z) } & 1
     \end{array}\right),\ z\in L_{0-};\vspace{2mm}\\
            \left(		\begin{array}{cc}
       	1& \overline{r}\left( z_3\right)T^{-2}(z_3) e^{-2it\theta(z)}\\
       	0 & 1
       \end{array}\right),\  z\in \overline{L}_{0-};\vspace{2mm}\\
       T(z)^{-\sigma_3} V(z)T(z)^{\sigma_3},\quad z\in \mathbb{R}\backslash\Gamma.
       	\end{cases}
       \end{align}
    \begin{itemize}
        \item $M^{(4)}(z)$ has the asymptotic behavior: $ M^{(4)}(z)=I+\mathcal{O}(z^{-1}),	\quad  z \to  \infty.$

        \item For $z\in \mathbb{C}\setminus \Sigma^{(4)} $, we have
        \begin{equation*}
            \bar{\partial}M^{(4)}(z)= M^{(4)}(z) \bar{\partial}R^{(3)}(z),
        \end{equation*}
        where
         \begin{align}\label{parR3}
    	\bar{\partial}R^{(3)}(z)=\begin{cases}
      		\left(\begin{array}{cc}
    	0&  0\\
    -\bar{\partial}R_{j}(z)e^{2it\theta(z)}&0
    \end{array} \right), \quad z \in \Omega_{j},\ j=0\pm,1,2,3,4,\\
          		\left(\begin{array}{cc}
    	0&  -\bar{\partial}\overline{R}_{j}(z)e^{-2it\theta(z)}\\0&0
    \end{array} \right), \quad z \in \overline{\Omega}_{j},\ j=0\pm,1,2,3,4,\\
   0, \quad { else}.
    	\end{cases}
    \end{align}
\end{itemize}
    \end{prob2}

In the next subsection, we decompose $M^{(4)}(z)$ into a pure RH problem  $M^{rhp}(z)$ with $\bar{\partial} R^{(3)}(z)=0$ and a pure $\bar{\partial}$-problem $M^{(5)}(z)$ with $\bar{\partial} R^{(3)} (z)\neq 0$
in the form

\begin{equation}
M^{(4)}(z)=\begin{cases}
M^{rhp}(z)\Longrightarrow\bar{\partial} R^{(3)}(z)=0,\\
M^{(5)}(z) M^{rhp}(z)\Longrightarrow\bar{\partial} R^{(3)} (z)\neq 0.
\end{cases}\label{rhdbar}
\end{equation}

     \begin{figure}[H]
        \begin{center}
  \begin{tikzpicture}[scale=0.8]
    \path[fill=gray!20] (-6,0)rectangle (6,2.5);
\fill [white]  plot [domain=2.5:0,smooth] (4.6+\x*\x/9,\x) -- (6,0) -- (6.1,2.5);
\fill [gray!20]  plot [domain=-2.5:0,smooth] (4.6+\x*\x/9,\x) -- (6,0) -- (6.1,-2.5);
\fill [white]  plot [domain=2.5:0,smooth] (-4.6-\x*\x/9,\x) -- (-6,0) -- (-6.1,2.5);
\fill [gray!20]  plot [domain=-2.5:0,smooth] (-4.6-\x*\x/9,\x) -- (-6,0) -- (-6.1,-2.5);
 \filldraw[gray!20](-1.6,0)--(-2.6,0) arc (180:360:1.3);
   \filldraw[gray!20](-1.0,0)--(-0.0,0) arc (180:360:1.3);
    \filldraw[white](-1.6,0)--(0,0) arc (0:180:1.3);
     \filldraw[white](1.3,0)--(2.6,0) arc (0:180:1.3);

                \draw[dotted,black!40,-latex](-6,0)--(6,0)node[black,right]{$\re z$};
                 \draw [ ] (-4.6, 0)--(-2.6,0);
                                    \draw [ ] (4.6, 0)--(2.6,0);
                                 \draw [-latex ] (-4.6, 0)--(-3.5,0);
                                    \draw [-latex ] (2.6, 0)--(4,0);

                \node[shape=circle,fill=blue, scale=0.15]  at (-3.45,0){0} ;
                      \node[shape=circle,fill=blue, scale=0.15]  at (3.45,0){0} ;
            \node[shape=circle,fill=black,scale=0.15] at (-2.6,0) {0};
              \node[shape=circle,fill=black,scale=0.15] at (2.6,0) {0};
             \node[shape=circle,fill=black,scale=0.15] at (-4.6,0) {0};
             \node[shape=circle,fill=black,scale=0.15] at (4.6,0) {0};
                \node[shape=circle,fill=blue,scale=0.15] at (0.0,0) {0};
                		\coordinate (A) at (-1.3,  0);
                \coordinate (B) at (1.3,  0);
                \coordinate (C) at (-3.45,  0);
                \coordinate (D) at (3.45,  0);

\draw[dotted,blue] (C) circle [radius=1.6];
\draw[dotted,blue] (D) circle [radius=1.6];

                \node[below,blue] at (0.0,0) {\footnotesize $0$};
                 \node[below,blue] at (-3.45,0) {\footnotesize $-1$};
                 \node[below,blue] at (3.45,0) {\footnotesize $1$};

                  \node[below] at (-4.3,0) {\footnotesize $z_4$};
                 \node[below] at (4.4,0) {\footnotesize $z_1$};
                   \node[below] at (-2.6,0) {\footnotesize $z_3$};
                 \node[below] at (2.6,0) {\footnotesize $z_2$};


                     \node[below] at (-0.75,1.2) {\footnotesize $\Sigma_{0-}$};
                 \node[below] at (0.85,1.2) {\footnotesize $\Sigma_{0+}$};
                 \node[below] at (-0.4,-0.4) {\footnotesize $\overline{\Sigma}_{0-}$};
                 \node[below] at (0.62,-0.4) {\footnotesize $\overline{\Sigma}_{0+}$};

                 \node[below] at (-0.75,0.6) {\footnotesize $L_{0-}$};
                 \node[below] at (0.87,0.6) {\footnotesize $L_{0+}$};
                 \node[below] at (-0.8,0.1) {\footnotesize $\overline{L}_{0-}$};
                 \node[below] at (1,-0.) {\footnotesize $\overline{L}_{0+}$};


                   \node[below] at (-1.8,1.2) {\footnotesize $\Sigma_{3}$};
                 \node[below] at (1.8,1.2) {\footnotesize $\Sigma_{2}$};
                 \node[below] at (-1.8,-0.4) {\footnotesize $\overline{\Sigma}_{3}$};
                 \node[below] at (2,-0.4) {\footnotesize $\overline{\Sigma}_{2}$};


                    \node[below] at (-5.4,1.6) {\footnotesize $\Sigma_{4}$};
                 \node[below] at (5.4,1.6) {\footnotesize $\Sigma_{1}$};
                 \node[below] at (-5.4,-1.0) {\footnotesize $\overline{\Sigma}_{4}$};
                 \node[below] at (5.4,-1.0) {\footnotesize $\overline{\Sigma}_{1}$};

                                 \draw [ ] (0.0, 0)--(-1.3,0.9);
                                    \draw [ ] (0.0, 0)--(1.3,0.9);
                                 \draw [ ] (0.0, 0)--(-1.3,-0.9);
                                  \draw [ ] (0.0, 0)--(1.3,-0.9);
                                 \draw[](-2.6,0)--(-1.3,0.9);
                                    \draw[](2.6,0)--(1.3,0.9);
                                  \draw[](-2.6,0)--(-1.3,-0.9);
                                  \draw[](2.6,0)--(1.3,-0.9);
                               \draw [ ] (4.6,0 )--(6,-1.5);
                                 \draw [ ] (4.6,0 )--(6,1.5);
                                 \draw[-latex](4.6,0 )--(5.25,-0.7);
                                       \draw[-latex](4.6,0 )--(5.25,0.7);
                                 \draw [ ] (-4.6,0)--(-6,-1.5);
                                 \draw [ ] (-4.6,0 )--(-6,1.5);
                                 \draw[-latex](-6.0,-1.5 )--(-5.25,-0.7);
                                       \draw[-latex](-6.0,1.5 )--(-5.25,0.7);

                                       \draw[](-1.3,0)--(-1.3,0.9);
                                             \draw[](1.3,0)--(1.3,0.9);
                                              \draw[-latex](-1.3,0)--(-1.3,0.45);
                                             \draw[-latex](1.3,0)--(1.3,0.45);

                                    \draw[](-1.3,0)--(-1.3,-0.9);
                                             \draw[](1.3,0)--(1.3,-0.9);
                                              \draw[-latex](-1.3,0)--(-1.3,-0.45);
                                             \draw[-latex](1.3,0)--(1.3,-0.45);

                            \draw [-latex] (-2.6,0)--(-1.95,0.45);
                              \draw [-latex] (-1.3,0.9)--(-0.65,0.45);
                                  \draw [-latex] (0,0)--(0.65,0.45);
                              \draw [-latex] (1.3,0.9)--(1.95,0.45);
                             \draw [-latex] (-2.6,0)--(-1.95,-0.45);
                             \draw [-latex] (-1.3,-0.9)--(-0.65,-0.45);
                                \draw [-latex] (0,0)--(0.65,-0.45);
                              \draw [-latex] (1.3,-0.9)--(1.95,-0.45);

  \end{tikzpicture}
            \caption{ \footnotesize{The jump contour $\Sigma^{(4)}$. }}
      \label{jumpm4}
        \end{center}
    \end{figure}
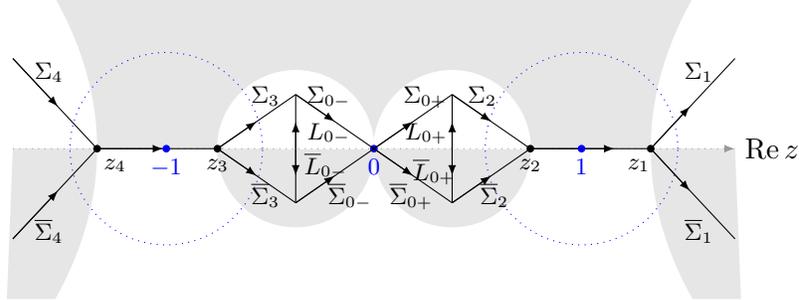

    \subsection{Pure RH problem} \label{modi3}
  \hspace*{\parindent}
In this subsection, we consider the contribution of the pure RH problem to the long-time asymptotics.
The pure RH problem  is given as follows.
	
\begin{prob}\label{mrhp}
    Find  $M^{rhp}(z)=M^{rhp}(z;x,t)$ which satisfies
	  \begin{itemize}
	  	\item  $M^{rhp}(z)$ is analytical in $\mathbb{C}\backslash \Sigma^{(4)}$. See   Figure \ref{jumpm4}.
	  	\item $M^{rhp}(z)$  satisfies the jump condition
\begin{equation*}
	  		M^{rhp}_+(z)=M^{rhp}_-(z)V^{(4)}(z),
	  	\end{equation*}
	  	where $V^{(4)}(z)$ is given by (\ref{V3}).
	  	\item   $M^{rhp}(z)$  has the same asymptotics with  $M^{(4)}(z)$.

	  \end{itemize}
\end{prob}

We consider the local models of $M^{rhp}(z)$ as $z=\pm1$.
 For a fixed parameter $0<\varepsilon <\sqrt{2C}$, we define two open disks in the neighborhood of  $z=\pm1$,
 \begin{align*}
\mathcal{U}_{L} := \{z \in \mathbb{C}: |z+1|< \varepsilon \tau^{-1/3+\varsigma}\},\ \
  \mathcal{U}_{R} := \{z \in \mathbb{C}: |z-1|< \varepsilon \tau^{-1/3+\varsigma}\},
 \end{align*}
the boundaries $\partial\mathcal{U}_{L}$ and $\partial\mathcal{U}_{R}$ with oriented counterclockwise, see Figure \ref{jumpm4}. The parameter $\varsigma$ is a constant and is defined by $0<\varsigma<1/6$ and $\sqrt{2C}\tau^{-1/3+\varsigma}<\rho$ as $t$ is large enough, $\rho$ is defined by (\ref{definerho}).
The saddle points belong to the corresponding disks
  \begin{align*}
  z_1, z_2 \in \mathcal{U}_{R},\ \ \ \  z_3, z_4 \in \mathcal{U}_{L}.
  \end{align*}
In the disks $\mathcal{U}_{L}$ and $\mathcal{U}_{R}$, there are two local RH problems associated with $ M_{L}(z)$ and $ M_{R}(z)$.

In fact, the jump matrix $V^{(4)}$ outside $\mathcal{U}_{L}$ and $\mathcal{U}_{R}$ decays   exponentially and uniformly fast with $t \to +\infty$ into  a identity matrix.
This  enlightens us to reconstruct the solution of $M^{rhp}(z)$ in the following three cases:
 \begin{equation} \label{trans6}
 M^{rhp}(z) = \begin{cases}
  M^{err}(z),\quad z \in \mathbb{C} \setminus \{\mathcal{U}_{L}\cup\mathcal{U}_{R}\},\\
  M^{err}(z) M_{L}(z), \quad z \in \mathcal{U}_{L},\\
 M^{err}(z) M_{R}(z), \quad z \in \mathcal{U}_{R},
 \end{cases}
 \end{equation}
where $ M^{err}(z)$ is an error function  determined in  subsection \ref{snrh} via  a small norm RH problem.
 $M_{L}(z)$ and $M_{R}(z)$ admit  local models  near $z=-1$ and $z=1$, which are approximated by an RH problem associated with Painlev\'{e} II transcendents. We expound the local  model $M_{R}(z)$ as $z=1$, the local model $ M_{L}(z)$ can be derived in a similar way.

The local RH problem of  $M_{R}(z)$ is as follows.
\begin{prob}
    Find  $ M_{R}(z):=M_{R}(z,x,t)$ with properties
	  \begin{itemize}
	  	\item  $M_{R}(z)$ is analytical in $\mathcal{U}_{R} \backslash  {\Sigma}_{R}$, where   $\Sigma_{R}= \Sigma^{(4)} \cap \mathcal{U}_{R}   $.
     See Figure \ref{sign3er}.
	  	\item  $M_{R}(z)$  satisfies the jump condition
\begin{equation*}
	  		M_{R+}(z)=M_{R-}(z) {V}_{R}(z), \ z\in  {\Sigma}_{R},
\end{equation*}
\end{itemize}
where
\begin{align}\label{vp}
        V_{R}(z)= \begin{cases}
       \left( \begin{array}{cc}
       		1& 0\\
       		 r(z_j) T^2(z_j) e^{2it\theta(z) }     & 1
       	\end{array}\right),\  z\in \{\Sigma_{j}\cap\mathcal{U}_{R}\}, \ j=1,2, \\
    \left(	\begin{array}{cc}
       			1& -\overline{r(z_j)} T^{-2}(z_j)e^{-2it\theta(z) }\\
       			0 & 1
       		\end{array}\right) ,\ z \in \{\overline{\Sigma}_{j}\cap\mathcal{U}_{R}\}, \ j=1,2,\\
 T(z)^{-\sigma_3} V(z)T(z)^{\sigma_3},  \   z\in (z_2,  z_1).
       	\end{cases}
       \end{align}

	  \begin{itemize}
         \item   $M_{R}(z)(I+\mathcal{O}(t^{-1/3}))^{-1} \to I, \ t\to +\infty$, uniformly for $z \in \partial \mathcal{U}_{R}.$
	  \end{itemize}
\end{prob}

\subsubsection{Local solvable RH problem}
\hspace*{\parindent}
In the region $\mathcal{D}$,  we have $\xi \to \left(-\frac{3}{2}\right)^-$  as $t \to +\infty$.  From
(\ref{z1234})-(\ref{iz1234}), we found that   the phase points $z_1$ and $z_2$ will merge to $z=1$, while $z_3$ and $z_4$ will merge to $z=-1$.

We seek a polynomial fit for the phase function $t\theta(z)$ as $z=1$ and find that
\begin{align}
t\theta(z)={}&3 t (z - 1)^3+(x+3 t) (z -1)+\frac{1}{2}x\left[\sum_{n=2}^{\infty}(-1)^{n-1}(z-1)^n\right]\notag\\
&+\frac{1}{4}t\left[3(z-1)^2-11(z-1)^3+\sum_{n=2}^{\infty}(-1)^{n-1}\frac{n^2+3n+8}{2}(z-1)^n\right]\notag\\
  \xlongequal{z\mapsto\hat{k} }{}&\frac{4}{3}\hat{k}^3+s\hat{k}+\mathcal{O}(t^{-1/3}\hat{k}^2),\label{z1}
 \end{align}
where
\begin{align}
  \hat{k}  = \tau^{\frac{1}{3}} (z-1),\ \
  s   = \frac{8}{9} \left(\xi+\frac{3}{2}\right)\tau^{\frac{2}{3}},\ \
  \tau=\frac{9}{4}t.\label{hatk}
\end{align}
In the same way, we find that the phase function $t\theta(z)$ has the same property as $z=-1$. It is given by
\begin{align}
t\theta(z)={}&3 t (z + 1)^3+(x+3 t) (z + 1)+\frac{1}{2}x\left[\sum_{n=2}^{\infty}(z+1)^n\right]\notag\\
&+\frac{1}{4}t\left[-3(z+1)^2-11(z+1)^3+\sum_{n=2}^{\infty}\frac{n^2+3n+8}{2}(z+1)^n\right]\notag\\
 \xlongequal{z\mapsto\check{k} }{}&\frac{4}{3}\check{k}^3+s\check{k}+\mathcal{O}(t^{-1/3}\check{k}^2),\label{z-1}
 \end{align}
where $s, \tau$ are defined by (\ref{hatk}), and
\begin{equation}
  \check{k}  = \tau^{\frac{1}{3}} (z+1).\label{checkk}
\end{equation}

Moreover,
the new phase points $\hat{k}_j$ and $\check{k}_j$ with the properties:
\begin{proposition} \label{opow}
 In the transition region $\mathcal{D}$ (\ref{D}),  the scaling transformations  (\ref{hatk}) and  (\ref{checkk}) lead to four new phase points  $\hat{k}_j= \tau^{1/3} (  z_j -1  ),\ j=1,2$ and $\check{k}_j= \tau^{1/3} (  z_j+1 ), \ j=3,4$ which are
always held by a constant
  \begin{align*}
  &\left|\hat{k}_{1}\right|< 2(9/4)^{1/3} \sqrt{2C},\ \ \left|\hat{k}_{2}\right|< (9/4)^{1/3} \sqrt{2C},\\
  &\left|\check{k}_{3}\right|< (9/4)^{1/3} \sqrt{2C},\ \ \left|\check{k}_{4}\right|< 2(9/4)^{1/3} \sqrt{2C}.
        \end{align*}
\end{proposition}

\begin{proof}
  From (\ref{2theta}), it can be seen that the phase points $z_j, j=1,2$ satisfy the equation
   \begin{equation*}
   3(\chi^2_j-2)+4\xi=0,
   \end{equation*}
   where $\chi_j=z_j+z_j^{-1}$ and $\xi$ satisfy the relation
   \begin{equation*}
   \chi_j=-\frac{4}{3}\xi+2.
   \end{equation*}
For the reason that $z_1, z_2\to 1$ as $t\to\infty$, we take a large $t$ such that $\frac{1}{2}<z_2<1$.
Furthermore, we have
\begin{align}
&z_2^2 +1 =\chi_2 z_2  \Longrightarrow    ( z_2 -1 )^2=(\chi_2-2)z_2. \label{por4}
\end{align}
Noting that $\chi_2 -2 >0$ and $0<z_2 <1$, we have
\begin{align}
& (\chi_2-2) z_2 <  \chi_2 -2  =   -\frac{4}{3}\xi
 <  -2\left(\xi+\frac{3}{2}\right).\label{23236}
\end{align}
Substituting (\ref{23236}) into (\ref{por4}) yields
\begin{align}
&( z_2 -1 )^2<  -2\left(\xi+\frac{3}{2}\right) <  2Ct^{-2/3},\label{por32}
\end{align}
 which implies that
\begin{align}
&     \big|\hat{k}_2 \big|= \big|\tau^{1/3} (  z_2- 1  ) \big| <(9/4)^{1/3} \sqrt{2C}.        \label{por33}
\end{align}
Recalling that the  symmetry $ z_1z_2=1$ which implies that  $z_1>1$.     We have
 \begin{align*}
z_1 -1 =  \frac{1}{z_2} -1    =  \frac{1- z_2}{z_2}.
\end{align*}
Using (\ref{por33}), we derive $\hat{k}_1$ can be controlled  by a constant
 \begin{align*}
\big| \hat{k}_1 \big|= \big|\tau^{1/3} (  z_1 -1   ) \big|= \left|  \frac{ \tau^{1/3} (1-z_2)}{z_2}\right|< 2(9/4)^{1/3} \sqrt{2C}.
\end{align*}
From (\ref{z1234}), we have $z_3=-z_2$ and $z_{4}=-z_1$,
the estimate of $\left|\check{k}_j\right|$ can be derived.
\end{proof}

 We define two open disks in the neighborhood of $\hat{k}=0$ and $\check{k}=0$,
 \begin{align*}
\check{\mathcal{U}}_{L} := \{\check{k} \in \mathbb{C}: |\check{k}|< \varepsilon\tau^{\varsigma} \},\ \
 \hat{\mathcal{U}}_{R} := \{\hat{k} \in \mathbb{C}: |\hat{k}|< \varepsilon\tau^{\varsigma} \},
 \end{align*}
which  are oriented counterclockwise, respectively.
There are new saddle points $\hat{k}_j, j=1,2$ and $\check{k}_j, j=3,4$ under the maps $z\mapsto\hat{k}$ (\ref{z1}) and $z\mapsto\check{k}$ (\ref{z-1}), respectively, which belong to the corresponding disks
  \begin{align*}
  \hat{k}_1, \hat{k}_2 \in \hat{\mathcal{U}}_{R},\ \ \ \  \check{k}_3, \check{k}_4 \in \check{\mathcal{U}}_{L}.
  \end{align*}
  The transformation (\ref{hatk}) and (\ref{checkk}) map  $\mathcal{U}_{R}$ into $\hat{\mathcal{U}}_{R}$ and  map $\mathcal{U}_{L}$ into $\check{\mathcal{U}}_{L}$, see Figure \ref{jumpm4}.

   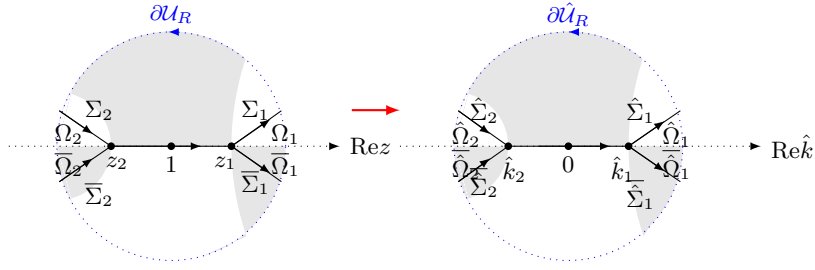
\begin{figure}
        \begin{center}
  \begin{tikzpicture}[scale=0.8]
      \filldraw[gray!20](-3.3,0)--(-1.4,0) arc (0:180:1.9);
       \filldraw[gray!20](3.3,0)--(5.2,0) arc (0:180:1.9);

\fill [white]  plot [domain=2.5:0,smooth] (4.3+\x*\x/9,\x) -- (6,0) -- (6.1,2.5);
\fill [gray!20]  plot [domain=-1.5:0,smooth] (4.3+\x*\x/9,\x)  -- (5.2,0)arc (0:-45:1.9);
\fill [white]  plot [domain=2.5:0,smooth] (-2.3+\x*\x/9,\x) -- (-1.4,0) -- (-1.5,2.5);
\fill [gray!20]  plot [domain=-1.5:0,smooth] (-2.3+\x*\x/9,\x)  -- (-1.4,0)arc (0:-75:1.9);
\fill[white](-5.2,0)--(-4.3,0)arc(0:75:0.9)--(-5.2,0)arc(180:150:1.9);
      \fill[gray!20](-5.2,0)--(-4.3,0) arc (0:-75:0.9)--(-5.2,0)arc(180:210:1.9);
  \fill[gray!20](1.4,0)--(2.3,0) arc (0:-75:0.9)--(1.4,0)arc(180:210:1.9);
 \fill[white](1.4,0)--(2.3,0) arc (0:75:0.9)--(1.4,0)arc(180:150:1.9);

                \draw[dotted,-latex](-6,0)--(-0.5,0)node[black,right]{\footnotesize Re$z$};
                \node[shape=circle,fill=black, scale=0.15]  at (-3.3,0){0} ;
            \node[shape=circle,fill=black,scale=0.15] at (-2.3,0) {0};
             \node[shape=circle,fill=black,scale=0.15] at (-4.3,0) {0};
\draw[blue,-latex](-3.3,1.9)--(-3.35,1.9);
                 \node[below] at (-3.3,0) {\footnotesize $1$};
                  \draw [dotted,blue](-3.3,0) circle (1.9);
                   \draw [-latex] (-2.3,0)--(-1.65,0.45);
                                 \draw[](-2.3,0)--(-1.47,0.59);
                      \draw [ ] (-4.3, 0)--(-5.16, 0.59);
                           \draw [-latex] (-5.16, 0.59)--(-4.6,0.2);

    \draw [ ] (-2.3,0)--(-1.47, -0.59);
          \draw [-latex] (-2.3,0)--(-1.65,-0.45);
   \draw [] (-4.3, 0)--(-5.16, -0.59);
       \draw [] (-4.3, 0)--(-2.3, 0);
    \draw [-latex](-4.3, 0)--(-2.8,0);
                \draw [-latex]  (-5.16, -0.59)--(-4.6, -0.2);
\node[blue]  at (-3.3,2.2) {\scriptsize $\partial \mathcal{U}_{R}$};

 \node[below] at (-4.2,0) {\footnotesize $z_2$};
                                         \node[below] at (-2.4,0) {\footnotesize $z_1$};
                                         \node  at (-4.5,0.6) {\footnotesize $\Sigma_{2}$};
                                         \node  at (-4.5,-0.8) {\footnotesize $\overline{\Sigma}_{2}$};
                                          \node  at (-1.9,0.6) {\footnotesize $\Sigma_{1}$};
                                         \node  at (-1.9,-0.6) {\footnotesize $\overline{\Sigma}_{1}$};

                                         \node at (-1.4,0.2) {\footnotesize $\Omega_1$};
                                          \node at (-1.4,-0.3) {\footnotesize $\overline{\Omega}_1$};
                                          \node at (-5.0,0.2) {\footnotesize $\Omega_2$};
                                          \node at (-5.0,-0.3) {\footnotesize $\overline{\Omega}_2$};

\draw[-latex,thick,red] (-0.3,0.6)--(0.5,0.6);

   \draw[dotted,-latex](0.5,0)--(6.5,0)node[black,right]{\footnotesize Re$\hat{k}$};
                \node[shape=circle,fill=black, scale=0.15]  at (3.3,0){0} ;
            \node[shape=circle,fill=black,scale=0.15] at (2.3,0) {0};
             \node[shape=circle,fill=black,scale=0.15] at (4.3,0) {0};

\draw[blue,-latex](3.35,1.9)--(3.3,1.9);
                 \node[below] at (3.3,0) {\footnotesize $0$};
                  \draw [dotted,blue](3.3,0) circle (1.9);
                   \draw [-latex] (1.47,0.59)--(2.0,0.2);
                                 \draw[](2.3,0)--(1.47,0.59);
                                  \draw [ ] (4.3, 0)--(5.16, 0.59);
                           \draw [-latex] (4.3, 0)--(4.9,0.42);

    \draw [ ] (2.3,0)--(1.47, -0.59);
          \draw [-latex] (1.47,-0.59)--(2,-0.2);
   \draw [] (4.3, 0)--(5.16, -0.59);
       \draw [] (4.3, 0)--(2.3, 0);
    \draw [-latex](2.8,0)--(4, 0);
                \draw [-latex] (4.3, 0)--(4.9,-0.42);
\node[blue]  at (3.3,2.2) {\scriptsize $\partial \hat{\mathcal{U}}_{R}$};

 \node[below] at (4.2,0) {\footnotesize $\hat{k}_1$};
                                         \node[below] at (2.4,0) {\footnotesize $\hat{k}_2$};
                                         \node  at (4.5,0.6) {\footnotesize $\hat{\Sigma}_{1}$};
                                          \node  at (5.1,0.2) {\footnotesize $\hat{\Omega}_{1}$};
                                           \node  at (5.1,-0.3) {\footnotesize $\overline{\hat{\Omega}}_{1}$};

                                         \node  at (4.5,-0.8) {\footnotesize $\overline{\hat{\Sigma}}_{1}$};

                                           \node  at (1.9,0.6) {\footnotesize $\hat{\Sigma}_{2}$};

                                             \node  at (1.6,0.2) {\footnotesize $\hat{\Omega}_{2}$};
                                           \node  at (1.6,-0.3) {\footnotesize $\overline{\hat{\Omega}}_{2}$};

                                             \node  at (1.9,-0.6) {\footnotesize $\overline{\hat{\Sigma}}_{2}$};
  \end{tikzpicture}
            \caption{\footnotesize{ The jump contours of  $\Sigma_{R} $ (left) and $\hat{\Sigma}_{R} $ (right).   }}
      \label{sign3er}
        \end{center}
    \end{figure}
Next, we show that the local RH problem $M_{R}(z)$ can be approximated by the solution of
 a model RH problem for Painlev\'{e}-II equation. From (\ref{hatk}), we have
\begin{equation*}
 z = \tau^{-1/3} \hat{k}+1,
\end{equation*}
and let
$
\widetilde{r}(z): =r(z)T^2(z)=r(\tau^{-1/3} \hat{k} +1)T^2(\tau^{-1/3} \hat{k}+1).
$
In particular, we have $\widetilde{r}(1) =r(1)T^{2}(1)$ and derive the following estimates of $\re\left[2i\left(\frac{4}{3}\hat{k}^3+s\hat{k}\right)\right]$.

\begin{lemma}\label{rehatk}
For $t\rightarrow \infty$, we have the following estimates
\begin{align}
&\re\left[2i\left(\frac{4}{3}\hat{k}^3+s\hat{k}\right)\right]\leq-4\hat{u}^2\hat{v},\ \ \hat{k}\in\hat{\Omega}_j,\ \ j=1,2,\label{rehatk1}\\
&\re\left[2i\left(\frac{4}{3}\hat{k}^3+s\hat{k}\right)\right]\geq4\hat{u}^2\hat{v},\ \ \hat{k}\in\overline{\hat{\Omega}}_j,\ \ j=1,2,\label{rehatk2}
\end{align}
where $\hat{k} = \hat{k}_1 +\hat{u}+i\hat{v}$.
\begin{proof}
Let  $z=z_1+u+iv$. For the reason that $\arg(z-z_1)<\pi/4$, we have $v<u$. Then
\begin{align*}
\re\left[2i\left(\frac{4}{3}\hat{k}^3+s\hat{k}\right)\right]&=\re\left(2i\left(3t(z-1)^3+2(x+3t)(z-1)\right)\right)\\
&\leq-12tu^2v+3tv\omega(z_1),
\end{align*}
where
\begin{equation*}
\omega(z_1)=-\left(\left( (-1 + z_1)  \left(1 + z_1 + 5 z_1^3 + z_1^2 (-7 + 12 u)\right) \right)/z_1^2\right),
\end{equation*}
and $\omega'(z_1)\leq0$. We derive $\omega(z_1)\leq\omega(1)=0$. We derive (\ref{rehatk1}). The next (\ref{rehatk2}) can be derived using a similar method.
\end{proof}
\end{lemma}

 \begin{proposition} \label{ppe}  For $j=1,2$, we have the following estimates
\begin{align}\label{e1}
 & \Big| \widetilde{r}\left(z\right)   e^{2it\theta \left(z\right)}- r_0 e^{8i\hat{k}^3/3+2is\hat{k} }  \Big|   \lesssim t^{-1/3+2\varsigma}, \ \hat{k} \in \left(\hat{k}_2,  \hat{k}_1\right),\\
 &  \Big| \widetilde{r}\left(z_j \right)   e^{2it\theta \left(z\right)}-  r_0 e^{8i\hat{k}^3/3+2is\hat{k} }  \Big|   \lesssim t^{-1/3+2\varsigma},    \ \hat{k} \in  \hat{\Sigma}_{R}\setminus [\hat{k}_2,\hat{k}_1],
\label{e2}
\end{align}
where  $r_0= \widetilde{r}(1)$.
\end{proposition}
\begin{proof}
For $\hat{k} \in \left(\hat{k}_2, \hat{k}_1\right)$, $\hat{k}$ is real and $z\in(z_2,z_1)$, we have
\begin{equation}
\left| \widetilde{r} (z)    e^{2i t   \theta \left(z\right) } - r_0  e^{  8i\hat{k}^3/3+ 2is \hat{k}  }    \right| \leq
\left|  \widetilde{r} (z) - \widetilde{r}(1) \right|+\left|\widetilde{r}(1)\right|\left|e^{2i t   \theta \left(z\right)}-e^{  8i\hat{k}^3/3+ 2is \hat{k}  }\right|.\label{poe1}
 \end{equation}
 For the reason that $T(z)$ is analytic in $\partial\mathcal{U}_{R}$, we have
 \begin{equation*}
 \parallel\widetilde{r}(z)\parallel_{H^1([z_2,z_1])}=\parallel r(z)\parallel_{H^1([z_2,z_1])},
 \end{equation*}
and using the H\"{o}lder's  inequality, we derive the following estimate
\begin{align*}
\left|  \widetilde{r} (z) - \widetilde{r}(1) \right|=\left|\int_{1}^z\widetilde{r}'(s)\mathrm{d}s\right|&\leq\parallel\widetilde{r}'\parallel_{L^2([z_2,z_1])}\left|z-1\right|^{1/2}t^{-1/6}\\
&\leq\parallel r\parallel_{H^1([z_2,z_1])}\left|\hat{k}\right|^{1/2}t^{-1/3+\varsigma/2}\\
&\lesssim t^{-1/3+\varsigma/2}.
\end{align*}
Moreover,
\begin{align*}
\left|e^{2i t   \theta \left(z\right)}-e^{  8i\hat{k}^3/3+ 2is \hat{k}  }\right|=\left|e^{i\mathcal{O}(t^{-1/3}\hat{k}^2)}-1\right|\leq e^{|\mathcal{O}(t^{-1/3}\hat{k}^2)|}-1\lesssim t^{-1/3+2\varsigma}.
\end{align*}
From Lemma \ref{rehatk}, it follows that
$ \left| e^{2i \left( \frac{4}{3} \hat{k}^3 +  s \hat{k} \right)}\right|$ is bounded.
As in the above proof, we can obtain the estimate (\ref{e2}). The estimate for the other jump line can be given in a similar way.
\end{proof}

Now, we find a new RH problem associated with $\hat{M}_{R}(\hat{k})$ in the disk $\hat{\mathcal{U}}_R$:
\begin{prob}
    Find  $ \hat{M}_{R}(\hat{k})=\hat{M}_{R}(\hat{k},x,t)$ with properties
	  \begin{itemize}
	  	\item  $\hat{M}_{R}(\hat{k})$ is analytical in $\hat{\mathcal{U}}_{ R} \backslash  \hat{\Sigma}_{R}$, where   $\hat{\Sigma}_{R}=\hat{\Sigma}_j\cap\hat{\mathcal{U}}_R$, see Figure \ref{sign3er}.
	  	\item  $\hat{M}_{R}(\hat{k})$  satisfies the jump condition
\begin{equation*}
	  		\hat{M}_{R+}(\hat{k})=\hat{M}_{R-}(\hat{k}) \hat{V}_{R}(\hat{k}), \ z\in  \hat{\Sigma}_{R},
\end{equation*}
\end{itemize}
where
\begin{align}\label{vp}
        \hat{V}_{R}(\hat{k})= \begin{cases}
       \left( \begin{array}{cc}
       		1& 0\\
       		 \tilde{r}(1) e^{2i\left(\frac{4}{3}\hat{k}^3+s\hat{k}\right) }     & 1
       	\end{array}\right):=B_+,\  \hat{k}\in \hat{\Sigma}_{j}, \ j=1,2; \\
    \left(	\begin{array}{cc}
       			1& -\overline{\tilde{r}(1)}e^{-2i\left(\frac{4}{3}\hat{k}^3+s\hat{k}\right) }\\
       			0 & 1
       		\end{array}\right):=B_-^{-1},\ \hat{k} \in \overline{\hat{\Sigma}}_{j}, \ j=1,2;\\
 B_-^{-1}B_+,  \   \hat{k}\in (\hat{k}_2,  \hat{k}_1).
       	\end{cases}
       \end{align}

	  \begin{itemize}
         \item   $\hat{M}_R(\hat{k}) \to I$, as $\hat{k}\to +\infty.$
	  \end{itemize}
\end{prob}
Define $\Lambda(\hat{k}):=M_R(\hat{k})\hat{M}_R^{-1}(\hat{k})$. The estimate in Proposition \ref{ppe} implies that $\Lambda(\hat{k})$ has the following asymptotic behavior.
\begin{lemma}
In the disk $\hat{\mathcal{U}}_R$, $\Lambda(\hat{k})$ admits the asymptotic behavior for a large $t$
\begin{align}
\Lambda(\hat{k})=I+\mathcal{O}(t^{-1/3+2\varsigma}), \ \ t\to+\infty.\label{lambda}
\end{align}
\end{lemma}
Moreover, $\Lambda(\hat{k})$ satisfies a RH problem:
\begin{prob}
Find $\Lambda(\hat{k}):=\Lambda(\hat{k},x,t)$ with properties:
  \begin{itemize}
	  	\item $\Lambda(\hat{k})$ is analytical in $\mathbb{C}\backslash\hat{\Sigma}_R$.
\item $\Lambda(\hat{k})$ satisfies the following jump condition
\begin{align*}
\Lambda_+(\hat{k})=\Lambda_-(\hat{k})V_{\Lambda}(\hat{k}),
\end{align*}
where
\begin{align}
V_{\Lambda}(\hat{k})=\hat{M}_{R,-}(\hat{k})V_{R}(\hat{k})\hat{V}^{-1}_{R}(\hat{k})\hat{M}^{-1}_{R,-}(\hat{k}).\label{vlambda}
\end{align}
\end{itemize}
\end{prob}
We derive the following  corollary \ref{c37} from (\ref{lambda}) and (\ref{vlambda}).
\begin{corollary}\label{c37}
In the transition region $\mathcal{D}$ (\ref{D}), as $t\to\infty$, we have
\begin{align}
&V_R \left(z\right) =  \hat{V}_R(\hat{k})+\mathcal{O}(t^{-1/3+2\varsigma}), \ \ \hat{k}\in \hat{\mathcal{U}}_{R},\nonumber\\
&M_R \left(z\right) =  \hat{M}_R(\hat{k})+\mathcal{O}(t^{-1/3+2\varsigma}), \ \ \hat{k}\in \hat{\mathcal{U}}_{R},\nonumber
\end{align}
\end{corollary}

In order to establish a connection  with Painlev\'{e}-II RH model,
we add four auxiliary contours $\mathcal{L}_j,~j=1,2,3,4$ which pass through the point $\hat{k}=0$. And the angle of $\mathcal{L}_j$ with the real axis is $\pi/3$. $\mathcal{L}_j$ together with the original contours $\hat{\Sigma}_j$ divide the complex plane $\mathbb{C}$ into eight regions $\Omega^\flat_j,~j=1,\cdots,8$, see Figure \ref{conlj}.

\begin{figure}[H]
\begin{center}
\begin{tikzpicture}\label{Fig3}
	\node    at (- 1.2, 0.3)  {\scriptsize $\varphi$};
	\draw [dotted ](-6.5,0)--(-0.5,0);
\draw [](-5,0)--(-2,0);
 \draw [  ](-5,0)--(-6.5,1.5);
 \draw [   -latex](-5,0)--(-4,0);
\draw [  -latex](-4,0)--(-2.7,0);
  \draw [ -latex ] (-6.5,1.5)--(-5.75,1.5/2);
 \draw [  ](-5,0)--(-6.5,-1.5);
   \draw [ -latex ] (-6.5,-1.5)--(-5.75,-1.5/2);
  \draw [ ](-2,0)--(-0.5,1.5 );
     \draw [ -latex ] (-2,0)--(-1.25, 1.5/2);
 \draw[ ](-2,0)--(-0.5,-1.5);
      \draw [ -latex ] (-2,0)--(-1.25,- 1.5/2);

 \draw[dotted ](-4.5,-1)--(-2.5, 1 );
 \draw[dotted,-latex ](-4.5,-1)--(-4, -0.5 );
 \draw[dotted,-latex ](-3.5,0)--(-3,  0.5 );
  \draw[dotted ](-4.5, 1)--(-2.5,-1 );
   \draw[dotted,-latex ](-4.5, 1)--(-4,  0.5 );
   \draw[dotted,
   -latex ](-3.5,0)--(-3, -0.5 );

  \node[shape=circle,fill=black, scale=0.15]  at (-2.0,0){0} ;
  \node[shape=circle,fill=black, scale=0.15]  at (-5,0){0} ;
 \node    at (-2.0, -0.2)  {\scriptsize $\hat k_1$};
  \node    at (-4.8, -0.2)  {\scriptsize $\hat k_2$};
  \node    at (-1.4, -0.2)  {\scriptsize $\Omega^\flat_1$};
    \node    at (-2.3, 0.5)  {\scriptsize $\Omega^\flat_2$};
 \node    at (-3.5,0.6)  {\scriptsize $\Omega^\flat_3$};
 \node    at (-4.7, 0.5)  {\scriptsize $\Omega^\flat_4$};
 \node    at (-5.6, -0.2)  {\scriptsize $\Omega^\flat_5$};
  \node    at (-4.7,-0.5)  {\scriptsize $\Omega^\flat_6$};
  \node    at (-3.5, -0.6)  {\scriptsize $\Omega^\flat_7$};
    \node    at (-2.3, -0.5)  {\scriptsize $\Omega^\flat_8$};
     \node  [below]  at (-1,1.8) {\scriptsize $\hat{\Sigma}_1$};
    \node  [below]  at (-2.3,1.5) {\scriptsize $\mathcal{L}_1$};
        \node  [below]  at (-4.6,1.5) {\scriptsize $\mathcal{L}_2$};
           \node  [below]  at (-6,1.8) {\scriptsize $\hat\Sigma_2$};
        \node  [below]  at (-6,-1.4) {\scriptsize $\overline{\hat\Sigma}_2$};
 \node  [below]  at (-4.8,-0.9) {\scriptsize $\mathcal{L}_3$};
  \node  [below]  at (-2.2,-0.9) {\scriptsize $\mathcal{L}_4$};
      \node  [below]  at (-1,-1.4) {\scriptsize $\overline{\hat\Sigma}_1$};

	\draw [dotted ](1.5,0)--(4.5,0);
\draw[    ](1.5,-1.5)--(4.5,1.5);
\draw[-latex ](1.5,-1.5)--(2.25,-0.75);
\draw[-latex ](3, 0)--(3.75, 0.75);
\draw[  ](1.5,1.5)--(4.5,-1.5);
\draw[-latex ](1.5, 1.5)--(2.25, 0.75);
\draw[ -latex ](3, 0)--(3.75,-0.75);
\node  [below]  at (1.2,1.8) {\scriptsize $\mathcal{L}_2$};
\node  [below]  at (4.9,1.8) {\scriptsize $\mathcal{L}_1$};
\node  [below]  at (4.9,-1.1) {\scriptsize $\mathcal{L}_4$};
\node  [below]  at (1.2,-1.1) {\scriptsize $\mathcal{L}_3$};
\node    at (3,-0.3)  {$0$};
 \draw [red, thick,-latex  ](-0.3,0)--(1,0);

		\node    at (3.8, 0.3)  {\scriptsize $\pi/4$};
\end{tikzpicture}
\end{center}
\caption {\footnotesize   In the left figure,  we add four  auxiliary contours on  the jump contour  for  $\hat{M}_{R}(\hat{k})$, which  can be transformed
into the  Painlev\'e-II model $  M_{P}(\hat{k})$  with  the jump contour  of   four rays  (the right one).}
\label{conlj}
\end{figure}
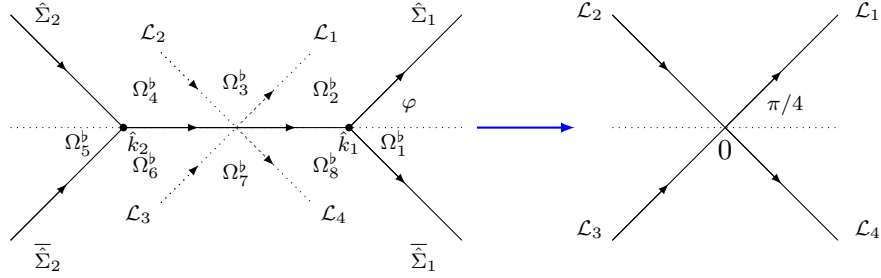

We define a  matrix-valued function
\begin{equation*}
P(\hat{k})=\begin{cases}
B_+^{-1},\quad \hat{k}\in\Omega^\flat_2\cup\Omega^\flat_4,\\
B_-^{-1},\quad \hat{k}\in\Omega^\flat_6\cup\Omega^\flat_8,\\
I,\quad else,
\end{cases}
\end{equation*}
and $\hat{M}_R(\hat{k})$ performs the following transformation
\begin{equation*}
\mathcal{M}_R(\hat{k})=\hat{M}_R(\hat{k})P(\hat{k}),
\end{equation*}
where $\mathcal{M}_R(\hat{k})$ admits the following RH problem.
\begin{prob}
    Find  $ \mathcal{M}_R(\hat{k})=\mathcal{M}_R(\hat{k},s)$ with properties
	  \begin{itemize}
	  	\item  $\mathcal{M}_R(\hat{k})$ is analytical in $\mathbb{C}\backslash\Sigma_P$, where   $\Sigma_P=\bigcup\limits_{j=1}^{4}\mathcal{L}_j$.
     	  	\item  $\mathcal{M}_R(\hat{k})$  satisfies the jump condition
\begin{equation*}
	  		\mathcal{M}_{R,+}(\hat{k})=\mathcal{M}_{R,-}(\hat{k})\mathcal{ V}_{R}(\hat{k}), \ \hat{k}\in  \Sigma_{P},
\end{equation*}
where
\begin{equation*}
\mathcal{V}_P(\hat{k})=\begin{cases}
B_+,\quad \hat{k}\in \mathcal{L}_j,~j=1,2,\\
B_-^{-1},\quad \hat{k}\in \mathcal{L}_j,~j=3,4.
\end{cases}
\end{equation*}
\item $\mathcal{M}_R(\hat{k})\to I,\ \ \hat{k}\to\infty.$
\end{itemize}
\end{prob}

Let  $\tilde{r}(1)=|\tilde{r}(1)|e^{i\varphi_0}$, where $\varphi_0=\arg{\tilde{r}(1)}$. We note that the solution $\mathcal{M}_R(\hat{k})$ can be expressed by the Painlev\'{e}-II model
\begin{equation}
\mathcal{M}_R(\hat{k})=\sigma_1e^{-\frac{i(\varphi_0+\pi/2)}{2}\widehat{\sigma}_3}M_P(\hat{k})\sigma_1,\label{379}
\end{equation}
where $M_P(\hat{k})$ satisfies a standard Painlev\'{e}-II model given by Appendix \ref{appx} as $p=|\tilde{r}(1)|$. Substituting (\ref{379}) into Corollary \ref{c37}, we obtain
\begin{equation}
M_{R,1}(s)=\mathcal{M}_{R,1}(s)+\mathcal{O}(t^{-1/3+2\varsigma})=\sigma_1e^{-\frac{i(\varphi_0+\pi/2)}{2}\widehat{\sigma}_3}M_{P,1}(s)\sigma_1+\mathcal{O}(t^{-1/3+2\varsigma}),\label{mr}
\end{equation}
where the subscript ``1"  of $M_{R,1}(s)$ denotes the coefficient of the Taylor expansion of the term $\hat{k}^{-1}$, and $M_{P,1}(s)$ is given by Appendix \ref{appx}. So the solution  of $M_R(z)$ as $t\to\infty$ is given by
\begin{equation*}
M_R(z)=I+\frac{M_{R,1}(s)}{\tau^{1/3}(z-1)}+\mathcal{O}(t^{-2/3+2\varsigma}).
\end{equation*}
For $z=-1$,  $M_L(z)$ is similarly given by
\begin{equation}
M_L(z)=I+\frac{M_{L,1}(s)}{\tau^{1/3}(z+1)}+\mathcal{O}(t^{-2/3+2\varsigma}),\ \  t\to\infty,\label{ml}
\end{equation}
where
\begin{equation*}
M_{L,1}(s)=\sigma_1e^{-\frac{i(\varphi_0+\pi/2)}{2}\widehat{\sigma}_3}M_{P,1}(s)\sigma_1+\mathcal{O}(t^{-1/3+2\varsigma}),
\end{equation*}
with $\varphi_0=\arg \tilde{r}(-1)$.


\subsubsection{A small norm RH problem for error function}\label{snrh}
\hspace*{\parindent}
In this subsection, we consider the error function  $ M^{err}(z)$ defined by  (\ref{trans6}) and
  $ M^{err}(z)$ admits the following  RH problem.

\begin{prob}\label{iew}
  Find      $ M^{err}(z)$ with the properties
          \begin{itemize}
          	\item   $ M^{err}(z)$ is analytical in $\mathbb{C}\backslash  \Sigma^{err}$,
          where $\Sigma^{err} = \{\partial \mathcal{U}_{R}\cup\partial \mathcal{U}_{L} \}\cup \{\Sigma^{(4)} \backslash \{\mathcal{U}_{R}\cup\mathcal{U}_{L}\}\}$, see Figure \ref{jumpe}.
          	\item  $ M^{err}(z)$  satisfies the jump condition
          	\begin{align*}
          		M^{err}_{+}(z)=M^{err}_-(z)V^{err}(z), \quad z\in \Sigma^{err},
          	\end{align*}
          	where $V^{err}(z)$ is given by
   \begin{align}
                V^{err}(z)= \begin{cases}
                   V^{(4)}(z),\quad z \in \Sigma^{(4)} \backslash  \{\mathcal{U}_{R}\cup\mathcal{U}_{L}\},\\
                  M_{R}^{-1}(z) , \quad z \in \partial \mathcal{U}_{R}, \\
                  M_{L}^{-1}(z) , \quad z \in \partial \mathcal{U}_{L}. \label{ioei}
                                  \end{cases}
            \end{align}

          	\item       $M^{err}(z)=I+\mathcal{O}(z^{-1}),	\quad  z\to  \infty.$	

          \end{itemize}
\end{prob}

 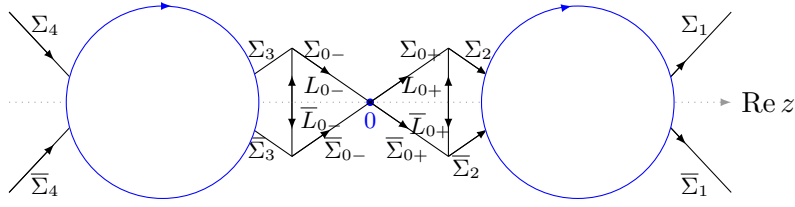
\begin{figure}
        \begin{center}
  \begin{tikzpicture}[scale=0.8]


\draw[dotted,black!40,-latex](-6,0)--(6,0)node[black,right]{$\re z$};
                 \draw [ ] (-4.6, 0)--(-2.6,0);
                                    \draw [ ] (4.6, 0)--(2.6,0);
                                 \draw [-latex ] (-4.6, 0)--(-3.5,0);
                                    \draw [-latex ] (2.6, 0)--(4,0);

     \node[shape=circle,fill=blue, scale=0.15]  at (-3.45,0){0} ;
                      \node[shape=circle,fill=blue, scale=0.15]  at (3.45,0){0} ;
            \node[shape=circle,fill=black,scale=0.15] at (-2.6,0) {0};
              \node[shape=circle,fill=black,scale=0.15] at (2.6,0) {0};
             \node[shape=circle,fill=black,scale=0.15] at (-4.6,0) {0};
             \node[shape=circle,fill=black,scale=0.15] at (4.6,0) {0};
                \node[shape=circle,fill=blue,scale=0.15] at (0.0,0) {0};
                \coordinate (C) at (-3.45,  0);
                \coordinate (D) at (3.45,  0);

                \node[below,blue] at (0.0,0) {\footnotesize $0$};
                 \node[below,blue] at (-3.45,0) {\footnotesize $-1$};
                 \node[below,blue] at (3.45,0) {\footnotesize $1$};

                  \node[below] at (-4.3,0) {\footnotesize $z_4$};
                 \node[below] at (4.4,0) {\footnotesize $z_1$};
                   \node[below] at (-2.6,0) {\footnotesize $z_3$};
                 \node[below] at (2.6,0) {\footnotesize $z_2$};


                     \node[below] at (-0.75,1.2) {\footnotesize $\Sigma_{0-}$};
                 \node[below] at (0.85,1.2) {\footnotesize $\Sigma_{0+}$};
                 \node[below] at (-0.4,-0.4) {\footnotesize $\overline{\Sigma}_{0-}$};
                 \node[below] at (0.62,-0.4) {\footnotesize $\overline{\Sigma}_{0+}$};

                 \node[below] at (-0.75,0.6) {\footnotesize $L_{0-}$};
                 \node[below] at (0.87,0.6) {\footnotesize $L_{0+}$};
                 \node[below] at (-0.8,0.1) {\footnotesize $\overline{L}_{0-}$};
                 \node[below] at (1,-0.) {\footnotesize $\overline{L}_{0+}$};


                   \node[below] at (-1.8,1.2) {\footnotesize $\Sigma_{3}$};
                 \node[below] at (1.8,1.2) {\footnotesize $\Sigma_{2}$};
                 \node[below] at (-1.8,-0.4) {\footnotesize $\overline{\Sigma}_{3}$};
                 \node[below] at (1.6,-0.7) {\footnotesize $\overline{\Sigma}_{2}$};


                    \node[below] at (-5.4,1.6) {\footnotesize $\Sigma_{4}$};
                 \node[below] at (5.4,1.6) {\footnotesize $\Sigma_{1}$};
                 \node[below] at (-5.4,-1.0) {\footnotesize $\overline{\Sigma}_{4}$};
                 \node[below] at (5.4,-1.0) {\footnotesize $\overline{\Sigma}_{1}$};

                                 \draw [ ] (0.0, 0)--(-1.3,0.9);
                                    \draw [ ] (0.0, 0)--(1.3,0.9);
                                 \draw [ ] (0.0, 0)--(-1.3,-0.9);
                                  \draw [ ] (0.0, 0)--(1.3,-0.9);
                                 \draw[](-2.6,0)--(-1.3,0.9);
                                    \draw[](2.6,0)--(1.3,0.9);
                                  \draw[](-2.6,0)--(-1.3,-0.9);
                                  \draw[](2.6,0)--(1.3,-0.9);
                               \draw [ ] (4.6,0 )--(6,-1.5);
                                 \draw [ ] (4.6,0 )--(6,1.5);
                                 \draw[-latex](4.6,0 )--(5.25,-0.7);
                                       \draw[-latex](4.6,0 )--(5.25,0.7);
                                 \draw [] (-4.6,0)--(-6,-1.5);
                                 \draw [ ] (-4.6,0 )--(-6,1.5);
                                 \draw[-latex](-6.0,-1.5 )--(-5.25,-0.7);
                                       \draw[-latex](-6.0,1.5 )--(-5.25,0.7);

                                       \draw[](-1.3,0)--(-1.3,0.9);
                                             \draw[](1.3,0)--(1.3,0.9);
                                              \draw[-latex](-1.3,0)--(-1.3,0.45);
                                             \draw[-latex](1.3,0)--(1.3,0.45);

                                    \draw[](-1.3,0)--(-1.3,-0.9);
                                             \draw[](1.3,0)--(1.3,-0.9);
                                              \draw[-latex](-1.3,0)--(-1.3,-0.45);
                                             \draw[-latex](1.3,0)--(1.3,-0.45);

                            \draw [-latex] (-2.6,0)--(-1.95,0.45);
                              \draw [-latex] (-1.3,0.9)--(-0.65,0.45);
                                  \draw [-latex] (0,0)--(0.65,0.45);
                              \draw [-latex] (1.3,0.9)--(1.95,0.45);
                             \draw [-latex] (-2.6,0)--(-1.95,-0.45);
                             \draw [-latex] (-1.3,-0.9)--(-0.65,-0.45);
                                \draw [-latex] (0,0)--(0.65,-0.45);
                              \draw [-latex] (1.3,-0.9)--(1.95,-0.45);

\filldraw[white](-5.05,0)--(-1.85,0) arc (0:360:1.6);
\filldraw[white](1.85,0)--(5.05,0) arc (0:-360:1.6);
\draw[blue] (C) circle [radius=1.6];
\draw[blue] (D) circle [radius=1.6];
\draw[blue,-latex](3.3,1.59)--(3.35,1.6);
\draw[blue,-latex](-3.35,1.6)--(-3.3,1.6);
  \end{tikzpicture}
            \caption{ \footnotesize{The jump contour $\Sigma^{err}$. }}
      \label{jumpe}
        \end{center}
    \end{figure}

\begin{proposition}\label{ve-i}
The  estimate of $V^{err}(z)-I$ results from the following
            \begin{equation*}
                \|V^{err}(z)-I\|_{L^p(\Sigma^{err})}=\begin{cases}
                    \mathcal{O}(e^{-c t}),\quad z\in \Sigma^{err} \backslash\{\mathcal{U}_{R}\cup\mathcal{U}_{L}\}, \\
                    \mathcal{O}(t^{-1/(3p)+(1-p)\varsigma/p}), \quad z\in \{\partial \mathcal{U}_{R}\cup\partial \mathcal{U}_{L}\}.
                \end{cases}
            \end{equation*}
            where $c:=c(\xi,\tilde{\varphi})>0$ is a constant and $2\leq p\leq\infty$.
\end{proposition}
\begin{proof}
For $z\in\Sigma^{err} \backslash\{\mathcal{U}_{R}\cup\mathcal{U}_{L}\}$, according to (\ref{ioei}) and Proposition \ref{rezpm}, we have
\begin{equation*}
 |V^{err}(z)-I|= |V^{(4)}(z)-I|\lesssim e^{-ct}.
\end{equation*}
For $z\in \{\partial \mathcal{U}_{R}\}$, according to (\ref{ioei}), we have
\begin{equation*}
 |V^{err}(z)-I|= |M_{R}^{-1}(z)-I|\lesssim t^{-\varsigma}.
\end{equation*}
The estimate for $z\in \{\partial \mathcal{U}_{L}\}$ can be derived in a similar way.
Then we have the estimate as $2\leq p\leq\infty$.
\end{proof}

Define a Cauchy integral operator
        \begin{align*}
           \mathcal{C}_{w^{err}}[f]=\mathcal{C}_-\left( f \left( V^{err}(z)-I \right) \right),
        \end{align*}
         where $w^{err} =V^{err}(z)-I$ and  $\mathcal{C}_-$ is the Cauchy projection operator on $\Sigma^{err}$.
       From Proposition \ref{ve-i}, we have
       $$\|\mathcal{C}_{w^{err}}\|_{L^2(\Sigma^{err})}\lesssim\|\mathcal{C}_-\|_{L^2(\Sigma^{err})}\|V^{err}-I\|_{L^{\infty}(\Sigma^{err})}=\mathcal{O}(t^{-\varsigma}).$$
 According to the  Beals-Coifman theorem,   the solution of the  RH problem \ref{iew} can be expressed as
         \begin{equation}
            M^{err}(z)=I +\frac{1}{2\pi i} \int_{\Sigma^{err}} \frac{\mu^{err}(\zeta) \left( V^{err}(\zeta)-I\right)}{\zeta-z}\, \mathrm{d}\zeta, \nonumber
         \end{equation}
        where $\mu^{err} \in L^2\left( \Sigma^{err}\right)$   satisfies $\left(I-\mathcal{C}_{w^{err}}\right)\mu^{err}=I$.
   From Proposition \ref{ve-i}, we have estimates
       \begin{equation}
          \|\mu^{err}-I\|_{L^2(\Sigma^{err})}=\mathcal{O}(t^{-1/6-\varsigma/2}), \quad \|V^{err}(z)-I\|_{L^2(\Sigma^{err})}=\mathcal{O}(t^{-1/6-\varsigma/2}), \label{owej}
       \end{equation}
which imply that the  RH problem \ref{iew} exists an unique solution.
 Moreover, we perform the expansion of  $M^{err}(z)$ at $z=\infty$
        \begin{equation}
            M^{err}(z)=I +\frac{M^{err}_1}{z} +\mathcal{O}\left(z^{-2}\right), \label{trans8}
        \end{equation}
      where
      $$ M^{err}_1=-\frac{1}{2\pi i} \int_{\Sigma^{err}} \mu^{err}(\zeta)\left( V^{err}(\zeta)-I \right)\, \mathrm{d}\zeta.$$
The coefficient $M^{err}_1$ in (\ref{trans8}) and $M^{err}(z)|_{z=0}$ admit the following estimates.
   \begin{proposition} For $t\to\infty$, $M^{err}_1$ in (\ref{trans8}) and $M^{err}(z)$ as $z=0$ can be estimated as follows
	\begin{align}
		&  M^{err}_1=    \tau^{-1/3} \hat{M}_{R}(\hat{k},s)+  \mathcal{O}(t^{-1/3-\varsigma}), \label{e001}\\[6pt]
		& M^{err}(0) = I+\tau^{-1/3} \hat{M}_{R}(\hat{k},s)+ \mathcal{O}(t^{-1/3-\varsigma}).\label{e00}
	\end{align}
\end{proposition}
\begin{proof}
By   (\ref{ioei}) and  (\ref{owej}), we obtain that
	\begin{align}
		M^{err}_1  ={}&-\frac{1}{2\pi i} \oint_{\partial \mathcal{U}_{R} \cup\partial \mathcal{U}_{L}}\left( V^{err}(\zeta)-I \right)\, \mathrm{d}\zeta -\frac{1}{2\pi i} \int_{ \Sigma^{err}\backslash\{\partial \mathcal{U}_{R}\cup\partial \mathcal{U}_{L}\}} \left( V^{err}(\zeta)-I \right) \mathrm{d}\zeta\notag\\
&-\frac{1}{2\pi i} \int_{\Sigma^{err}}\left(\mu^{err}(\zeta)-I\right)\left( V^{err}(\zeta)-I \right)\, \mathrm{d}\zeta\notag\\
				 = {}&\tau^{-1/3}\left(M_{R,1}(s)+M_{L,1}(s)\right)+ \mathcal{O}(t^{-1/3-\varsigma}).\label{55-1}
	\end{align}
Recalling (\ref{mr}) and (\ref{ml}), (\ref{e001}) can be derived from (\ref{55-1}).
	In a similar way, we have
	\begin{align}
		M^{err}(0) &= I + \frac{1}{2\pi i} \oint_{\partial \mathcal{U}_{R}\cup\partial \mathcal{U}_{L}} \frac{ V^{err}(\zeta)-I}{\zeta} \, \mathrm{d}\zeta + \mathcal{O}(t^{-1/3-\varsigma})\nonumber\\
		&= I -\tau^{-1/3}\left(M_{R,1}(s)-M_{L,1}(s)\right) + \mathcal{O}(t^{-1/3-\varsigma}).\label{55-2}
	\end{align}
Then (\ref{e00}) can be derived from  (\ref{55-2}).
\end{proof}

\subsection{Pure $\bar{\partial}$-problem}\label{modi4}
\hspace*{\parindent}
In this subsection, we   consider the contribution of  a pure $\bar{\partial}$-problem $M^{(5)}(z)$.
   From (\ref{rhdbar}), we have
            \begin{equation}\label{transd}
                M^{(5)}(z)=M^{(4)}(z)\left(M^{rhp}(z)\right)^{-1},
            \end{equation}
            which  satisfies the following pure $\bar{\partial}$-problem.

\begin{prob3}\label{trad}
 Find  $M^{(5)}(z):=M^{(5)}(z,x,t)$  which satisfies
            \begin{itemize}
                \item   $M^{(5)}(z)$ is continuous in $\mathbb{C}$ and has sectionally continuous first partial derivatives in $\mathbb{C}  \backslash \left(\mathbb{R} \cup \Sigma^{(4)}\right)$.

                \item For $z\in \mathbb{C}$,  $M^{(5)}(z)$ satisfies the $ \bar{\partial}$-equation
                \begin{equation*}
                    \bar{\partial} M^{(5)}(z) = M^{(5)}(z) W^{(5)}(z),
                \end{equation*}
                where
                \begin{equation*}
                W^{(5)}(z):=M^{rhp}(z)  \bar{\partial} R^{(3)}(z) \left(M^{rhp}(z)\right)^{-1},
                \end{equation*}
                 and  $ \bar{\partial}  R^{(3)}(z)$ has been given in (\ref{R3}).
                  \item    $M^{(5)}(z)=I+\mathcal{O}(z^{-1}) ~as\ z\to  \infty$.	
            \end{itemize}

\end{prob3}
The solution of the $ \bar{\partial}$-RH problem   \ref{trad} can be given by
\begin{equation} \label{Im3}
	M^{(5)}(z)=I-\frac{1}{\pi}  \iint_\mathbb{C} \frac{ M^{(5)}(\zeta) W^{(5)}(\zeta)}{\zeta-z} \, \mathrm{d}\zeta\wedge\mathrm{d}\bar{\zeta},
\end{equation}
where $\mathrm{d}\zeta\wedge\mathrm{d}\bar{\zeta}$ is the Lebesgue measure and
$(\ref{Im3})$ can be  written as an operator equation
\begin{equation}
	(I-\mathcal{S}) M^{(5)}(z)=I, \label{Sm3}
\end{equation}
where the Cauchy operator $\mathcal{S}$ in (\ref{Sm3}) is given by
\begin{equation}
	\mathcal{S}[f](z):=\frac{1}{\pi} \iint \frac{f(\zeta)W^{(5)}(\zeta)}{\zeta-z} \, \mathrm{d}\zeta\wedge\mathrm{d}\bar{\zeta}.\label{opes}
\end{equation}

\begin{proposition}\label{pss}
	The  Cauchy operator $\mathcal{S}[f](z)$ in (\ref{opes}) is a small norm and admits the estimate
	$$\|\mathcal{S}\|_{L^\infty\to L^\infty}\lesssim  t^{-1/6}, \ \  t\to \infty,$$
	which implies the existence of $(I-\mathcal{S})^{-1}$ for a large $t$.
\end{proposition}
\begin{proof}
	For  $\forall f \in L^{\infty}(\Omega_1)$, from Proposition \ref{rezpm} and \ref{prop3}, we have
\begin{align*}
		\|\mathcal{S}[f]\|_{L^\infty(\mathbb{C})} &\le \|f\|_{L^\infty(\mathbb{C})} \frac{1}{\pi } \iint_{\Omega_1} \frac{|\bar{\partial}R_1e^{2it\theta}(\zeta)|}{|\zeta-z|} \mathrm{d}\zeta\wedge\mathrm{d}\bar{\zeta}\\
&\leq c\left(I_1+I_2+I_3+I_4\right),
\end{align*}	
		where
\begin{subequations}
	\begin{align}
		&I_1  =   \iint_{\Omega_{1}\cup\{|z|\leq\sqrt2\}} \frac{\left|r'(|\zeta|) \right|e^{\re(2it\theta)}}{|\zeta-z|} \mathrm{d}\zeta\wedge\mathrm{d}\bar{\zeta},
\ \
I_2  =   \iint_{\Omega_{1}\cup\{|z|>\sqrt2\}} \frac{\left|r'(|\zeta|) \right|e^{\re(2it\theta)}}{|\zeta-z|} \mathrm{d}\zeta\wedge\mathrm{d}\bar{\zeta},\nonumber\\
		&I_3  =   \iint_{\Omega_{1}\cup\{|z|\leq\sqrt2\}} \frac{\left| \zeta-z_1 \right|^{-1/2}e^{\re(2it\theta)}}{|\zeta-z|} \mathrm{d}\zeta\wedge\mathrm{d}\bar{\zeta},\ \
I_4  =   \iint_{\Omega_{1}\cup\{|z|>\sqrt2\}} \frac{\left| \zeta-z_1 \right|^{-1/2}e^{\re(2it\theta)}}{|\zeta-z|} \mathrm{d}\zeta\wedge\mathrm{d}\bar{\zeta}.\nonumber
	\end{align}\label{i1234}
\end{subequations}
Setting $z=x+iy$, $\zeta =z_1 + u+iv$,
   using the Cauchy-Schwartz's inequality, we have
	\begin{align*}
		|I_1| = \int_{0}^{\sqrt{2}\sin{\omega}} \int_{v}^{\sqrt{2}\cos{\omega}-z_1}  \frac{\left|r'(|\zeta|) \right|e^{\re(2it\theta)}}{|\zeta-z|} \mathrm{d}u \mathrm{d}v 	
&\lesssim t^{-1/6}.
	\end{align*}
\begin{align*}
		|I_2| = \int_{\sqrt{2}\sin{\omega}}^\infty \int_{\sqrt{2}\cos{\omega}-z_1}^{\infty}  \frac{\left|r'(|\zeta|) \right|e^{\re(2it\theta)}}{|\zeta-z|} \mathrm{d}u \mathrm{d}v 	
&\lesssim t^{-1/2}.
	\end{align*}
	Using the H\"{o}lder's inequality with $p>2$ and $1/p+1/q=1$, we have
	\begin{align*}
		&|I_3| \lesssim  \int_{0}^{\sqrt2\sin{\omega}} v^{1/p-1/2}\|v-y\|^{1/q-1}e^{-2tv^3} \mathrm{d}v \lesssim t^{-1/6},\\
		&|I_4| \lesssim  \int_{\sqrt2\sin{\omega}}^{\infty} v^{1/p-1/2}\|v-y\|^{1/q-1}e^{-2(\sqrt{2}-1)tv} \mathrm{d}v \lesssim t^{-1/2}.
	\end{align*}
\end{proof}

Proposition \ref{pss} implies that   the operator  equation (\ref{Sm3})  exists an unique solution. For $z\to\infty$, (\ref{Im3})
can be expanded as the form
\begin{equation}
	M^{(5)}(z)=I + \frac{M^{(5)}_1(x,t) }{z} +\mathcal{O}(z^{-2}), \quad z\to \infty, \label{trans9}
\end{equation}
where
\begin{align}\label{expanm51}
M^{(5)}_1(x,t)=\frac{1}{\pi} \iint_{\mathbb{C}} M^{(5)}(\zeta)W^{(5)}(\zeta)\, \mathrm{d}\zeta\wedge\mathrm{d}\bar{\zeta}.
\end{align}
For $z=0$, (\ref{Im3}) with the form
\begin{equation} \label{m50}
	M^{(5)}(0)=I-\frac{1}{\pi}  \iint_\mathbb{C} \frac{ M^{(5)}(\zeta) W^{(5)}(\zeta)}{\zeta} \, \mathrm{d}\zeta\wedge\mathrm{d}\bar{\zeta}.
\end{equation}

\begin{proposition} In the transition region $\mathcal{D}$, the coefficient $M^{(5)}_1(x,t)$ in (\ref{trans9}) and $M^{(5)}(z)|_{z=0}-I$ admit the  following estimates
	\begin{equation}\label{m51infty}
		|M^{(5)}_1(x,t)| \lesssim t^{-1/2}, \ \ |M^{(5)}(0)-I| \lesssim t^{-1/2}, \ t\to\infty.
	\end{equation}
\end{proposition}
\begin{proof}
	Similar to the proof of Proposition \ref{pss}, we take $z \in \Omega_{1}$ as an example and divide the integration (\ref{expanm51}) on $\Omega_{1}$ into four parts.
First, we consider the estimation of  $M^{(5)}_1(x,t)$.
By (\ref{transd}) and the boundedness of $M^{(4)}(z)$ and $M^{rhp}(z)$ on $\Omega_{1}$, we have
	\begin{align}\label{m51esti}
	|M^{(5)}_1(x,t)| \lesssim I_1 + I_2+I_3 + I_4,
	\end{align}
	where
\begin{subequations}
	\begin{align}
		&I_1  =   \iint_{\Omega_{1}\cup\{|z|\leq\sqrt2\}} \left|r'(|\zeta|) \right|e^{\re(2it\theta)} \mathrm{d}\zeta\wedge\mathrm{d}\bar{\zeta},
\ \
I_2  =   \iint_{\Omega_{1}\cup\{|z|>\sqrt2\}} \left|r'(|\zeta|) \right|e^{\re(2it\theta)} \mathrm{d}\zeta\wedge\mathrm{d}\bar{\zeta},\nonumber\\
		&I_3  =   \iint_{\Omega_{1}\cup\{|z|\leq\sqrt2\}} \left| \zeta-z_1 \right|^{-1/2}e^{\re(2it\theta)} \mathrm{d}\zeta\wedge\mathrm{d}\bar{\zeta},\ \
I_4  =   \iint_{\Omega_{1}\cup\{|z|>\sqrt2\}} \left| \zeta-z_1 \right|^{-1/2}e^{\re(2it\theta)} \mathrm{d}\zeta\wedge\mathrm{d}\bar{\zeta}.\nonumber
	\end{align}
\end{subequations}
 Let  $\zeta =z_1 + u+iv$,  draw support from Cauchy-Schwartz's inequality, we have
		\begin{align*}
		&|I_1| = \int_{0}^{\sqrt{2}\sin{\omega}} \int_{v}^{\sqrt{2}\cos{\omega}-z_1}  \left|r'(|\zeta|) \right|e^{\re(2it\theta)} \mathrm{d}u \mathrm{d}v
\lesssim t^{-1/2},\\
		&|I_2| = \int_{\sqrt{2}\sin{\omega}}^\infty \int_{\sqrt{2}\cos{\omega}-z_1}^{\infty}  \left|r'(|\zeta|) \right|e^{\re(2it\theta)} \mathrm{d}u \mathrm{d}v
\lesssim t^{-3/2}.
	\end{align*}
	By H\"{o}lder's inequality with $p>2$ and $1/p+1/q=1$, we have
	\begin{align*}
		&|I_3|  \lesssim\int_0^{\sqrt{2}}\int_0^{\sqrt{2}\cos{\omega}-z_1}|u+iv|^{-\frac{1}{2}}e^{-2tv^3}\mathrm{d}u\mathrm{d}v\lesssim t^{-1/2},\\
&|I_4|  \lesssim\int^\infty_{\sqrt{2}}\int^\infty_{\sqrt{2}\cos{\omega}-z_1}|u+iv|^{-\frac{1}{2}}e^{-2(\sqrt{2}-1)tv}\mathrm{d}u\mathrm{d}v\lesssim t^{-3/2}.\\
	\end{align*}
For $z \in \Omega_{1}$, we have $|z|^{-1} \le |z_1|^{-1}$. According to (\ref{m50}), we have
\begin{equation*}
 |M^{(5)}(0)-I| \lesssim \iint_{\Omega_{1}} \left|\bar{\partial}R_{1}(\zeta)e^{2it\theta}\right| \mathrm{d}\zeta\wedge\mathrm{d}\bar{\zeta}.
 \end{equation*}
 By the similar estimate with (\ref{m51esti}),
we have $|M^{(5)}(0)-I| \lesssim t^{-1/2}$.
\end{proof}

 To derive the potential function $q(x,t)$ from the reconstruction formula (\ref{sol}), it is necessary to recover $M^{(3)}(0)$ by the following proposition.
\begin{proposition}For $t\to\infty$, $M^{(3)}(z)|_{z=0}$ admits a  estimate
	\begin{align}\label{m30}
		M^{(3)}(0)= M^{err}(0) + \mathcal{O}(t^{-1/2}),
	\end{align}
where $M^{err}(0)$ is given by (\ref{e00}).
\end{proposition}
\begin{proof}
	Reviewing the series of transformations  (\ref{trans4}), (\ref{trans6}) and (\ref{transd}),  for a large $z$  and $R^{(3)}(z)=I$,  the solution of $M^{(3)}(z)$ is given by
	\begin{align*}
		M^{(3)}(z) = M^{(5)}(z)M^{err}(z) .
	\end{align*}
From  (\ref{trans9}) and (\ref{m51infty}), we derive
\begin{align*}
	M^{(3)}(z)= M^{err}(z)  + \mathcal{O} (t^{-1/2}),
\end{align*}
which leads to (\ref{m30})  if we take $z=0$.
\end{proof}

\subsection{Painlev\'{e} Asymptotic}\label{recover}
\hspace*{\parindent}
In this subsection, we give the  proof of Theorem \ref{th}.
\begin{proof}
	
	Inverting  the  transformations (\ref{trans1}),  (\ref{trans3}), (\ref{trans4}),  (\ref{trans6}) and (\ref{transd}) in Section \ref{sec3}. As $z\to \infty$, the matrixes  $R^{(3)}(z)= G(z)=I$. The  solution of RH problem \ref{RHP0} is shown as
	\begin{align}
M(z) 		= T^{\sigma_3}(\infty)J(z)  M^{(5)}(z) M^{err}(z)T^{-\sigma_3}(z) + \mathcal{O}(e^{-ct}). \label{rem}
	\end{align}
Recalling the expansion of $T(z)$ shown in (\ref{Texpan}), we substitute (\ref{trans8}) and   (\ref{trans9}) into (\ref{rem}) and obtain
\begin{align*}
M(z)=T^{\widehat{\sigma}_3}(\infty)\left[ I+ \frac{1}{z}\left( \sigma_2 M^{(3)}(0)^{-1} +M_1^{(5)}+M_1^{err}-iT_1^{-\sigma_3}\right) \right ] + \mathcal{O}(e^{-ct}).
	\end{align*}
According to the reconstruction formula (\ref{sol})
	\begin{align*}
		q(x,t)
 =T^2(\infty)\left(-1-i\tau^{-1/3}\left(\beta_1+\beta_2\right)\right) + \mathcal{O}\left( t^{-\frac{1}{3}-\varsigma}\right),
	\end{align*}
which leads to the result in Theorem \ref{th}.
\end{proof}

\appendix

\section{Painlev\'{e}-II RH model} \label{appx}
\hspace*{\parindent}
The well-known Painlev\'{e}-II equation is given by
\begin{align}
	u_{ss}(s) = 2u^3(s) +su(s), \quad s \in \mathbb{R},\label{p2}
\end{align}
which can be solved using the RH problem \ref{painrh} as follows \cite{fas1983,its1986}.
\begin{prob}\label{painrh}
	Find   $M_{P}(k)=M_{P}(k,s,t)$ with properties
	\begin{itemize}
		\item  $M_{P}(k)$ is analytical in $\mathbb{C}\setminus  \Sigma_{P}$.
		\item $M_{P}(k)$ satisfies the jump condition
		\begin{equation*}
			M_{P,+}( k)=M_{P,-}(k)V_{P}(k),
		\end{equation*}
		where
\begin{equation*}		
V_{P}(k)=\begin{cases}
\left(\begin{array}{cc}
1&0\\
pe^{2i(\frac{4}{3}k^3+sk)}&1
\end{array}
\right),\ \ k\in\Sigma^{1}_{P},\\
\left(\begin{array}{cc}
1&re^{-2i(\frac{4}{3}k^3+sk)}\\
0&1
\end{array}
\right),\ \ k\in\Sigma^{2}_{P},\\
\left(\begin{array}{cc}
1&0\\
qe^{2i(\frac{4}{3}k^3+sk)}&1
\end{array}
\right),\ \ k\in\Sigma^{3}_{P},\\
\left(\begin{array}{cc}
1&-pe^{-2i(\frac{4}{3}k^3+sk)}\\
0&1
\end{array}
\right),\ \ k\in\Sigma^{4}_{P},\\
\left(\begin{array}{cc}
1&0\\
-re^{2i(\frac{4}{3}k^3+sk)}&1
\end{array}
\right),\ \ k\in\Sigma^{5}_{P},\\
\left(\begin{array}{cc}
1&-qe^{-2i(\frac{4}{3}k^3+sk)}\\
0&1
\end{array}
\right),\ \ k\in\Sigma^{6}_{P},\\
\end{cases}
\end{equation*}
where $\Sigma^j_P =  \mathbb{R}^+e^{i\left[\frac{\pi}{6}+\frac{(j-1)\pi}{3}\right]}.$ Let $\Sigma_P$ denote the oriented contour consisting of six rays, $$\Sigma_P = \bigcup_{j=1}^6  \Sigma^j_P,$$  see Figure \ref{Sixrays}.
The parameters $p, q$ and $r$ in $V_{P}(k)$  are complex values satisfying the relation
		\begin{align*}
		r=	p+q+pqr.
		\end{align*}
		\item $M_{P}(k)$  satisfies the following asymptotic behaviors
		\begin{align*}
			&M_{P}( k)=I+\mathcal{O}(k ^{-1}),	\quad k \to  \infty,\\
			& M_{P}( k) = \mathcal{O}(1),\quad k \to 0.
		\end{align*}
		
	\end{itemize}
\end{prob}

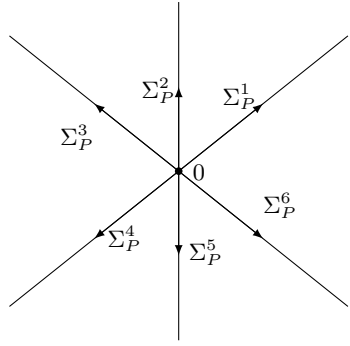
\begin{figure}[H]
	\begin{center}
		\begin{tikzpicture}[scale=0.9]

			\node[shape=circle,fill=black,scale=0.15] at (0,0) {0};
			\node[below] at (0.3,0.25) {\footnotesize $0$};
			\draw [] (0,-2.5 )--(0,2.5);
			\draw [-latex] (0,0)--(0,1.25);
			\draw [-latex] (0,0 )--(0,-1.25);
			\draw [ ] (0,0 )--(2.5,2);
			\draw [-latex] (0,0)--(1.25,1);
			\draw [] (0,0 )--(2.5,-2);
			\draw [-latex] (0,0)--(1.25,-1);
			\draw [] (0,0 )--(-2.5,2);
			\draw [-latex] (0,0)--(-1.25,1);
			\draw [] (0,0 )--(-2.5,-2);
			\draw [-latex] (0,0)--(-1.25,-1);

			\node at (0.9,1.1) {\footnotesize$\Sigma_P^1$};
			\node at (1.5,-0.5) {\footnotesize$\Sigma_P^6 $};
			\node at (-1.5,0.5) {\footnotesize$\Sigma_P^3$};
			\node at (-0.8,-1) {\footnotesize$\Sigma_P^4$};
			\node at (-0.3,1.2) {\footnotesize$\Sigma_P^2$};
			\node at ( 0.4,-1.2) {\footnotesize$\Sigma_P^5$};
			
			
		\end{tikzpicture}
		\caption{ \footnotesize { The jump contour $\Sigma_P$.}}
		\label{Sixrays}
	\end{center}
\end{figure}
Then
\begin{align}
	u(s) = 2\left(M_{P,1}(s)\right)_{12} = 2 \left(M_{P,1}(s)\right)_{21},\label{up2}
\end{align}
solves the Painlev\'{e}-II equation (\ref{p2}), where
\begin{align*}
	M_{P}(k) = I + \frac{M_{P,1}(s)}{k} + \mathcal{O} \left(k^{-2}\right), \quad k \to \infty.
\end{align*}
For  $\forall q\in i\mathbb{R}$, $|q|<1, p=-q, r=0$, the formula (\ref{up2}) has a global, real solution $u(s)$ of Painlev\'{e}-II equation with the asymptotic behavior
\begin{align*}
	u(s) = a \mathrm{Ai}(s) +\mathcal{O}\left( e^{-(4/3)s^{3/2}}s^{-1/4}\right),\quad s\to +\infty,
\end{align*}
where $a =- \im q$ and $\mathrm{Ai}(s)$ denotes the  Airy function.
Moreover, the subleading coefficient $M_{P,1}(s)$ is given by
\begin{align}\label{posee}
M_{P,1}(s) = \frac{1}{2} \begin{pmatrix} -i\int_s^\infty u(\zeta)^2\mathrm{d}\zeta & u(s) \\ u(s) & i\int_s^\infty u(\zeta)^2\mathrm{d}\zeta \end{pmatrix},
\end{align}
and for each $c > 0$,
\begin{align}
	\sup_{k \in \mathbb{C}\setminus \Sigma_P} \sup_{s \geq -c} |M_P(k,s)|  < \infty.\label{mPbounded}
\end{align}
\vspace{6mm}

 \noindent\textbf{Acknowledgements}

 This work is supported by  the National Natural Science Foundation of China (Grant No. 11671095, 51879045).\vspace{2mm}

    \noindent\textbf{Data Availability Statements}

    The data that supports the findings of this study are available within the article.\vspace{2mm}

    \noindent{\bf Conflict of Interest}

    The authors have no conflicts to disclose.

\end{document}